\documentclass[letterpaper]{article} %
\usepackage{aaai2026}  %
\usepackage{times}  %
\usepackage{helvet}  %
\usepackage{courier}  %
\usepackage[hyphens]{url}  %
\usepackage{graphicx} %
\usepackage{natbib}  %
\usepackage{caption} %
\usepackage{algorithm}

\usepackage{newfloat}
\usepackage{listings}
\DeclareCaptionStyle{ruled}{labelfont=normalfont,labelsep=colon,strut=off} %
\floatstyle{ruled}
\newfloat{listing}{tb}{lst}{}
\floatname{listing}{Listing}
\usepackage{soul}
\usepackage{url}
\usepackage[utf8]{inputenc}
\usepackage{booktabs}
\usepackage{algorithm}
\usepackage{algpseudocodex}
\usepackage[switch]{lineno}

\usepackage{amsfonts, amsmath, amssymb, amsthm}
\usepackage{mathtools}
\usepackage{nicefrac}
\usepackage{paralist}

\usepackage[shortlabels,inline]{enumitem}

\usepackage{microtype}

\usepackage[nameinlink]{cleveref}
\crefname{problem}{problem}{problems} 
\Crefname{problem}{Problem}{Problems} 
\Crefname{assumption}{Assumption}{Assumptions} 

\usepackage{dsfont}

\usepackage[appendix=append, forwardlinking=yes, bibliography=common]{apxproof} 

\newtheoremrep{theorem}{Theorem}[section]
\newtheoremrep{lemma}[theorem]{Lemma}
\newtheoremrep{proposition}{Proposition}[section]
\newtheoremrep{corollary}{Corollary}[theorem]
\newtheoremrep{assumption}{Assumption}

\theoremstyle{definition}

\newtheorem*{definition*}{Definition}

\newtheorem*{theorem*}{Theorem}

\newtheorem*{problem*}{Problem}

\newcommand{\eps}{\varepsilon}

\newcommand{\R}[0]{\mathbb{R}}

\DeclarePairedDelimiter\abs{\lvert}{\rvert}%
\DeclarePairedDelimiter\norm{\lVert}{\rVert}%

\newcommand{\poly}[0]{$\mathsf{P}$}
\newcommand{\NP}[0]{$\mathsf{NP}$}
\newcommand{\NPH}[0]{\NP-hard}
\newcommand{\NPC}[0]{\NP-complete}

\newcommand{\GI}[0]{$\mathsf{GI}$}

\newcommand{\SW}[0]{\mathit{SW}}

\newcommand{\game}[0]{G}

\renewcommand{\u}[0]{\textbf{u}}
\renewcommand{\a}[0]{a}

\newcommand{\C}[0]{\Delta}

\newcommand{\F}[0]{\mathcal{F}}

\newcommand{\red}[1]{\bar{#1}} 

\newcommand{\SPIs}[0]{SPIs}

\newcommand{\remap}[0]{\Psi}
\newcommand{\remapu}[0]{\hat \Psi}
\newcommand{\policy}[0]{\Pi}
\newcommand{\oc}[0]{\Phi}

\newcommand{\iso}[0]{\phi}

\newcommand{\pgeq}[0]{\succeq}
\newcommand{\pge}[0]{\succ}

\newcommand{\tgame}[0]{\mathcal{T}}

\newcommand{\GPRname}[0]{\textsc{Entrywise Positive Affine Vector Remapping}}

\makeatletter
\let\oldabs\abs
\def\abs{\@ifstar{\oldabs}{\oldabs*}}
\let\oldnorm\norm
\def\norm{\@ifstar{\oldnorm}{\oldnorm*}}
\makeatother

\usepackage{ifthen}
\usepackage{xcolor}
\newboolean{commentsactivated}
\setboolean{commentsactivated}{false}

\newcommand{\cameraReady}[1]{{#1}}
\newboolean{cameraReadyCommentsOn}
\setboolean{cameraReadyCommentsOn}{false}
\newcommand{\crcomment}[1]{\ifthenelse{\boolean{cameraReadyCommentsOn}}{{\color{teal} {\em #1 }}}{}}

\newboolean{afterIJCAI}
\setboolean{afterIJCAI}{false}
\newcommand{\later}[1]{\ifthenelse{\boolean{afterIJCAI}}{{#1}}{}}

\newboolean{haveSpace}
\setboolean{haveSpace}{false}
\newcommand{\cut}[1]{\ifthenelse{\boolean{haveSpace}}{{#1}}{}}

\newcommand{\co}[1]{\ifthenelse{\boolean{commentsactivated}}{{\color{teal} {\em CO: #1 }}}{}}
\newcommand{\ns}[1]{\ifthenelse{\boolean{commentsactivated}}{{\color{red} {\em NS: #1 }}}{}}

\newcommand{\blue}[1]{\ifthenelse{\boolean{commentsactivated}}{{\color{blue} {\em Note: #1 }}}{}}
\newcommand{\note}[1]{\ifthenelse{\boolean{commentsactivated}}{{\color{blue} {\em Note: #1 }}}{}}

\newcommand{\todoNS}[1]
{\ifthenelse{\boolean{commentsactivated}}{{\color{violet} {\em TODO NS: #1 }}}{}}
\newcommand{\todoCO}[1]
{\ifthenelse{\boolean{commentsactivated}}{{\color{orange} {\em TODO CO: #1 }}}{}}

\newcommand{\remove}[1]{}

\usepackage{verbatim}
\usepackage{csquotes}

\usepackage{multirow}

\usepackage{array}
\title{Promises Made, Promises Kept:\\ Safe Pareto Improvements via Ex Post Verifiable Commitments}
\author {
    Nathaniel Sauerberg\equalcontrib \textsuperscript{\rm 1},
    Caspar Oesterheld\equalcontrib \textsuperscript{\rm 2}
    }
\affiliations {
    \textsuperscript{\rm 1}University of Texas at Austin\\
    \textsuperscript{\rm 2}Foundations of Cooperative AI Lab, Carnegie Mellon University\\
    nsauerberg@utexas.edu, oesterheld@cmu.edu
}

\begin{document}

\maketitle

\begin{abstract}
A safe Pareto improvement (SPI) \cite{SPI-original}
is a modification of a game that leaves all players better off with certainty. 
SPIs are typically proven under qualitative assumptions about the way different games are played. 
For example, we assume that strictly dominated strategies can be iteratively removed and that isomorphic games are played isomorphically.
In this work, we study SPIs achieved through three types of \textit{ex post} verifiable commitments -- promises about player behavior from which deviations can be detected by observing the game. 
First, we consider disarmament -- commitments not to play certain actions. 
Next, we consider SPIs based on \textit{token games}. A token game is a game played by simply announcing an action (via cheap talk). As such, its outcome is intrinsically meaningless. However, we assume the players commit in advance to play specific (pure or correlated) strategy profiles in the original game as a function of the token game outcome. Under such commitments, the token game becomes a new, meaningful normal-form game.
Finally, we consider default-conditional commitment: SPIs in settings where the players' default ways of playing the original game can be credibly revealed and hence the players can commit to act as a function of this default. 
We characterize the complexity of deciding whether SPIs exist in all three settings, giving a mixture of characterizations and efficient algorithms and \NP- and \textsc{Graph Isomorphism}-hardness results. 
\end{abstract}

\section{Introduction}

Among the most important applications of game theory is guiding decisions that shape a downstream strategic interaction. 
To make this tractable, it's common to reduce games to a single value by assuming that the games will resolve in a particular way. %
For example, the literature on Stackelberg games generally makes an explicit assumption that the followers will play the best 
or worst %
\cite[e.g.][]{coniglio2020computing-pessimistic-SE-multiple-followers-mixed-pure, sauerberg2024commitment-to-payments, basilico2017methods_multi-follower-SE}
Nash equilibrium for the leader, while mechanism design typically assumes that the truthful equilibrium will be played.
Similarly, notions like price of anarchy/stability \cite{koutsoupias1999worst_case_equilibria_price_of_anarchy}, value of mediation \cite{ashlagi2008value_of_correlation}, and value of recall \cite{berker2025value_of_recall} %
evaluate the importance of a particular affordance (centralized control, mediation, or recall)  by bounding the ratio of some welfare objective between the games with and without the affordance, assuming a particular (best- or worst-case) equilibrium will be played in each game.
Work on game theory with simulation asks when allowing one player to pay to learn the other’s strategy introduces a new Nash equilibrium which is Pareto-better than all existing equilibria \cite{kovarik2023-GT-pure-sim,kovarik2025game-GT-mixed-sim}. %

However, compressing complex games down to a single outcome is not without loss. 
In particular, equilibrium selection is an unsolved and arguably unsolvable problem \cite{norde1996equilibrium-selection-and-consistency}, so it may not be safe to assume that players play any particular equilibrium, or indeed play an equilibrium at all. %
Safe Pareto improvements (SPIs) \cite{SPI-original} offer a more general framework for analyzing interventions--one that doesn’t require assumptions about how individual games are resolved. %
SPIs are interventions that improve \textit{all} possible outcomes of the game, given explicit, usually mild assumptions on the relationships between the ways different games are played. 
For example, we might assume that isomorphic games are played isomorphically and that removing strictly dominated strategies doesn't change how a game is played. %
In this paper, we introduce and study SPIs via \textit{ex post} verifiable commitments--commitments that the other players or an outside observer can verify to have been followed by observing the game. 

\begin{table}[b]
\centering
\begin{tabular}{ccccc}
                        & FA                           & FN                           & PA                           & PN                           \\ \cline{2-5} 
\multicolumn{1}{c|}{FA} & \multicolumn{1}{c|}{$2, 2$}    & \multicolumn{1}{c|}{$-1, 8\phantom{-}$}   & \multicolumn{1}{c|}{$-10, 10\phantom{-}$} & \multicolumn{1}{c|}{$-10, 10\phantom{-}$} \\ \cline{2-5} 
\multicolumn{1}{c|}{FN} & \multicolumn{1}{c|}{$\phantom{-}8, -1$}   & \multicolumn{1}{c|}{$-100, -100$} & \multicolumn{1}{c|}{$-10, 10\phantom{-}$} & \multicolumn{1}{c|}{$-10, 10\phantom{-}$}  \\ \cline{2-5} 
\multicolumn{1}{c|}{PA} & \multicolumn{1}{c|}{$\phantom{-}10, -10$} & \multicolumn{1}{c|}{$\phantom{-}10, -10$} & \multicolumn{1}{c|}{$2, 2$}    & \multicolumn{1}{c|}{$0, 5$}    \\ \cline{2-5} 
\multicolumn{1}{c|}{PN} & \multicolumn{1}{c|}{$\phantom{-}10, -10$} & \multicolumn{1}{c|}{$\phantom{-}10, -10$}    & \multicolumn{1}{c|}{$5, 0$}    & \multicolumn{1}{c|}{$-10, -10$} \\ \cline{2-5} 
\end{tabular}%
\caption{Seaway Dispute. The actions stand for claim Full/Partial seaway via Navy/Announcement.}
\label{tab:extended-chicken-example-intro}

\centering
\begin{tabular}{lcc}
                              & \multicolumn{1}{l}{Token PA}  & \multicolumn{1}{l}{Token PN}    \\ \cline{2-3} 
\multicolumn{1}{l|}{Token PA} & \multicolumn{1}{c|}{$3.2, 3.2$} & \multicolumn{1}{c|}{$2, 5$}       \\ \cline{2-3} 
\multicolumn{1}{l|}{Token PN} & \multicolumn{1}{c|}{$5, 2$}     & \multicolumn{1}{c|}{$-4, -4$} \\ \cline{2-3} 
\end{tabular}%
\caption{Token Game for Seaway Dispute}
\label{tab:token-game-intro}
\end{table}

Consider the following strategic interaction, modeled by \Cref{tab:extended-chicken-example-intro}.
Two countries are disputing the control of a seaway passing that’s newly navigable due to melting sea ice. 
Part of the seaway lies within each country’s established territory, but most had previously been unclaimed sea ice. %
The countries choose both (1) whether to claim the full seaway ($\mathrm{F}$) or only the disputed part ($\mathrm{P}$), and (2) whether to assert their claim via naval occupation ($\mathrm{N}$) or diplomatic announcement ($\mathrm{A}$).
If either both countries claim the full seaway ($\mathrm{F}$) or both claim only the disputed portion ($\mathrm{P}$), the game has typical “chicken” dynamics.%
\begin{itemize}
    \item $(\mathrm{N}, \mathrm{A})$ or $(\mathrm{A}, \mathrm{N})$: If one country claims with their Navy and the other via Announcement, the country that claims with its Navy wins control of the disputed territory. 
    \item $(\mathrm{A}, \mathrm{A})$: If both countries claim via Announcement, they will eventually agree to joint control through diplomacy.
    \item $(\mathrm{N}, \mathrm{N})$: If both countries claim with their Navies, it results in costly warfare. This is especially costly if the countries claim the full seaway and so invade each other.
\end{itemize}
If only one country attempts to claim the other’s territory, the international outrage results in an outcome favorable to the other country. 
For this reason, $\mathrm{PA}$ strictly dominates $\mathrm{FA}$ and $\mathrm{FN}$. 
Therefore, we assume that by default the players will play $\mathrm{PN}$ or $\mathrm{PA}$, so the game is equivalent to its bottom right quadrant.  
There are increasing returns to controlling more of the seaway, so the sum of the payoffs for the different outcomes follows $\SW(\mathrm{FN}, \mathrm{FA}) > \SW(\mathrm{PN}, \mathrm{PA}) > \SW(\mathrm{PA}, \mathrm{PA}) = \SW(\mathrm{FA}, \mathrm{FA})$. 

Suppose both countries believe they’re likely to achieve a payoff of $5$ from $(\mathrm{PN}, \mathrm{PA})$ or $(\mathrm{PA}, \mathrm{PN})$ in the default game (or they want to represent that belief for strategic reasons).
No strategy profile gives both players a utility of at least $5$ (or even more than $3.5$), so no agreement to play a single strategy profile is possible.  
However, the players can still guarantee a Pareto improvement if they commit to resolve the game by playing the “token game” in \Cref{tab:token-game-intro}.
They act in the token game simply by (privately) writing their actions on pieces of paper and then flipping them over to observe the outcome. 
Once the token outcome $t$ is observed, the players are required to play a (predetermined) correlated strategy profile in the original game with expected utility $\u(t)$. 
\cameraReady{
Any jointly observable source of randomness, e.g. physical or cryptographic \cite{blum1983coin-fipping-by-phone} coin flips, can be used to draw an outcome from this correlated strategy profile.  
Given such a source, it's \textit{ex post} verifiable that the players play according to the prescribed outcome, so we assume they do.
}%

For this to work, of course, all payoffs in the token game must be realizable (in expectation) by some correlated strategy profile in the original game. 
Here, the $(5,2)$ payoff can be achieved by $(2/3)(\mathrm{FN}, \mathrm{FA}) + (1/3)(\mathrm{FA}, \mathrm{FN}) = (2/3)(8, -1) + (1/3)(-1, 8)$, and the $(2, 5)$ payoff symmetrically. 
The $(3.2, 3.2)$ payoff can be achieved by playing $.2(\mathrm{PA}, \mathrm{PA}) + .4(\mathrm{FN}, \mathrm{FA}) + .4(\mathrm{FA}, \mathrm{FN})$, and the $(-4, -4)$ payoff can be achieved by playing $(1/2)(\mathrm{PA}, \mathrm{PA}) + (1/2)(\mathrm{PN}, \mathrm{PN})$. 

Also, observe that the token game is isomorphic to the bottom right quadrant of the original game: a payoff of $v_i$ in the original game corresponds to a payoff of $.6v_i + 2$ in the token game. 
Therefore, we can reasonably assume that the players will play them isomorphically. 
E.g., if a player would have played $\mathrm{PA}$ in the original game, they'd play Token $\mathrm{PA}$ in the token game, and so on. 
Since each token outcome Pareto improves on its counterpart, the token game is a guaranteed Pareto improvement on the original game, an SPI.

In this paper, we consider SPIs achieved by three different types of \textit{ex post} verifiable commitment: the token game SPIs exemplified by the previous example and two others.
Situations where \textit{ex post} verifiable commitments can be made credible are frequent. 
In particular, they could be enforced by reputation costs or by external authorities through legal contracts. 
Moreover, \textit{ex post} verifiability seems close to necessary for any type of external enforcement of commitment.

Existing schemes for achieving SPIs, such as the delegation game SPIs proposed in \cite{SPI-original}, require forms of commitment that seem more difficult to achieve.  
In the delegation game setting, the original players delegate the game to representatives, assigning them a utility function but otherwise leaving how to play up to the representatives. 
By default, they instruct the representatives with their true utility functions, but SPIs can sometimes be achieved by the players instead making joint agreements to assign alternate utility functions. 

Making it credible that one’s representative indeed plays according to an alternate utility function seems to require a high degree of transparency into the representative's decision making process. 
Though such transparency is sometimes achievable, it seems unattainable in cases where the decisions are being made in the minds of the people or groups of people with stake in the outcome of the game.
Importantly, the SPIs achieved in the present paper don’t require any assumptions on the process that decides how to play the games. 
In particular, our schemes allow the original, self-interested players to simply play the modified game themselves.

\textbf{Contributions.}~~%
In \Cref{sec:disarmament-SPIs}, we consider SPIs that can be achieved by %
players committing not to take, i.e., \textit{disarm}, particular actions. We find that deciding the existence of disarmament SPIs is \NP-complete. We further find that deciding whether a \textit{given} disarmament %
is an SPI is polynomial-time equivalent to the graph isomorphism problem (which is believed to be \NP-intermediate, i.e., in \NP, but neither in \poly{} nor \NP-hard).

Next, in \Cref{sec:token-game-SPIs}, we study token game SPIs like the one for the Seaway Dispute described above. We distinguish two types of token SPIs. 
In the first type, the token outcomes can represent distributions over outcomes of the default game (as in the Seaway Dispute example). 
In this case, the existence of SPIs can be decided in polynomial time, and we obtain an explicit characterization of the existence of SPIs in the two-player case. %
In the second type, token outcomes can only represent a single outcome of the default game. 
There, we give an algorithm for finding SPIs that runs in polynomial time in games with a constant number of players and quasi-polynomial time in general. 
We also show that the problem becomes \NPC{} in more succinct game representations (e.g. payoff tables that only store non-zero entries).

Finally, in \Cref{sec:default-remapping-spis}, we study SPIs achieved by in a setting where players can credibly reveal their default. 
We show that finding unilateral default-remapping SPIs, where a player commits to act according to some function of their default strategy, is \NPH{}.
However, in the omnilateral default-remapping setting, where all players can credibly reveal their default action and jointly commit to play a strategy profile as a function of the default outcome, we show that SPIs exist whenever the original game contains Pareto-suboptimal outcomes after dominated strategies are iteratively removed.

\subsection{Related Work}
Most closely related to our paper is the prior work on safe Pareto improvements \cite{SPI-original,Oesterheld2025,digiovanni2024safe}. %
We use the %
framework from \citet{SPI-original,Oesterheld2025}%
, see our \Cref{sec:preliminaries}.
We consider different interventions on strategic interaction than \citeauthor{SPI-original} (see above). %
The token games studied in \Cref{sec:token-game-SPIs} resemble the token games studied by \citet[Sect.\ 5]{SPI-original}%
. The main difference is that they %
allow giving the agents arbitrary utility functions over the token outcomes. In contrast, we assume that the utility of a token outcome is simply the utility of the distribution over base game outcomes associated with the token outcome. Consequently, the results are different.
Our study of default-conditional SPIs (\Cref{subsec:multilateral-default-remapping-spis}) is inspired by \citet{digiovanni2024safe}, but they focus on the implementation of \cameraReady{and incentives for} commitment to default-remapping \cameraReady{in a program game setting} and don't study the unilateral default remapping case. 

Many forms of \textit{ex post} verifiable commitment (e.g., Stackelberg games) have been studied. %
Most closely related is the prior work on disarmament games (\citealt{deng2017disarmament,deng2018disarmament}; see also \citealt{Renou2009_commitment_games_simulatenous_disarmament}; \citealt{bade2009bilateral_commitment_cf_iterated_disarmament}; \citealt[][]{collina2024value-ambiguous-commitments-multi-follower-games}), which studies similar forms of commitment to our \Cref{sec:disarmament-SPIs}, though not in an SPI framework. %
\cameraReady{The literature on program games and open source game theory (\citealt{tennenholtz2004program_equilibrium, critch2022cooperative-Open-Source-GT}; cf.\ \citealt{kalai2010commitment_folk_thm}) can also be viewed as a studying form of \textit{ex post} verifiable commitment: players commit to play the output of a computer program which can read the other players’ programs. 
In addition to being the setting of \cite{digiovanni2024safe}, this bears some similarity to our token games: In both, the players commit to condition their final actions on a prior interaction (the token game or the program game).
}
We discuss more distantly related work on \textit{ex post} verifiability in \Cref{subsec:distantly_related_work}. %

\begin{toappendix}
\subsection{More Distantly Related Work}\label{subsec:distantly_related_work}
Stackelberg commitment -- commitment by a leader to play particular pure, mixed, or correlated strategy \cite{vonStengel2004_leadership_mixed-commitment, conitzer_optimal_strategy_to_commit, commitment-to-correlated-strategies} -- is \textit{ex post} verifiable by our definition. 
Commitment to outcome-conditional payments (\citealt{internal_implementation}; cf.\ \citealt{sauerberg2024commitment-to-payments, Gupta2015_rational-explanation-bribery}), is also \textit{ex post} verifiable. 
There is also work \cite[e.g.][]{grigoryan2023theory_auditability_of_SC_mechanisms} on whether groups of participants can detect \textit{ex post} that a mechanism was not faithfully executed
and on designing \textit{credible mechanisms} \cite{akbarpour2018credible-mechanisms, credible-auctions-via-crypto, credible-auctions-via-crypto-all-dists, credible-auctions-via-blockchain} -- mechanisms for which any \textit{profitable} deviation from the mechanism by the principal is detectable by a single agent. %

Finally, we discuss two papers with some conceptual similarity to our work. 
\citet{ijcai2024p333_mechanisms-play-game-not-flip-coin} considers derandomizing social choice mechanisms by having the players submit integers and using their sum modulo $m$ as a random seed, inducing a game reminiscent of our token games.
\citet{drakopoulos2023blockchain-mediated-persuasion} consider a signaling setting where the sender lacks full commitment power, %
but can design a smart contract which accepts an (unverifiable) world state report from the sender, charges the sender a report-dependent cost, and sends the receiver a corresponding (credibly randomized) signal. 
This is comparable to disarmament of mixed strategies (i.e. signal distributions), but with the additional ability to impose costs on the remaining strategies. 
\end{toappendix}

\section{Preliminaries}
\label{sec:preliminaries}

\textbf{Game theory.}~~We here introduce some %
game-theoretic notation and terminology. %
An \emph{$n$-player (normal-form) game (NFG) $\game$} is a pair $(A, \u)$, where $A = \bigtimes_i A_i$ for some a nonempty set of \emph{actions} $A_i$ for each player $i$, and $\u: A \rightarrow \R^n$ is a \emph{utility function}, with $u_i(a)$ Player $i$'s utility if the players play $\a$. 
We assume $|A_i| \geq 2$ for all Players $i$ unless otherwise stated.
We call the elements of $A$ \textit{outcomes} or action profiles. We use $\Delta(A)$ to denote the set of \textit{correlated strategy profiles}, i.e., distributions over $A$. We extend $\u$ to strategy profiles by taking the expectation: $\u(c)\coloneqq \sum_{\a\in A} c(\a) \u (\a)$\cameraReady{, where $c(a)$ is the probability assigned to outcome $a$ by the correlated strategy profile $c$}. % 
We define $\u(A) = \{\u(a) : a \in A\}$, and define $\u(\Delta(A))$ similarly. %
We use $-i$ to denote the set of players other than $i$. %

For any $n$-player game $\game=(A,\mathbf u)$ and nonempty sets $\hat A_1\subseteq A_1,...,\hat A_n \subseteq A_n$ and letting $\hat A = \hat A_1 \times ... \times A_n$, note that $(\hat A, \mathbf u_{|\hat A})$ is a new game, where $u_{|\hat A}$ denotes the restriction of $\mathbf u$ to $\hat A$. We call this a subgame of $\game$. We %
typically just write $(\hat A, \mathbf u)$, omitting that $\mathbf u$ is restricted to the new action sets. We will often obtain a subgame by removing some set of actions $A_i'$. We then use $\game-A_i'$ as shorthand for $(A_1\times ... \times A_{i-1} \times (A_i-A_i')\times A_{i+1} \times ... \times A_n, \mathbf u)$, the subgame of $\game$ obtained by removing $A_i'$ from $\game$.

Given utility functions $\u$, we say that some outcome \textit{$\a'$ is a (weak) Pareto improvement on $\a$} if for all $i$ we have that $u_i(\a')\geq u_i(\a)$.
We then write $\u(\a')\pgeq \u(\a)$. 
We say that $\a'$ is a \textit{strict} Pareto improvement on $\a$, or $\u(\a') \pge \u(\a)$, if additionally there is a player $i$ s.t.\ $u_i(\a')>u_i(\a)$. 
We say that an outcome $\a$ is \textit{Pareto optimal} within some set if there is no strict Pareto improvement on $\a$ in that set.

Let $\game$ be a game and let $a_i,a_i' \in A_i$ be actions for Player $i$. We say that \emph{$a_i'$ strictly dominates $a_i$} if for all $a_{-i}\in A_{-i}$ we have that $u_i(a_i',a_{-i}) > u_i(a_i,a_{-i})$.

 A function $\iso\colon A \rightarrow A'$ defined by bijections $\iso_i: A_i \rightarrow A_i'$ is a \textit{(game) isomorphism} from $(A,\mathbf u)$ to $(A',\mathbf u ')$ if there exist some $m, b \in \R^n$ with all $m_i > 0$ such that $u_i'(\iso_1(a_1), %
..., \iso_n(a_n) %
) = m_i u_i(a) + b_i$ for all $a \in A$ and all players $i$.
An isomorphism is Pareto improving if $u'_i(\iso(a)) \geq u_i(a)$ for all players $i$ and all $a \in A$ and strictly Pareto improving if this inequality is strict for at least one player and outcome%
.

\cameraReady{
We take Player $i$ to be the same person (or company, etc.) across games. 
For this reason, we don't allow isomorphisms to permute players \cite[as done by e.g.][]{harsanyi1988general,Gabarro_Complexity_Game_isomorphism}. 
We also assume the utilities of a player are comparable between games and generally consider a single utility function $\u$ which operates on outcomes of all games.
This renders the notion of Pareto improvements between games meaningful.
}

\textbf{Outcome correspondences.}~~Following \citet{SPI-original}, we %
reason about safe Pareto improvements by reasoning about outcome correspondence relationships. 
We imagine that the players' strategies across all games can be represented by an (unknown) policy function $\policy$ which maps arbitrary games to their outcome. 
An outcome correspondence is a claim relating the results of playing two different games $\game,\game'$, i.e. a claim about the relationship between $\policy(\game)$ and $\policy(\game')$. 
For example, one possible outcome correspondence is the claim: if playing $\game = (A, \u)$ would result in some $\a \in A$, then playing $\game' = (A', \u)$ would result in $\a'$ or $\a''$ (for some $a', a'' \in A'$). 
This is a claim about $\policy$: If $\policy(\game) = \a$, then $\policy(\game') \in \{a', a''\}$.
More generally, let $\game=(A,\u)$ and $\game'=(A',\u')$ be two games and let $\Phi$ be a multivalued function from $A$ to $A'$. 
Then $\game \sim_{\Phi} \game'$ denotes the (outcome correspondence) claim that, whatever outcome $\policy(\game)$ is, the outcome $\policy(\game')$ will satisfy $\policy(\game') \in \Phi(\policy(\game))$. 

We typically have to make assumptions about what kinds of outcome correspondences hold between games. We make essentially the same assumptions as \citet{SPI-original,Oesterheld2025}. 
The first %
is that %
we can remove strictly dominated actions and the resulting game will be played in the same way \citep[cf.][]{Pearce1984,Kohlberg1986}. %

\begin{assumption}%
\label{assumption:elimination-of-dominated-strategies}
    Let $\game=(A_1,...,A_n,\mathbf{u})$ be a game. Let $\hat a_i$ be an action for Player $i$ that is strictly dominated in $\game$. Then $\game\sim_{\Xi} (A_i - \{\hat a_i\}, A_{-i}, \mathbf{u})$, where $\Xi(\mathbf a) = \emptyset$ if $a_i=\hat a_i$ and $\Xi(\mathbf a) = \{ \mathbf a \}$ otherwise.
    In other words, $\policy_i(\game) \neq \hat a_i$ and $\policy_i(\game) = \policy_i(\game - \{\hat a_i\})$.
\end{assumption}

We often consider the subgame obtained by iteratively removing all strictly dominated actions from $G$, which we denote $\red \game=(\red A, \mathbf u)$  and refer to as the \textit{reduced game}. %
This game is well-known to be unique %
\cite{Pearce1984,Gilboa1990,Apt2004}. %

Our second assumption is, roughly, that isomorphic games are played isomorphically%
.%

\begin{assumption}\label{assumption:isomorphism-assumption}
    Let $\game$ and $\game'$ be isomorphic games without strictly dominated actions. Then let $\Phi$ be the \textit{union} of the isomorphisms from $\game$ to $\game'$, i.e., for every outcome $\mathbf a$ of $\game$, we let $\Phi(\mathbf a) = \{ \phi(\mathbf a) \mid \phi \text{ isomorphism from } \game \text{ to } \game' \}$. Then $\game \sim_{\Phi} \game'$.
    In other words, there must be an isomorphism $\phi$ from $\game$ to $\game'$ such that $\phi(\policy(\game)) = \policy(\game')$.
\end{assumption}

Further, we use the following \textit{transitivity} rule:
if $\game\sim_\Phi \game'$ and $\game' \sim_\Xi \game''$, then $\game \sim _{\Xi \circ \Phi} \game''$. %
Finally,
between any two games $\game$ and $\game'$ the following trivial outcome correspondence holds: $(A,\u)\sim_{\mathrm{all}_{A,A'}} (A',\u')$, where $\mathrm{all}_{A,A'}(a) \coloneqq A'$ for all $a\in A$. %
That is, whatever outcome occurs in $\game$, some outcome in $A'$ must obtain if $\game'$ were to be played.

We call $\game'$ a \textit{safe Pareto improvement (SPI)} on $\game$ if there is a $\Phi$ s.t.\
\begin{enumerate*}
\item $\game\sim_{\Phi} \game'$,
\item for all $\a$ and $\a'\in \Phi(\a)$ we have that $\a'\pgeq \a$, and
\item there exists some realization of $\policy$ (satisfying any assumptions made) such that $\policy(\game') \pge \policy(\game)$
\end{enumerate*}. 
In other words, $\game'$ is an SPI on $\game$ if there is a strictly Pareto-improving outcome correspondence from $\game$ to $\game'$, where strictly Pareto improving means that at least one outcome (which is possible under the assumptions) is guaranteed to be strictly Pareto improved.

\section{Disarmament \SPIs}
\label{sec:disarmament-SPIs}

Perhaps the simplest form of ex-post-verifiable commitment is commitment against taking particular actions. Most straightforwardly, one of the players could commit unilaterally. That is, before playing a game $\game$, Player 1 could announce that they won't take any action from some set $\tilde A_1 \subset A_1$. 
Following \citeauthor{deng2017disarmament} (\citeyear{deng2017disarmament}, \citeyear{deng2018disarmament}), we call such a commitment a \textit{disarmament of $\tilde A_1$}. 
If Player 1's announcement is credible%
, the game $\game - \tilde A_1$ (the game obtained from %
$\game$ by removing the actions in $\tilde A_1$%
) %
is played instead of $\game$. 
We also consider multilateral disarmament. That is, Players 1 and 2 jointly agree not to play $\tilde A_1$ and $\tilde A_2$, respectively. %
Such bilateral disarmament is still \textit{ex post} verifiable.

\cameraReady{As an example, imagine that Alice and Bob are negotiating a contract. For simplicity, imagine that each player can make just three levels of demands: (h(igh), m(edium), and l(ow)). The outcome is determined by the combination of demands. Roughly, higher demands result in better outcomes for a player, but if aggregate demands are too high (at least one high and one medium), no favorable agreement can be reached. 
Additionally, let's imagine that Alice -- since she is a contract lawyer and Bob is not -- can incorporate ``loopholes'' into the fine print of proposed contracts (f) or not (nf). Adding these loopholes is generally good for Alice and bad for Bob. 
Meanwhile, Bob can insist on a simple contract (s) to minimize the impact of loopholes. Unfortunately, simple contracts also reduce flexibility and thus decrease payoffs relative to complicated contracts (c). A version of this game is visualized in \Cref{table:bargaining-with-loopholes}.

\begin{table}
\begin{center}
\begin{tabular}{c|c|c|c|c|c|c|}
   \multicolumn{1}{c}{} & \multicolumn{1}{c}{$(\text{l},\text{s})$}  & \multicolumn{1}{c}{$(\text{m},\text{s})$} &   \multicolumn{1}{c}{$(\text{h},\text{s})$} &   \multicolumn{1}{c}{$(\text{l},\text{c})$} &   \multicolumn{1}{c}{$(\text{m},\text{c})$} &   \multicolumn{1}{c}{$(\text{h},\text{c})$}  \\\cline{2-7}
   $(\text{l},\text{nf})$ &  $2,4$ & $2,5$ & $1,7$ & $4,4$ & $4,5$ & $3,7$ \\\cline{2-7}
   $(\text{m},\text{nf})$ & $3,4$ & $3,5$ & $0,2$ & $5,4$ & $5,5$ & $2,2$ \\\cline{2-7}
  $(\text{h},\text{nf})$ & $5,3$ & $0,2$ & $0,2$ & $7,3$ & $2,2$ & $2,2$ \\\cline{2-7}
   $(\text{l},\text{f})$ & $3,3$ & $3,4$ & $2,6$ & $5,2$ & $6,2$ & $6,2$  \\\cline{2-7}
   $(\text{m},\text{f})$ & $4,3$ & $4,4$ & $1,1$ & $6,2$ & $7,2$ & $3,1$\\\cline{2-7}
   $(\text{h},\text{f})$ & $6,2$ & $1,1$ & $1,1$ & $8,1$ & $3,1$& $3,1$ \\\cline{2-7}
\end{tabular}
\end{center}
\caption{Negotiation Game.}
    \label{table:bargaining-with-loopholes}
\end{table}

By default, Alice will incorporate loopholes into the contract by dominance reasoning.
Anticipating this, dominance reasoning suggests that Bob will insist on a simple contract. That is, the game reduces to its lower-left quadrant.

Now imagine that Alice can credibly commit against the use of loopholes. For instance, we could imagine that she can publicly promise Alice and Bob's social circle that that she won't include loopholes; and we might imagine that the social opprobrium from breaking such a promise would outweigh the gains from including the loopholes. Then there's no reason for Bob to insist on a simple contract. Thus, after such disarmament, the game reduces to the upper-right quadrant, which is an SPI on the original game.}

We \cameraReady{now} consider the computational question of whether for a given game $\game$, there is a disarmament s.t.\ the game resulting from the disarmament is an SPI on $\game$, as well as the question of whether a \textit{given} disarmament induces an SPI.

To get started, we prove a general result characterizing SPIs induced by \Cref{assumption:elimination-of-dominated-strategies,assumption:isomorphism-assumption}. Roughly, to assess whether $\game'$ is an SPI on $\game$ we only need to consider the reduced versions of the two games, $\red{\game'}$ and $\red{\game}$. 
For $\game'$ to be an SPI on $\game$, we need either either $\red{\game}$ to be isomorphic to $\red{\game'}$ via a Pareto-improving isomorphism or every outcome of $\red{\game'}$ to Pareto-dominate every outcome of $\red{\game}$.  

\begin{lemmarep}%
\label{lemma:SPI-iff-isomorphism-on-reduced-game-or-degenerate}
    Consider two games $\game$ and $\game'$. %
    Then $\game'$ %
    is an SPI on $\game$ under \Cref{assumption:elimination-of-dominated-strategies,assumption:isomorphism-assumption} if and only if at least one of the following two conditions holds:
    \begin{enumerate}
        \item \label{item:lemma-condition-isomorphism} There is a strictly Pareto improving isomorphism between $\red \game$ and $\red{\game'}$, i.e., an isomorphism $\iso$ from $\red{\game}$ to $\red{\game'}$ where $\iso(a) \pgeq a$ for all $a \in A$ and $\iso(a) \pge a$ for at least one $a \in A$.%
        \item \label{item:lemma-condition-simple} $\u(a') \pgeq \u(a)$ for all outcomes $a \in \red{A}$ and $a' \in \red{A'}$ and, for at least one outcome $a\in \red{A}$, $\u(a') \pge \u(a)$ for all $a' \in \red{A'}$.
    \end{enumerate}
\end{lemmarep}

\begin{proof}
    $\Rightarrow$/``only if'': We prove this by proving the contrapositive. That is, we prove that if neither of the two given conditions hold, then $\game'$ is not an SPI on $\game$. To prove this non-SPI claim, we need to construct an assignment $\Pi$ of outcomes to games such that the assignment satisfies \Cref{assumption:isomorphism-assumption,assumption:elimination-of-dominated-strategies} %
    and under which $\Pi(\game')$ is not Pareto-better than $\Pi(\game)$.%

    We construct this %
    assignment $\Pi$ as follows. By assumption, there is an outcome $o$ of $\red{\game}$ and an outcome $o'$ of $\red{\game'}$ such that $o'$ doesn't Pareto-dominate $o$. So assign $o$ to $\Pi(\red{\game})$ and $\Pi(\red{\game'})$. Next, partition the set of all normal-form games into sets that are isomorphic after full reduction. (It's easy to see that the induced relation between games is an equivalence relation and thus that this is indeed a partition.) By assumption, $\red{\game}$ and $\red{\game'
    }$ %
    are in separate sets. Now for the sets that contain $\red{\game}$ and $\red{\game'}$, assign outcomes consistently with the already assigned ones. In particular, assign $o$ to $\game$ and $o'$ to $\game'$.  We thus achieve our main objective of making it so that $\game'$ is not an SPI on $\game$. All we have left to do is show that we can complete our assignment of outcomes to games, which we can do as follows: In all other sets, pick any fully reduced game and assign an outcome arbitrarily. Then assign the rest of the set consistently, as before.%
    
    Clearly, by choice of $o$, $o'$, $\Pi(\game)$ isn't Pareto dominated $\Pi(\game')$. Further, it is easy to see that the assignment satisfies \Cref{assumption:isomorphism-assumption,assumption:elimination-of-dominated-strategies}.

    $\Leftarrow$/``if'': %
    Let's consider the first condition (\Cref{item:lemma-condition-isomorphism}). By \Cref{assumption:elimination-of-dominated-strategies}, we have $\game\sim \red{\game}$ and $\game'\sim \red{\game'}$. By \Cref{assumption:isomorphism-assumption}, we have that $\red{\game}\sim_{\Phi} \red{\game'}$, where $\Phi$ is the union of all the isomorphisms between $\red{\game}$ and $\red{\game'}$. Note that we know that at least one of the isomorphisms from $\red{\game}$ to $\red{\game'}$ is Pareto-improving. It's easy to see that all isomorphisms between two given games must induce the same mapping between utility vectors. (They all must map the best/worst outcomes for any Player $i$ in one game to the best/worst outcomes for Player $i$ in the other game. Since a linear function is uniquely specified by two points, they must all act on Player $i$'s utilities in the same way.) Thus, from the fact that one isomorphism is Pareto improving, we can infer that all isomorphisms from $\red{\game}$ to $\red{\game'}$ are Pareto improving. Thus, $\Phi$ is (weakly) Pareto improving. To prove strictness, we also need to construct an assignment of $\Pi$ that assigns a strictly Pareto-better outcome to $\game'$ than to $\game$. This can be done by assigning outcomes $o'$ and $o$ s.t.\ $(o,o')\in \Phi$ and $o'$ is strictly Pareto better than $o$. Such a pair of outcomes exists by \Cref{item:lemma-condition-isomorphism}. The rest of the construction works the same way as the construction above.

    The proof for the second condition (\Cref{item:lemma-condition-simple}) works the same way, except that instead of \Cref{assumption:isomorphism-assumption}, we invoke the trivial outcome correspondence $\mathrm{all}_{\red{ A}, \red{A'}}$ between the reduced games. The fact that this outcome correspondence is Pareto-improving follows immediately from the condition.    
\end{proof}

We call an SPI \textit{simple} if it can be proven using only condition \ref{item:lemma-condition-simple}%
. That is, $\game'$ is a \textit{simple} SPI on $\game$ under \Cref{assumption:elimination-of-dominated-strategies} if, for all outcomes $a' \in \red A'$ and all outcomes $a \in \red A$, $\u(a') \pgeq \u(a)$. Note that simple SPIs can be proved without \Cref{assumption:isomorphism-assumption}.
Similarly, we refer to SPIs based on condition \ref{item:lemma-condition-isomorphism} as \textit{isomorphism} SPIs. That is, a game $\game'$ is an \textit{isomorphism} SPI on a game $\game$ under \Cref{assumption:elimination-of-dominated-strategies,assumption:isomorphism-assumption} if there exists a Pareto-improving isomorphism between $\red G$ and $\red{G'}$.

We now consider the problem of deciding whether a \textit{given} disarmament is a safe Pareto improvement. We show that even under strong restrictions this problem is graph-isomorphism-complete (\GI-complete). %
The graph isomorphism problem is commonly believed to be \NP-intermediate, i.e., in \NP, not %
solvable in polynomial time, but not \NP-hard. (For %
discussions of \GI, see \citealt{mathon1979note}; \citealt{zemlyachenko1985graph}; \citealt{Koebler1993}, \citealt{grohe2020graph}.)

\begin{theoremrep}\label{thm:disarmament-SPI-decision-problem-GI-complete}
The following problem is \GI-complete: Given a game $\game = (A_1, ..., A_n, \mathbf u)$ and sets of actions $(\tilde A_i)_i$ for each player, decide whether the game %
$\game ' = (A_1 - \tilde A _1, ..., A_n - \tilde A_n, \mathbf u)$ is an %
SPI on $\game$ under \Cref{assumption:elimination-of-dominated-strategies,assumption:isomorphism-assumption}.
The problem remains \GI-complete if we restrict attention to %
$n=2$, $|\tilde A_1|=1$ and $\tilde A_2=\emptyset$.
\end{theoremrep}

\begin{proofsketch}
    Whether $\game'$ is a simple SPI on $\game$ can be decided in polynomial time, so we focus on isomorphism SPIs. The first central idea behind the proof is that deciding various questions about whether a given pair of games are isomorphic is \GI-complete, see \Cref{appendix:game-isomorphism-graph-isomorphism-complete}. {\GI}-membership is then easy to prove. For hardness, we reduce from the problem of deciding whether two games are isomorphic%
    . To do this, we construct for any pair of games $\game$, $\game'$, a new game with two properties. First, it reduces to $\game$ plus some gadget when no actions are disarmed. %
    Second, under a particular unilateral disarmament with $|\tilde A_1|=1$ and $\tilde A_2=\emptyset$, it reduces to $\game'+(\epsilon, \epsilon)$, i.e., the game arising from $\game'$ by giving each player an extra $\epsilon$ in each outcome, plus an isomorphic gadget. Whether the disarmament is an SPI then becomes equivalent to the question whether $\game$ and $\game'$ are isomorphic.
\end{proofsketch}

\begin{proof}
\underline{\GI{}-Membership:} Let $\game'=(A_1 - \tilde A _1, ..., A_n - \tilde A_n, \mathbf u)$. By \Cref{lemma:SPI-iff-isomorphism-on-reduced-game-or-degenerate}, the problem of deciding whether $\game'$ is an SPI on $\game$ is equivalent to the question of whether either (1) there is a (strictly) Pareto-improving isomorphism from the fully reduced $\bar \game'$ to the fully reduced $\red{ \game}$ or (2) every outcome of $\bar\game'$ is (strictly) Pareto-better than all outcomes of $\bar\game$.
(Note that $\red{\game'}$ and $\red{\game}$ can be constructed in polynomial time.) By \Cref{prop:Pareto-improving-isomorphism-GI-complete}, deciding the existence of such a (strictly) Pareto-improving isomorphism is in \GI. Clearly, deciding whether all outcomes of $\bar\game'$ Pareto dominate all outcomes of $\bar\game$ -- to determine whether $\bar\game'$ is a simple SPI on $\bar\game$ --
can be done in polynomial time. It follows that the problem is in \GI.

\underline{\GI{}-hardness:} We will reduce from the problem of deciding whether given games $\game$ and $\game'$ are isomorphic via an isomorphism with coefficients 1 and 0, which is \GI{}-complete by \Cref{prop:1-0-isomorphism-GI-complete}. WLOG let the range of utilities in $\game$ and $\game'$ be $[0,1]$. Consider the game in \Cref{table:game-for-proof-of-thm:disarmament-SPI-decision-problem-GI-complete} and the proposed disarmament %
$\tilde A_1=\{x\}$. We will show that this disarmament is an SPI if and only if $\game$ and $\game'$ are isomorphic.

\begin{table*}
\makebox[\linewidth][c]{%
\begin{tabular}{ccccccc}
                                      & $a_1,...,a_m$                                      & $x$                                                  & $x'$                                                & $a_1',...,a_m'$                                   & $y$                                                 & $y'$                                                  \\  \cline{2-7} 
\multicolumn{1}{c|}{$a_1,...,a_n$}    & \multicolumn{1}{c|}{$G$}                           & \multicolumn{1}{c|}{$0$, $0$}                        & \multicolumn{1}{c|}{$-100$, $10$}                   & \multicolumn{1}{c|}{$-\epsilon$, $-\epsilon$}     & \multicolumn{1}{c|}{$-2\epsilon$, $-\epsilon$}      & \multicolumn{1}{c|}{$-100-\epsilon$, $-\epsilon$}              \\ \cline{2-7} 
\multicolumn{1}{c|}{$x$}              & \multicolumn{1}{c|}{$0$, $10+\epsilon$}            & \multicolumn{1}{c|}{$10$, $10+\epsilon$}             & \multicolumn{1}{c|}{$10$, $10$}                     & \multicolumn{1}{c|}{$1+2\epsilon$, $10-\epsilon$} & \multicolumn{1}{c|}{$10+2\epsilon$, $10-2\epsilon$} & \multicolumn{1}{c|}{$10+2\epsilon$, $10-2\epsilon$}   \\ \cline{2-7} 
\multicolumn{1}{c|}{$x'$}             & \multicolumn{1}{c|}{$1$, $-100$}                   & \multicolumn{1}{c|}{$0$, $-100$}                     & \multicolumn{1}{c|}{$-10$, $-10$}                   & \multicolumn{1}{c|}{$-\epsilon$, $-100-\epsilon$} & \multicolumn{1}{c|}{$-\epsilon$, $-100-\epsilon$}   & \multicolumn{1}{c|}{$-100-\epsilon$, $-100-\epsilon$} \\ \cline{2-7} 
\multicolumn{1}{c|}{$a_1',...,a_n'$}  & \multicolumn{1}{c|}{$-\epsilon$, $-100-2\epsilon$} & \multicolumn{1}{c|}{$10-\epsilon$, $-100-2\epsilon$} & \multicolumn{1}{c|}{$10-\epsilon$, $-100-\epsilon$} & \multicolumn{1}{c|}{$G'+\epsilon$}                & \multicolumn{1}{c|}{$\epsilon$, $\epsilon$}         & \multicolumn{1}{c|}{$-100+\epsilon$, $10+\epsilon$}   \\ \cline{2-7} 
\multicolumn{1}{c|}{$y$}              & \multicolumn{1}{c|}{$-\epsilon$, $-100-2\epsilon$} & \multicolumn{1}{c|}{$10-\epsilon$, $-100-2\epsilon$} & \multicolumn{1}{c|}{$10-\epsilon$, $-100-\epsilon$} & \multicolumn{1}{c|}{$\epsilon$, $10+2\epsilon$}   & \multicolumn{1}{c|}{$10+\epsilon$, $10+2\epsilon$}  & \multicolumn{1}{c|}{$10+\epsilon$, $10+\epsilon$}     \\ \cline{2-7} 
\multicolumn{1}{c|}{$y'$}             & \multicolumn{1}{c|}{$-\epsilon$, $-100-2\epsilon$} & \multicolumn{1}{c|}{$10-\epsilon$, $-100-2\epsilon$} & \multicolumn{1}{c|}{$10-\epsilon$, $-100-\epsilon$} & \multicolumn{1}{c|}{$1+\epsilon$, $-100+\epsilon$}& \multicolumn{1}{c|}{$\epsilon$, $-100+\epsilon$}    & \multicolumn{1}{c|}{$-10+\epsilon$, $-10+\epsilon$}   \\ \cline{2-7} 
\end{tabular}%
}
\caption{Construction for the hardness part of the proof of \Cref{thm:disarmament-SPI-decision-problem-GI-complete}}
\label{table:game-for-proof-of-thm:disarmament-SPI-decision-problem-GI-complete}
\end{table*}

First note that the game in \Cref{table:game-for-proof-of-thm:disarmament-SPI-decision-problem-GI-complete} reduces to its top-left 3-by-3 quadrant: First for Player 1, $x$ dominates $a_1',...,a_n'$, $y$ and $y'$. After removing those actions, Player 2's $a_1,...,a_m$ dominate $a_1',...,a_m'$, $x$ dominates $y$ and $x'$ dominates $y'$.

Second, note that after the removal of $x$ for Player 1 in \Cref{table:game-for-proof-of-thm:disarmament-SPI-decision-problem-GI-complete}, the game reduces to its bottom-right 3-by-3 quadrant: First Player 2's $x'$ dominates $a_1,...,a_m$ and $x$. After the removal of $a_1,...,a_m$, Player 1's $a_1',...,a_n'$ dominate $a_1,...,a_n$. Then Player 1's $x'$ is dominated by all of Player 1's other actions. Finally Player 2's $x'$ is dominated by all of Player 2's other actions.

Further notice that some outcomes in the top-left quadrant (e.g., $(x,x')$) are non-Pareto-improved by some outcomes in the bottom-right quadrant.

Thus, by \Cref{lemma:SPI-iff-isomorphism-on-reduced-game-or-degenerate} we have that the disarmament of $x$ for Player 1 is an SPI if and only if there is a Pareto-improving isomorphism between the bottom-right quadrant and the top-left quadrant. 
It is easy to see that this is the case if and only if $\game$ is isomorphic to $\game'$.
\end{proof}
\begin{toappendix}
Why do we specifically consider the problem of identifying isomorphisms with coefficients $1$ and $0$? Because there might be other isomorphisms between $\game$ and $\game'$ that don't induce corresponding isomorphisms between the top-left and bottom-right 3x3 quadrant in the game of \Cref{table:game-for-proof-of-thm:disarmament-SPI-decision-problem-GI-complete}.
\end{toappendix}

Next we consider the problem of deciding whether a given game has \textit{any} disarmament SPI (rather than evaluating a specific candidate). This problem is NP-complete, even if we restrict attention to unilateral SPIs.

\begin{theoremrep}\label{thm:disarmament-SPI-NP-complete}
The following problem is \NP-complete: Given a game $\game$, decide whether there exist %
sets $A_1',...,A_n'$ %
s.t.\ $(A_1 - A_1', ..., A_n - A_n', \mathbf u)$ is a %
SPI on $\game$ under \Cref{assumption:isomorphism-assumption,assumption:elimination-of-dominated-strategies}. The problem remains \NP-complete if we restrict attention to two-player games and $\bar A_2=\emptyset$.
\end{theoremrep}

\begin{proofsketch}
The difficult part is proving hardness. Similar to %
\Cref{thm:disarmament-SPI-decision-problem-GI-complete}, the first central idea is to use the \NP{}-hardness %
of determining whether one game $\game$ can be isomorphically mapped into a subgame of another game $\game'$ (\Cref{thm:subgame-isomorphism-problem}), where the isomorphism keeps utilities constant. %
The main challenge of the proof then is to construct a game that (without disarmament) reduces to %
$\game$ (plus some gadget) and that by (unilateral disarmament) can be made to reduce to any %
subgame of $\game'$ (plus an isomorphic gadget) with an extra utility of $\epsilon$ for all players. Then there is a (strict) (unilateral) disarmament SPI if and only if $\game'$ has a subgame that is isomorphic to $\game$.
\end{proofsketch}

\begin{proof}
    \underline{\NP-membership}:  As certificates (witnesses) we can use the sets $\tilde A_1,\dots,\tilde A_n$, along with the isomorphism between the fully reduced versions of $\game$ and $(A_1 - \tilde A_1, ..., A_n - \tilde A_n, \mathbf u)$. Clearly, these are polynomially sized and can be verified in polynomial time.

    \underline{\NP-hardness}: We reduce from one of the \NPC{} subgame isomorphism problems of \Cref{thm:subgame-isomorphism-problem}: 
    Given games $G$ and $G'$ that cannot be reduced by strict dominance decide whether there is a $1$-$0$-coefficient isomorphism from $G$ into a subgame of $G'$ (i.e., a game obtained by removing some of the actions in $G'$).
    Consider a pair of games $G$ and $G'$ with utilities bounded between $0$ and $1$. (This is w.l.o.g. because if it was not the case, we can renormalize the utilities without changing whether an isomorphism exists.)

    \begin{table*}[]
    \centering
    \begin{tabular}{c|c|c|c|c|c|c|c|c|c|}
    \multicolumn{1}{c}{} & \multicolumn{3}{c}{$\{D\} \times A_2$} & \multicolumn{1}{c}{$(D,\bar a_2)$} & \multicolumn{3}{c}{$\{P\} \times A_2'$} & \multicolumn{1}{c}{$(P, \bar a_2)$} \\
     \cline{2-9}
    $\{ R \} \times A_1'$ & \multicolumn{4}{c|}{\multirow{4}{*}{$-10,-10$}}  & \multicolumn{3}{c|}{$G' + (\epsilon, \epsilon)$} & $\epsilon, 2+\epsilon$\\
    \cline{6-9}
    \multirow{3}{*}{$\{R\} \times A_2'$} &\multicolumn{4}{c|}{}  & \multicolumn{1}{c}{$-1+\epsilon, 3+\epsilon$} & \multicolumn{1}{c}{} & $-2,\epsilon$ & $1+\epsilon, 2+\epsilon$ \\
    &\multicolumn{4}{c|}{}  & \multicolumn{1}{c}{} & \multicolumn{1}{c}{$\ddots$} & & $\vdots$ \\
    &\multicolumn{4}{c|}{}  & \multicolumn{1}{c}{$-2,\epsilon$}& \multicolumn{1}{c}{} & $-1+\epsilon,3+\epsilon$ & $1+\epsilon, 2+\epsilon$ \\
    \cline{2-9}
    $ \{T\} \times A_1$ & \multicolumn{3}{c|}{$G$} & $0,2$ & \multicolumn{4}{c|}{\multirow{4}{*}{$10,-10$}} \\
    \cline{2-5}
    \multirow{3}{*}{$\{T\} \times A_2$} & \multicolumn{1}{c}{$-1,3$} & \multicolumn{1}{c}{} & $-2,-2$ & $1,2$ & \multicolumn{4}{c|}{} \\
    & \multicolumn{1}{c}{} & \multicolumn{1}{c}{$\ddots$} & & $\vdots$ & \multicolumn{4}{c|}{} \\
    & \multicolumn{1}{c}{$-2, -2$}& \multicolumn{1}{c}{} & $-1,3$ & $1,2$ & \multicolumn{4}{c|}{}\\
    \cline{2-9}
    \end{tabular}
    \caption{Construction for the hardness part of the proof of \Cref{thm:disarmament-SPI-NP-complete}.}
    \label{table:construction-for-proof-of-thm:disarmament-SPI-NP-complete-new}
    \end{table*}

    We claim there is a (unilateral) disarmament SPI in \Cref{table:construction-for-proof-of-thm:disarmament-SPI-NP-complete-new} if and only if there is a $1$-$0$ subgame isomorphism from $G$ into $G'$. Specifically, we show first that if there is a $1$-$0$ subgame isomorphism from $G$ into $G'$, then there is a unilateral SPI ($\Leftarrow$). Second, we show that if there is any SPI, then there is also a unilateral SPI in particular and moreover there is a $1$-$0$ subgame isomorphism from $G$ into $G'$ ($\Rightarrow$). This proves NP-hardness of both the unilateral and the bilateral versions of the problem. %

    We first give some brief intuition for the game in \Cref{table:construction-for-proof-of-thm:disarmament-SPI-NP-complete-new} (and for the proof as a whole). Player 1's $T$ and $R$ stand for ``temptation'' and ``refrain'' (or resist temptation), respectively. Player 2's $D$ and $P$ stand for ``defensive'' and ``permissive'', respectively. Without any disarmament, Player 1 is ``tempted'': the $T$ actions strictly dominate the $R$ actions. Player 2 expects Player 1 to be tempted and thus will act defensively. So by default (absent disarmament), the players play the bottom-left $T$--$D$ quadrant. If Player 1 commits to resist temptation (i.e., to play one of the $R$ actions), then Player 2 has no reason to defend (play a $D$ action) and acts permissively instead. Thus, if Player 1 commits against $T$, the players play the top-right quadrant of the game. Playing $(R,P)$ is potentially Pareto-better than playing $(T,D)$, because both players get an extra $\epsilon$ for playing $(R,P)$. However, the players also play another game in parallel and this other game varies depending on whether they play $(R,P)$ or $(T,D)$. To obtain a safe Pareto improvement based on disarming $T$, the players might need to disarm further actions to make the resulting $(R,P)$ subgame isomorphic to the original $(T,D)$ subgame.

    We now conduct the formal proof by showing both directions of the equivalence.
    For both directions, note first that the game in \Cref{table:construction-for-proof-of-thm:disarmament-SPI-NP-complete-new} reduces to its bottom left quadrant.%

    $\Leftarrow$: We argue that if there is a subgame of $G'$ that is $1,0$-coefficient isomorphic to $\game$, then there is a unilateral disarmament SPI. Let the isomorphic subgame of $\game'$ be $\hat\game$ with action sets $\hat A_1$ and $\hat A_2$. Then Player 1 can disarm $\{T\} \times (A_1 \cup A_2)$ %
    and $\{R\} \times ((A_1'-\hat A_1) \cup (A_2' - \hat A_2))$. After this disarmament, note first that all of Player 2's $D$ actions are strictly dominated by the $P$ actions. Further note that for every action $a_2'\in A_2' - \hat A_2$, the action $(P,a_2')$ is now dominated by $(P,\bar a_2)$ (because the only action against which $(P,a_2')$ is a better response is $(R,a_2')$, which was disarmed). %
    Thus, the remaining actions for Player 2 are $\{P\} \times (\hat A_2 \cup \{ \bar a_2\})$. Clearly, if $\hat G$ is isomorphic (with coefficients 1, 0) to $G$%
    , then the bottom-left quadrant is isomorphic with coefficients $1$ and $\epsilon$ to the reduced game after disarmament, $(\{R\} \times (\hat A_1 \cup \hat A_2), \{P\} \times (\hat A_2 \cup \{\bar a_2\}), \mathbf u)$. Thus, the disarmament is an SPI.

    $\Rightarrow$: We have left to show that if there is an SPI, there is a $1,0$ subgame isomorphism from $G$ into $G'$. The proof mostly works by characterizing what this safe Pareto improvement will have to look like, and then extracting the isomorphism from it.
    
    Notice first that in order to Pareto improve on the default (i.e., playing the game of \Cref{table:construction-for-proof-of-thm:disarmament-SPI-NP-complete-new} and thus the bottom-left quadrant), Player 1 has to disarm at least all $T$ actions. This is because if any $T$ actions remain, then (regardless of what, if anything, Player 2 disarms) the game will reduce by iterated dominance to some part of its bottom-left quadrant. But the bottom-left quadrant contains no outcome that is weakly Pareto-improving on all other outcomes in the quadrant. %
    Thus, since no subset can be isomorphic to the full quadrant, by \Cref{lemma:SPI-iff-isomorphism-on-reduced-game-or-degenerate}, there can be no non-trivial SPI that consists just of parts of the bottom left quadrant.

    If Player 1 does disarm all the $T$ actions, then (regardless of what else is disarmed), Player 2's $D$ actions will all be dominated by the $P$ actions. Thus, the SPI results in some part of the upper right quadrant being played.

    Now note again that the top-right quadrant doesn't contain any outcome that (very weakly) Pareto-dominates all outcomes in the bottom-left quadrant. Thus, by \Cref{lemma:SPI-iff-isomorphism-on-reduced-game-or-degenerate}, any disarmament SPI must be a disarmament that (after elimination by dominance) results in a game that is isomorphic to the bottom-left quadrant.

    Now if $\Phi$ is an isomorphism between the bottom-left quadrant and a subgame of the top-right quadrant, then it's easy to see that $\Phi$ must have coefficients $1$ and $\epsilon$, that $\Phi$ must map the %
    $\{T\}\times A_2$ actions onto $\{R\} \times A_2'$, $\{T\}\times A_1$ into $\{ R \} \times A_1'$,  $\{D\}\times A_2$ onto $\{P\} \times A_2'$, and $(D,\bar a_2)$ onto $(P,\bar a_2)$. Furthermore, the elements of $A_2'$ in the image of $\{T\}\times A_2$%
    under $\Phi$ must be the same as the elements of $A_2'$ in the image of $\{D\} \times A_2$ under $\Phi$. Call these actions $\hat A_2$.

    From the above, we can see that the SPI can be achieved unilaterally. Player 1 can disarm the $T$ actions, and disarm %
    $\{R\}\times (A_2' - \hat A_2)$. By dominance, Player 2's $D$ actions as well as the $\{P\}\times (A_2' - \hat A_2)$ will be removed by strict dominance.

    Finally, it is easy to see that $\Phi$ on $\{T\}\times A_1$ and $\{D\}\times A_2$ induces a $1$-$0$ isomorphism between $G$ and $G'$, as desired.
\end{proof}

\cameraReady{
Note that if we bound the number of actions that can be disarmed, the problem returns to being \GI-complete%
, since there are only polynomially many %
disarmaments to try. %
}%

Throughout this paper we consider safe \textit{Pareto} improvements. For \textit{unilateral} disarmaments in particular it is also natural to consider ``safe $u_1$ improvements'', i.e., unilateral disarmaments that are guaranteed (by \Cref{assumption:elimination-of-dominated-strategies,assumption:isomorphism-assumption}) to be better for Player 1. Our proofs of \Cref{thm:disarmament-SPI-decision-problem-GI-complete,thm:disarmament-SPI-NP-complete} apply not just to SPIs but also to safe $u_1$ improvements.

\section{Token Game \SPIs}
\label{sec:token-game-SPIs}

In this section, we consider \SPIs{} achieved by commitments to resolve a game by playing a token game. 
Token games are the same type of mathematical object as \enquote{normal} games, and we'll typically denote them $\tgame = (T, \u)$.
All of our assumptions apply to token games in the same way as to normal games.
However, token games are intrinsically meaningless; their actions and payoffs don't represent anything in the real world. 
Instead, their payoffs must be realized by playing actions in the original game.
We imagine this works as follows. 
Suppose that instead of playing $\game$ directly, the players agree to resolve it by playing the token game $\tgame = (T, \u)$. 
To do so, they simultaneously declare their token actions $t_i \in T_i$, perhaps by writing them down on pieces of paper and then flipping them over. %
This results in a token outcome $t \in T$. 
The players then realize the token payoffs $\u(t)$ by playing some strategy profile in $G$ with that (expected) payoff. 

We say a token game $\tgame$ can be realized in a game $G$ if there exists a \textit{realization function} $\remap: T \rightarrow \F(A)$ such that $\u(t) = \u(\remap(t))$ for all $t \in T$. 
We'll consider two cases for $\F(A)$: the set of correlated strategy profiles $\C(A)$ and the set of pure strategy profiles $A$, and refer to the constructed games as correlated and pure token games, respectively.
Formally, a pure/correlated token SPI on a game $\game$ is a pure/correlated token game $\tgame$ which is realizable in $\game$ and where $\game \sim_{\oc} \tgame$ via a Pareto-improving outcome correspondence $\oc$.

Commitments to play as prescribed by a token game can easily be made \textit{ex post} verifiable. 
The players need to ensure that each player chooses their token action before learning the others'.
This can be done through cryptographic commitment \cite{brassard1988minimum_disclosure_proofs_of_knowledge_commitment, goldreich2004foundations_cryptography} or using physical assumptions\cameraReady{, e.g. by privately writing the actions on pieces of paper}.
\cameraReady{
For correlated token games, the players also need to (\textit{ex post} verifiably) correlate their strategies. 
This is easy to do given a shared source of randomness: the randomness selects an outcome/strategy profile which the players are then required to play. 
This randomness could be provided by physical or cryptographic coin flipping \cite{blum1983coin-fipping-by-phone}, which can even be done non-interactively \cite[cf.][]{zk-mechanisms}, or by having the players submit strings of (purportedly random) bits and taking their XOR \cite[cf.][]{ijcai2024p333_mechanisms-play-game-not-flip-coin}.}
\cameraReady{Note that, compared to cryptographic protocols for implementing correlated equilibria \cite{dodis2000_crypto_impl_CE}, our solutions are much simpler because the players don't need distinct signals. Indeed making the players' signals distinct would generally break \textit{ex post} verifiability. }

We do not consider mixed token games, as they %
offer little benefit over correlated token games and come with substantial drawbacks.
In contrast to the other sections, we do not consider a unilateral version of token SPIs. 
Token SPIs are commitments to play a strategy determined by a token game, which doesn't make much sense if the other players don't participate. 
For further discussion, see \Cref{appendix_subsec:token_game_discussion}.

\begin{toappendix}
\subsection{Discussion of Alternate Versions of Token Games }\label{appendix_subsec:token_game_discussion}
We consider both pure and correlated token games, as each have advantages over the other. 
Pure token games guarantee the players their token payoffs exactly, while in correlated token games these payoffs are achieved only in expectation. 
On the other, correlated token games allow for a much larger space of feasible token payoffs and hence wider range of SPIs. 

We do not consider mixed token SPIs, as they seem to offer little benefit over correlated token SPIs while coming with substantial drawbacks.  
Of course, the space of correlated strategy profiles is larger and more computationally amenable than the space of mixed strategy profiles.
For this reason, mixed strategies are usually considered in cases where the players choose their actions independently, but playing a token game already requires substantial communication.
In particular, when submitting their token actions, the players could also submit a string of randomly generated bits. 
The players could then use (for example) the XOR of these bits as the randomness for the correlated strategy profile prescribed by the token game. 
This is similar to the idea from \cite{ijcai2024p333_mechanisms-play-game-not-flip-coin}.
\subsection{Deferred Proofs from \Cref{sec:token-game-SPIs}}
\end{toappendix}

\Cref{lemma:SPI-iff-isomorphism-on-reduced-game-or-degenerate} shows that there are two types of SPIs: simple SPIs and isomorphism SPIs. We first show that simple token SPIs can be found in polynomial time, before moving on to isomorphism token SPIs, which will be our primary focus.

\textbf{Simple Token SPIs.}~~Applying the definition of simple SPIs to the present setting, a token game $\tgame$ is a simple SPI on a game $\game$ under \Cref{assumption:elimination-of-dominated-strategies} if, 
\begin{enumerate*}[label=(\alph*)]
    \item $\u(t) \pgeq \u(a)$ for all outcomes $t$ in the reduced token game $\red T$ and all outcomes $a \in \red A$ and 
    \item there exists an outcome $a \in \red A$ such that $\u(t) \pge \u(a)$ for all $t \in \red T$
\end{enumerate*}.
As one might expect, there's a simple characterization of simple token SPIs in both the pure and correlated cases. 

\begin{theoremrep}\label{thm:complexity-decide-simple-token-SPI}
    A game $\game$ admits a simple token SPI realizable in $\F(A)$ if and only if there exists a payoff in $\u(\F(A))$ which weakly Pareto dominates all of $\u(\red A)$ and strictly Pareto dominates at least one payoff in $\u(\red A)$.
    For both pure and correlated token SPIs, it can be decided in polynomial time whether a simple token SPI exists. 
\end{theoremrep}

\begin{proof}
    \textbf{Characterization:} We first prove the ``if and only if'' claim from the first sentence.    
    
    ($\Leftarrow$) Suppose there exists a payoff in $\u(\F(a))$ which weakly Pareto dominates all of $\u(\red A)$ and strictly Pareto dominates at least one payoff in $\u(\red A)$. 
    Let $\delta$ be an element of $\F(A)$ with $\u(\delta)$ as in the previous sentence.
    Then the token game $\tgame$ with $|T| =\{t\}$ and $\u(t) = \u(\delta)$ is a simple token SPI on $\game$ and is realizable by $\remap(t) = \delta$.
    
    $(\Rightarrow)$ 
    Suppose $\tgame = (T, \u)$ is a simple token SPI on $\game$, and let $\remap$ be a realization function for $\tgame$ in $\F(A)$. 
    Consider $\remap(t)$ for some arbitrary $t \in T$.
    Then, by the definition of simple SPIs, $\u(t) \pgeq \u(a)$ for all outcomes $a \in \red{A}$ and, for at least one outcome $a\in \red{A}$, $\u(t) \pge \u(a)$.
    Hence, $\remap(t) \in \F(A)$ satisfies the desired conditions. 

    \textbf{Complexity: }
    If $\F(A) = A$, we can decide whether a simple token SPI exists by simply computing each player's maximum utility over $\red A$ $v_i^*$ and then iterating over the $a \in A$ to check whether any $\u(a)$ strictly Pareto dominates $v^*$. 
    This can be done in linear time. 
    
    For the case where $\F(A) = \C(A)$, we reduce the problem to checking whether the optimal solution of the following polynomially sized linear program is strictly greater than $0$.
\begin{align*}
    &\text{Maximize}\quad  \sum_{\red{a} \in \bar A} \sum_{i \in [n]} \left[\left( \sum_{a \in A} p_a u_i(a)\right) - u_i(\bar a) \right]  &\\
    &\text{Subject to:} &\\
    & p_a \geq 0 &\text{for all } a \in A \\
    &\sum_{a \in A} p_a = 1 \\
    &\sum_{a \in A} p_a u_i(a) \geq u_i(\bar a) & \text{for all } \red{a} \in \bar A, i \in [n] \\
\end{align*}

    The variables $p_a$ collectively represent a probability distribution over $A$, i.e. an element of $\C(A)$. 
    The utilities $u_i(a)$ and $u_i(\red a)$ are inputs to the problem instance and so are constants from the perspective of the LP. 
    It's easy to see that the program is indeed linear. 
    
    The first two sets of constraints ensure that $\{p_a\}$ is indeed a probability distribution. 
    The third set of constraints ensures that the distribution represented by $\{p_a\}$ is indeed Pareto better than every outcome in the reduced game $\red A$. 
    The objective sums the utility gain of $\{p_a\}$ over $\red a$ over outcomes $\red a$ and players $i$, and so is strictly positive at optimality if and only if $\{p_a\}$ strictly Pareto dominates at least one point in $\red A$.  

    The LP has $O(|A|)$ variables and $O(|A| n)$ constraints, both of which are polynomial in the input size, so it can be solved in polynomial time.  
    
\end{proof}

\textbf{Isomorphism Token SPIs.}~~We'll focus on isomorphism token SPIs for the rest of the section. 
By definition, a token game $\tgame$ is an isomorphism SPI on a game $\game$ under \Cref{assumption:elimination-of-dominated-strategies,assumption:isomorphism-assumption} if there exists a Pareto improving isomorphism between $\red G$ and $\red{T}$.
We begin by making some simplifying observations. 
When constructing a token game, there's no reason to include any token strategies that can be eliminated by \Cref{assumption:elimination-of-dominated-strategies}, so we'll only consider token games that contain no strictly dominated actions. 
In addition, since the SPI requires an isomorphism from $T$ to $\red A$, we can also consider only token games with $|T_i| = |\red A_i|$ for all $i$. %

A token SPI $\tgame$ on $G$ requires a Pareto-improving isomorphism from $\red G$ to $\red \tgame$ and a realization function $\remap: \u(T) \rightarrow \u(\F(A))$.
The following technical lemma shows that, rather than needing to think about these two functions and their composition, 
we can consider a single function $\remapu: \u(\red A) \rightarrow \u(\F(A))$, which we %
call a utility remapping function. 
We'll %
call a utility remapping function \textit{valid} if it is entrywise positive affine and strictly Pareto improving on $\u(\red A)$.
That is, 
    \begin{enumerate*}[(1)]
        \item For all outcomes $a \in \red A$, $\remapu(\u(a)) \succeq \u(a)$, 
        \item For some outcome $a \in \red A$, $\remapu(\u(a)) \succ \u(a)$, and
        \item For all players $i$, there exist $m_i, b_i \in \R$ with $m_i > 0$ such that $\remapu_i(v) = m_i v_i + b_i$ for all $v \in \u(\red A)$.
    \end{enumerate*} 
The lemma shows a correspondence between %
isomorphism token SPIs realized in $\F(A)$ on $\game$ and %
valid utility remapping functions $\remapu: \u(\red A) \rightarrow \u(\F(A))$. 
Roughly, for any isomorphism token SPI, there's a valid utility remapping function $\remapu: \u(\red A) \rightarrow \u(\F(A))$ which characterizes the SPI's effect on payoffs, and for any valid $\remapu$, there's an isomorphism token SPI with the effect on payoffs described by $\remapu$.

\begin{lemmarep}\label{lemma:nondegenerate-remap-SPI-iff-remapping-linear}
    Let $G$ be a game and $\tgame$ be an isomorphism token SPI on $G$ under \Cref{assumption:elimination-of-dominated-strategies,assumption:isomorphism-assumption} that can be realized in $\F(A)$.
    Then there exists a valid utility remapping function $\remapu: \u(\red A) \rightarrow \u(\F(A))$ such that, for all $a \in \red A$ and any isomorphism $\iso$ from $G$ to $\tgame$, $\u(\iso(a)) = \remapu(\u(a))$.
    Conversely, let $\remapu: \u(\red A) \rightarrow \u(\F(A))$ be a valid utility remapping function on the game $G$. 
    Then there exists an isomorphism token SPI $\tgame$ under \Cref{assumption:elimination-of-dominated-strategies,assumption:isomorphism-assumption} that can be realized in $\F(A)$ and for which, for all $a \in \red A$ and all isomorphisms $\iso$ from $G$ to $\tgame$, $\remapu(\u(a))= \u(\iso(a))$.
    In particular, there exists an isomorphism token SPI realizable in $\F(A)$ if and only if there exists a valid utility remapping function into $\u(\F(A))$.
\end{lemmarep}

\begin{proof}
    Note that for any pair of isomorphic games $G$ and $G'$, all isomorphisms between $G$ and $G'$ must induce the same mapping between payoff vectors, i.e. have the same parameters $m$ and $b$. 
    (They all must map the best/worst outcomes for any Player $i$ in $G$ to the best/worst outcomes for Player $i$ in $G'$. Since an affine function is uniquely specified by two points, they must all act on Player $i$'s utilities in the same way.)

    We'll first prove the first claim. Let $\tgame = (T, \u)$ be an isomorphism token SPI on $G$ that can be realized in $\F(A)$.
    We need to show that there exists a valid $\remapu: \u(\red A) \rightarrow \u(\F(A))$ such that for all $a \in \red A$ and all isomorphisms $\iso$ from $G$ to $\tgame$, $\u(\iso(a)) = \remapu(\u(a))$.
    Since $\tgame$ is realizable in $\F(A)$, there exists a realization function $\remap: T \rightarrow \F(A)$.
    By the definition of isomorphism SPI (\Cref{lemma:SPI-iff-isomorphism-on-reduced-game-or-degenerate}), there exists a strictly Pareto-improving isomorphism $\iso$ from $\red A$ to $T$. 
    Define $\remapu$ by $\remapu(\u(a)) = \u(\iso(a))$ for all $a \in \red A$.
    We know that $\remapu$ is into $\u(\F(A))$ because $\tgame$ is realizable in $\F(A)$. 
    Since all isomorphisms from $G$ to $\tgame$ induce the same effect on payoff vectors, $\remapu(\u(a)) = \u(\iso'(a))$ for all isomorphisms $\iso'$. 
    We need to show $\remapu$ is valid. 
    By the definition of game isomorphism, there exist some $m, b \in \R^n$ with all $m_i > 0$ such that $u_i(\iso(a)) = m_i u_i(a) + b_i$ for all $a \in A$ and all players $i$.
    Since $\remapu(\u(a)) = \u(\iso(a))$ and $u_i(\iso(a)) = m_i u_i(a) + b_i$, this immediately implies that $\remapu$ is a well-defined function (i.e. not a multifunction) and is entrywise positive affine. 
    Finally, $\remapu$ is strictly Pareto improving on $\u(\red A)$ because $\iso$ is strictly Pareto improving on $\red A$.
    
    We'll now prove the second claim.
    Let $\remapu$ be a valid utility remapping function into $\u(\F(A))$.
    We must show there exists an isomorphism token SPI $\tgame$ that can be realized in $\F(A)$ and for which $\remapu(\u(a))= \u(\iso'(a))$ for all $a \in \red A$ and all isomorphisms $\iso'$ from $G$ to $\tgame$.
    
    Let $\tgame$ be a token game with $|T_i| = |\red A_i|$ for all $i \in [n]$. Fix a bijection $\iso_i$ from $\red A_i$ to $T_i$ for all $i$ and let $\iso = (\iso_i)_{i \in [n]}$.
    Define $\u(T)$ by $\u(t) = \remapu(\u(\iso^{-1}(t)))$.
    Then $\tgame$ is realizable in $\F(A)$ because $\remapu$ is into $\F(A)$.
    Because $\remapu$ is entrywise positive affine, there exists $m_i > 0$ and $b_i$ for all players $i$ such that $\remapu_i(v) = m_i v_i + b_i$.
    Hence, $u_i(\iso(a)) = \remapu_i(\u(a)) = m_i u_i(a) + b_i$ for all $i$, and $\iso$ is a game isomorphism. 
    We immediately have $\remapu(\u(a))= \u(\iso(a))$ from our definition of $\u(T)$, and equality also holds for all other isomorphisms $\iso'$ from $G$ to $\tgame$ because they all induce the same mapping on payoffs. 

    It remains to show that $\tgame$ is an SPI on $G$. 
    Let $\oc$ be the union of all isomorphisms from $G$ to $\tgame$. 
    We have $G \sim \red G$ by \Cref{assumption:elimination-of-dominated-strategies}, which is trivially (weakly) Pareto improving, and we have $\red G \sim_{\oc} \tgame$ by \Cref{assumption:isomorphism-assumption}. 
    The latter outcome correspondence is also weakly Pareto improving because
    for any $a \in \red A$ and any $t \in \oc(a)$, $t = \iso'(a)$ for some isomorphism $\iso'$. 
    But all isomorphisms have the same effect on payoffs as $\iso$, so $\u(t)= \u(\iso(a)) = \remapu(\u(a)) \pge \u(a)$ because $\remapu$ is Pareto improving on $\red A$.
    Finally, we show that $\oc$ is strictly Pareto improving.
    Because $\remapu$ is strictly Pareto improving on $\u(\red A)$, there exists some $a \in \red A$ such that $\remapu(\u(a)) \pge \u(a)$. 
    This outcome $a$ is possible under the assumptions and for all $t \in \oc(a)$, $\u(t) = \remapu(\u(a)) \pge \u(a)$, as desired. 
    (Formally, we can construct an assignment $\policy$ satisfying \Cref{assumption:elimination-of-dominated-strategies,assumption:isomorphism-assumption} where $\Pi(G) = a$, $\Pi(\tgame) = \iso(a)$, and $\Pi(G')$ for all other games $G'$ is assigned consistently with these, as in the proof of \Cref{lemma:SPI-iff-isomorphism-on-reduced-game-or-degenerate}.)

    The \enquote{in particular} if and only if claim follows immediately from the correspondence proven above. 
\end{proof}

We will now apply \Cref{lemma:nondegenerate-remap-SPI-iff-remapping-linear} to pure and correlated token SPIs, beginning with the latter.
We show that correlated token SPIs can be found efficiently. 
In addition, we characterize the existence of correlated token SPIs in two-player games. 

\begin{theoremrep}[Characterization of isomorphism correlated token SPIs]\label{thm:complexity-decide-correlated-remapping-SPI}
    It can be decided in polynomial time whether a game $\game$ admits an isomorphism correlated token SPI. %
    Furthermore, if $\game$ has exactly two players, we have the following characterization of when isomorphism correlated token SPIs exist. 
    Let $V = \u(\red A)$, $v_i^{min}$ and $v_i^{max}$ be the minimum and maximum values of $v_i$ in $V$, and $V^* \subseteq V$ be the set of points in $V$ which cannot be strictly Pareto improved in $\u(\C(A))$.
    Assume $|V|\geq 2$, as otherwise isomorphism token SPIs are equivalent to simple SPIs and there's an SPI iff the unique point in $V$ is not Pareto optimal in $\C(A)$.
    \begin{enumerate}
        \item If $|V^*| =0$, $\game$ admits the desired SPI. %
        \item If $|V^*| =1$, call that point $v^*$. Then
        \begin{enumerate}
            \item If $v^*_i \in \left\{ v_i^{min}, v_i^{max} \right\}$ for both $i$, $\game$ admits the desired SPI.
            \item If only one player $i$ has $v^*_i \in \left\{ v_i^{min}, v_i^{max} \right\}$, $\game$ admits the desired SPI if and only if, for all $v$ in $V$ with $v_i \neq v_i^*$, $(v_i + \eps_v, v_{-i}) \in u(\C(A))$ for some $\eps_v >0$.
            \item If for both $i$, $v^*_i \not\in \left\{ v_i^{min}, v_i^{max} \right\}$, $\game$ does not admit the desired SPI.
        \end{enumerate}
        \item If $|V^*| \geq 2$, $\game$ does not admit the desired SPI.
    \end{enumerate}
\end{theoremrep}

\begin{proofsketch}
    To prove the first part of the theorem, we reduce the decision problem to checking the optimal value of a linear program. %
    For the characterization in the 2-player case, we use \Cref{lemma:nondegenerate-remap-SPI-iff-remapping-linear}, demonstrating a valid $\remapu$ for the positive results and showing none exists for the negative results. 

    A key observation for the negative results is that, if a value $v_i$ cannot be Pareto improved within $\u(\C(A))$, strictly for Player $i$, then it must be a fixed point of $\remapu_i$. 
    In case 3, where $|V^*|\geq 2$, each of these Pareto optimal values must be a fixed point of $\remapu_i$ in every dimension $i$. 
    Hence, each $\remapu_i$ has at least two fixed points and must be the identity by linearity, so there can be no SPI.
    In case 2(c), where $|V^*| = 1$, the Pareto optimal value $v^*$ is a fixed point of each $\remapu_i$ at an intermediate value $v_i \not \in \{v_i^{min}, v_i^{max}\}$. 
    This implies that each $\remapu_i$ must be the identity; otherwise it would fail to be improving on either the values less than $v_i$ or those greater than $v_i$.

    For case 1, when $|V^*|=0$, we show that the utility remapping function $\remapu(v) = (1-\eps) v + \eps r^{max}$, where $r^{max} = \left(\max_{r \in R}r_1, \max_{r \in R}r_2 \right)$ is Pareto improving and feasible for some $\eps > 0$. 
    Geometrically, this corresponds to mapping each value $v$ some $\eps$ of the way towards $r^{max}$ on the line segment between $v$ and $r^{max}$. 

    For case 2(a), where $|V^*|=1$ and this point $v^*$ satisfies $v^*_i \in \left\{ v_i^{min}, v_i^{max} \right\}$ for both $i$, we have two subcases. 
    If $v^{*}$ is maximal in both dimensions, $\remapu(v) = (1-\eps) v + \eps v^{max}$ is feasible by convexity. %
    If $v^{*}$ is maximal in dimension $i$ and minimal in dimension $j$, we show that the $\remapu$ defined by $\remapu_i(v) = (1-\eps)v_i + \eps v_i^*$ and $\remap_j(v) = v$ is Pareto improving and feasible for some $\eps > 0$.
    (Note that $v^*$ can't be minimal in both dimensions since then we would have $|V|=1$.)
    
    For case 2(b), where $v^*$ satisfies $v^*_i \in \left\{ v_i^{min}, v_i^{max} \right\}$ for only one $i$, $v^*$ must be maximal in dimension $i$ and hence $v_i^{*}$ is a fixed point of $\remapu_i$.
    Since $v^*_j$ is an intermediate fixed point of $\remapu_j$, $\remapu_j$ must be the identity. 
    Thus, the only potential Pareto-improving $\remapu$ has the form $\remapu_i(v) = (1-\eps)v_i + \eps v_i^*$, as in case 2(a). 
    This is feasible if and only if all points $v$ with $v_i \neq v_i^{*}$ can be improved in the $i$ dimension, as desired. 
\end{proofsketch}

\begin{proof}
\textbf{Characterization:} 
By \Cref{lemma:nondegenerate-remap-SPI-iff-remapping-linear}, the desired SPI exists if and only if there's a positive affine utility remapping function $\remapu: \u(\red A) \rightarrow \u(\C(A))$ which is strictly Pareto improving on $\u(\red A)$.

A substantial part of the proof involves reasoning about values $v_i$ which must be fixed points of $\remapu_i$, i.e. where $\remapu_i(v_i) = v_i$. 
Let $R$ denote the convex region $\u(\C(A))$. %
Observe that, if for any $v \in V$, there does not exist any $v' \in R$ with $v' \pgeq v$ and $v'_i > v_i$, then $v_i$ must be a fixed point of $\remapu_i$: Player $i$'s utility cannot be increased without decreasing another Player's utility and thus rendering $\remapu$ not Pareto improving. 
In particular, any point $v \in V^*$ must be a fixed point of $\remapu_i$ for all $i$. 

Now, observe that if any $\remapu_i$ has two fixed points, the only possible positive affine $\remapu_i$ is the identity. 
In addition, if any $\remapu_i$ has a fixed point at an intermediate value $v_i^* \not\in \{ v_i^{max}, v_i^{min}\}$, then $\remapu_i$ must also be the identity. 
Any $\remapu_i$ with $m_i >1$ would fail to be improving for Player $i$ for $v_i < v_i^*$, and any $\remapu_i$ with $m_i < 1$ would fail to be improving for Player $i$ for $v_i > v_i^*$. 

We are now ready to prove the characterization. 
Let's define $r_i^{max} = \max_{r \in R} r_i$. 
Let's also define $v^{max} = \left(v_1^{max}, v_2^{max}\right)$ and $r^{max}$ analogously.
\cameraReady{In this proof, we'll use $\mathds{1}_i$ to denote the standard basis vector with $1$ in dimension $i$ and $0$ in dimension $-i$, so $(v_i + \eps_v, v_{-i}) = v + \eps_v \mathds{1}_i$.}%  

\textit{Case 1:} Assume $|V^*| = 0$. 
We claim that $\remapu(v) = (1-\eps) v + \eps r^{max}$ meets the required conditions for some $\eps >0$.
First, observe that $\remapu$ is Pareto improving on $V$ and strictly Pareto improving on $V-\{v^{max}\}$, which is nonempty since $|V|\geq 2$.  
Hence, we just need to show that there exists $\eps >0$ such that $\remapu$ is feasible, that is $\remapu(v) \in R$ for all $v \in V$.   
For each $v \in V$, consider the line segment from $v$ to $r^{max}$, parameterized by $(1-\eps) v + \eps r^{max}$ for $\eps \in [0,1]$. 
If $r^{max}$ is in $R$ then by the convexity of $R$, each line segment is entirely in $R$ and hence $\remapu$ is feasible for any $\eps \leq 1$.

Otherwise, $r^{max}$ not in $R$. 
We seek to show that each $v$ can be mapped some fraction of the way $\eps_v$ along the segment $vr^{max}$, i.e. to the point $(1-\eps_{v}) v + \eps_{v} r^{max}$. 
If so, we can say $\eps = \min_v \eps_v$, and have that $\remapu(v)$ is in $R$ for all $v$ and hence is feasible, as desired. 

\begin{figure}[h]
\begin{center}
\begin{tikzpicture}[scale=.9]

    \draw[thick, gray!50, fill=gray!10] 
        (0,0) -- (2, -1) -- (4, -1) -- (6, 0) -- (5, 3) -- 
        (4, 4) -- (2, 5) -- (0, 5) -- (-1, 4) -- (-1, 2) -- cycle;

    \draw[thick, red] (6,0) -- (5, 3) -- (4, 4) -- (2, 5);

    \coordinate (rmax) at (6,5);
    \filldraw[blue] (rmax) circle (2pt) node[above right] {$r^{max}$};

    \draw[->,thick] (-1.5,-1.5) -- (7,-1.5) node[below left] {Player $i$'s utility};
    \draw[->,thick] (-1.5,-1.5) -- (-1.5,6) node[right] {Player $j$'s utility};

    \coordinate (p1) at (2,-0.5);
    \coordinate (p2) at (4,3.5);
    \coordinate (p3) at (1.5,3.5);
    \coordinate (p4) at (4,1);
    \coordinate (p5) at (-0.5,4);
    \coordinate (p6) at (1,5);
    \coordinate (p7) at (-0.5,1);
    \coordinate (p8) at (5,-0.5);

    \foreach \pt in {p1, p2, p3, p4, p5, p6, p7, p8}{
        \draw[dashed, blue,thick] (\pt) -- (rmax);
    }

    \foreach \pt in {p1, p2, p3, p4, p5, p6, p7, p8}{
        \filldraw (\pt) circle (2pt);
    }

    \coordinate (v) at (-0.5,1);
    \filldraw (v) circle (2pt) node[left] {$v$};

    \coordinate (p) at (-0.5,1.5);
    \filldraw[green!40!black] (p) circle (2pt) node[above]{$p$};

    \coordinate (m) at (6, 0);
    \filldraw[green!40!black] (m) circle (2pt) node[right] {$m$};

    \coordinate (w) at (1/11, 15/11);
    \filldraw[green!40!black] (w) circle (2pt) node[below] {$w$};

    \draw[very thick, densely dotted, green!40!black] (m) -- (p);
 
\end{tikzpicture}
\end{center}
    \caption{Case 1, subcase where $r^{max} \not \in R$}
    \label{fig:Correlated-Toke-SPI-Case_1}
\end{figure}
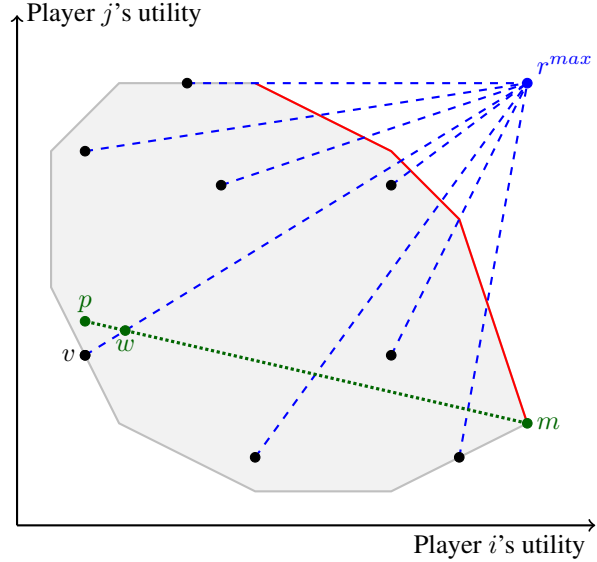

This case is illustrated in \Cref{fig:Correlated-Toke-SPI-Case_1}.
The light gray region is $R$.
The Pareto frontier is red while the other boundaries of $R$ are darker gray. 
The black circles are points that could be in $V$, and the segments $vr^{max}$ are dashed blue.
It's geometrically intuitive that each $v$ can be mapped some nonzero distance along $vr^{max}$, and we prove that formally below.
The proof includes the construction of points $m$ and $p$, and we include an example of these for the labeled $v$ in the figure.

    Consider an arbitrary $v \in V$. 
    We know that $v$ is not on the Pareto frontier of $R$ since $|V^*|=0$.
    We now consider two cases based on whether $v_i = r_i^{max}$ for some $i$. 
    (Of course, $v_i$ cannot equal $r_i^{max}$ for both $i$ since then $v = r^{max} \in R$.)
    
    First, consider the case where the equality holds for some $i$, and say w.l.o.g. $v_1 = r_1^{max}$.
    We also know that $v_2 < r_2^{max}$, as otherwise we'd have $v=r^{max}$. 
    Since $v$ is not Pareto optimal in $R$, there exists a point $p \in R$ which strictly Pareto dominates $v$. 
    But then it must be the case that $p_2 > v_2$ and $p_1 = v_1 = r^{max}$. 
    In other words, $p$ is on the segment $vr^{max}$ (and $p\neq v$), so we can map $v$ some distance nonzero along $vr^{max}$, as desired. 

    Now, consider the case where $v_i \neq r_i^{max}$ for either $i$. 
    For simplicity, let's renormalize our points such that $v = (0,0)$, $r^{max} = (1,1)$, and $vr^{max}$ is a segment of the identity line.
    Again, consider some point $p \in R$ which strictly Pareto dominates $v$. 
    If $p$ is on the segment $vr^{max}$, then we can map $v$ some distance along $vr^{max}$ and we're immediately done, so we can assume $p_1 \neq p_2$.
    Let's say w.l.o.g. that $p_2>p_1$. 
    We then know that $p_1 \geq v_1 = 0$ and $p_2 > v_2 = 0$.

    Fix a point $m \in R$ such that $m_1 = r_1^{max} =1$. 
    Since $m_1 = 1 > m_2$ and $p_2 > p_1$, i.e. $m$ and $p$ lie on the opposite side of the identity line, the segment $mp$ must intersect the identity line at some point. 
    This intersection point must be on the segment $vr^{max}$ because $m_1 = 1$ while $0 \leq p_1 < 1$, and the intersection point is not $p$ itself because $p_2 > p_1$. 
    Note that the entire segment $pm$ is weakly greater than $v$ in the first dimension and strictly greater everywhere except possibly the $p$ endpoint. 
    Therefore, the intersection of $mp$ and $vr^{max}$ is a point $w$ with $w_1 > v_1 = 0$. 
    Hence, $w$ is a point on $vr^{max}$ which is not equal to $v$. 
    Finally, since $m$ and $p$ are both in $R$, $w$ is feasible by the convexity of $R$.
    Hence, we can map $v$ some nonzero distance along $vr^{max}$, as desired.

\textit{Case 2:} Assume $|V^*|= 1$, and call that point $v^*$.

\textit{Subcase (a):} Suppose $v^*_i \in \left\{ v_i^{min}, v_i^{max} \right\}$ for both $i$. We seek to show that $\game$ admits the desired SPI.
First, suppose $v^*_i = v_i^{max}$ for both $i$. 
Then we can take $\remapu(v) = (1-\eps)v + \eps v^*$ for any $\eps \in (0,1)$. 
This is very similar to the (simpler) subcase of case 1, where $r^{max}$ is in $R$.
Because $v^*$ is maximal in both dimensions, $\remapu$ is Pareto improving on $V$ and strictly Pareto improving on $V-\{v^{max}\}$, which is nonempty since $|V|\geq 2$.  
Since $v^* \in R$, $\remapu$ is feasible by the convexity of $R$.
(Aside: If we let $\eps=1$, $\remap$ corresponds to the simple SPI where all outcomes of the token game have payoff $v^{max}$.)

It can't be the case that $v^*_i = v_i^{min}$ for both $i$, since then $v^*$ would be the unique point in $V$, contradicting our assumption that $|V|\geq 2$. 
Hence, it remains to consider the case that $v^* = \left(v_i^{max}, v_j^{min}\right)$. 
We claim that defining $\remapu$ by $\remapu_i(v) = (1-\eps)v_i + \eps v_i^*$ for some $\eps > 0$ and $\remapu_j(v) = v$ gives SPI for $\eps >0$. 
This is clearly Pareto improving on $V$, and strictly so for $v$ with $v_i < v_i^* = v_i^{max}$.
The set of such points is nonempty because $|V|\geq 2$ and any $v\neq v^*$ with $v_i = v_i^{*}$ would Pareto dominate it since $v_j^* = v_j^{min}$, a contradiction.
Hence, it only remains to show feasibility. 

\begin{figure}[h]
\begin{center}
\begin{tikzpicture}[scale=.9]
    \draw[thick, gray!50, fill=gray!10] 
        (0,0) -- (2, -1) -- (4, -1) -- (6, 0) -- (5, 3) -- 
        (4, 4) -- (2, 5) -- (0, 5) -- (-1, 4) -- (-1, 2) -- cycle;

    \draw[thick, red] (6,0) -- (5, 3) -- (4, 4) -- (2, 5);

    \draw[->,thick] (-1.5,-1.5) -- (7,-1.5) node[below left] {Player $i$'s utility};
    \draw[->,thick] (-1.5,-1.5) -- (-1.5,6) node[right] {Player $j$'s utility};

    \coordinate (p2) at (4.5,3);
    \coordinate (p3) at (1.5,3.5);
    \coordinate (p4) at (4,1.5);
    \coordinate (p5) at (-0.5,4);
    \coordinate (p6) at (1,5);
    \coordinate (p7) at (-0.5,1);

    \coordinate (vstar) at (5 + 2/3, 1); 

    \draw[thick, dashed] (5 + 2/3, -1.5) -- (5 + 2/3, 6);

    \foreach \pt in {p2, p3, p4, p5, p6, p7}{
        \draw[dashed, blue, thick] (\pt) -- (\pt -| vstar);
        \filldraw (\pt) circle (2pt);
    }

    \coordinate (p) at (4,2);
    \filldraw[green!40!black] (p) circle (2pt) node[above]{$p$};
    \draw[very thick, densely dotted, green!40!black] (p) -- (vstar);

    \filldraw (p4) circle (2pt) node[left]{$v$};
    \filldraw[green!40!black] (4 + 5/6, 1.5) circle (2pt) node[above right]{$w$};

    \filldraw (vstar) circle (2pt) node[right]{$v^* = (v_i^{max}, v_j^{min}$)};

\end{tikzpicture}
\end{center}
    \caption{Case 2(a), subcase where $v^* = (v_i^{max}, v_j^{min}$)}
    \label{fig:Correlated-Toke-SPI-Case_2a}
\end{figure}
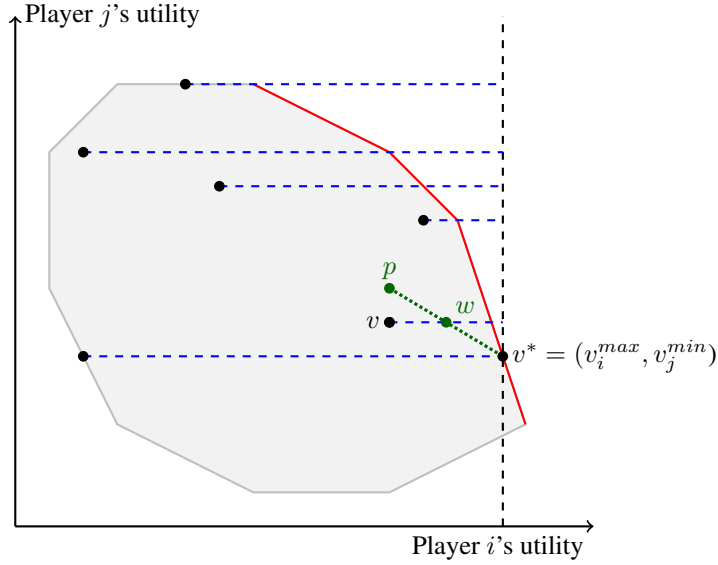

This case is illustrated in \Cref{fig:Correlated-Toke-SPI-Case_2a}.
As before, the light gray region is $R$, the red segments are the Pareto frontier, and the darker gray segments are the other boundaries of $R$. 
The black circles are points that could be in $V$; Note that no point in $V$ can be below $v^*$ in the picture. 
Our proposed $\remapu$ maps each point in $V$ some $\eps$ fraction of the way along the dashed blue line towards the vertical line through $v^*$.
It's geometrically intuitive that this $\remapu$ is feasible, but we prove it formally below. 
To show feasibility for each $v$, our proof constructs another pair of points $p$ and $w$. 
The figure shows an example of this for the labeled point $v$, with the $p$ and $w$ in green. 

For feasibility, we claim that it suffices to show that, for all points $v$ with $v_i \neq v_i^*$, we have $v+ \eps_v \mathds{1}_i \in R$ for some $\eps_v > 0$. 
Given this, taking $\eps = \min_v \frac{\eps_v }{v_i^* - v_i}$ makes the $\remap$ above feasible: 
For each $v$, $\remap_i(v) = (1-\eps)v_i + \eps v_i^* = v_i + \eps(v_i^*-v_i) \leq v_i + \eps_v$, so $\remap(v)$ is on the line between $v$ and $v+ \eps_v \mathds{1}_i$ and feasible by the convexity of $R$. 

We now show that, for all points $v$ with $v_i \neq v_i^*$, we have $v+ \eps_v \mathds{1}_i \in R$ for some $\eps_v > 0$. 
Consider an arbitrary $v$ with $v_i < v_i^{max}$. 
Because $v \not\in V^*$, there must exist some point $p \in R$ that strictly Pareto dominates $v$. 
If $p_j = v_j$, then $p_i > v_i$ and we're immediately done.
Otherwise, we have $p_i \geq v_i$ and $p_j > v_j$.
We also know that $v_i^* > v_i$ and $v_j^* \leq v_j$ since $v_j^* = v_j^{min}$.
Consider the line segment $\overline{v^* p}$, which is in $R$ by the convexity of $R$.
Because $v_j^* \leq v_j$ and $p_j > v_j$, the segment contains a point $w$ with $w_j = v_j$, which is not the $p$ endpoint. 
And since $v^*_i > v_i$ and $p_i \geq v_i$, the segment is strictly greater than $v$ in the $i$ dimension everywhere but possibly the $p$ endpoint, and so in particular $w_i > v_i$.
Therefore, $w = v+ \eps_v \mathds{1}_i$ for some $\eps_v > 0$ and is in $R$, as desired. 

\textit{Subcase (b):} Suppose $v^*_i \in \left\{ v_i^{min}, v_i^{max} \right\}$ for exactly one player $i$. We seek to show that $\game$ admits the desired SPI if and only if, for all $v$ in $V$ with $v_i \neq v_i^*$, $v + \eps \mathds{1}_i \in R$ for some $\eps >0$.
Note that if $v^*_i = v_i^{min}$, then $v^*$ would also achieve $v_j^{max}$ because any point that exceeds $v^*$ in the $j$ dimension would Pareto dominate it. Hence, we can assume $v^*_i = v_i^{max}$.

\begin{figure}[h]
\begin{center}
\begin{tikzpicture}[scale=.9]
    \draw[thick, gray!50, fill=gray!10] 
        (0,0) -- (2, -1) -- (4, -1) -- (6, 0) -- (5, 3) -- 
        (4, 4) -- (2, 5) -- (0, 5) -- (-1, 4) -- (-1, 2) -- cycle;

    \draw[thick, red] (6,0) -- (5, 3) -- (4, 4) -- (2, 5);

    \draw[->,thick] (-1.5,-1.5) -- (7,-1.5) node[below left] {Player $i$'s utility};
    \draw[->,thick] (-1.5,-1.5) -- (-1.5,6) node[right] {Player $j$'s utility};

    \coordinate (p0) at (5+1/3, -1/3);
    \coordinate (p1) at (2,0);
    \coordinate (p2) at (4.5,3);
    \coordinate (p3) at (1.5,3.5);
    \coordinate (p4) at (4,1.5);
    \coordinate (p5) at (-0.5,4);
    \coordinate (p6) at (1,5);
    \coordinate (p7) at (-0.5,1);

    \coordinate (vstar) at (5 + 1/3, 2); 
    \filldraw (vstar) circle (2pt) node[right]{$v^*$};

    \draw[dashed,thick] (5 +1/3, -1.5) -- (5 + 1/3, 6);

    \foreach \pt in {p0, p1, p2, p3, p4, p5, p6, p7}{
        \draw[dashed, blue,thick] (\pt) -- (\pt -| vstar);
        \filldraw (\pt) circle (2pt);
    }

    \coordinate (p9) at (4, -1);
    \filldraw[green!40!black] (p9) circle (2pt) node[below]{points described in the iff condition};
    \draw[very thick, densely dotted, green!40!black] (4,-1) -- (5+1/3, -1/3); 

    \filldraw (p0) circle (2pt);

\end{tikzpicture}
\end{center}
    \caption{Case 2(b), where $v^*_i = v_i^{max}$ and $v_j^* \not \in \{v_j^{min}, v_j^{max}\}$}
    \label{fig:Correlated-Toke-SPI-Case_2b}
\end{figure}
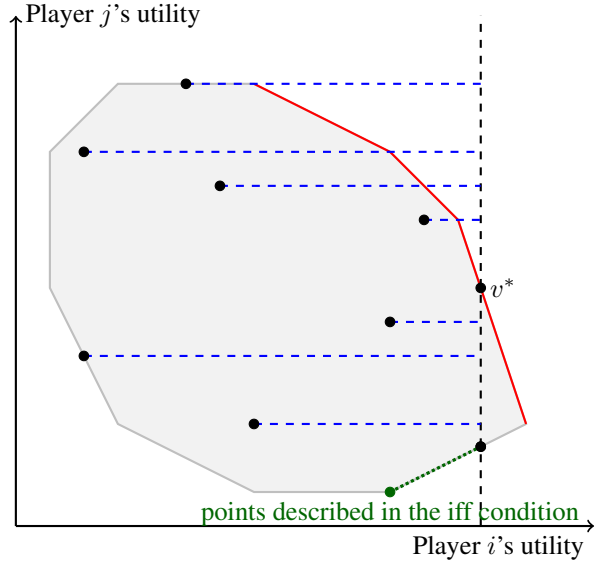

This case is illustrated in \Cref{fig:Correlated-Toke-SPI-Case_2b}.
As before, the light gray region is $R$, the red segments are the Pareto frontier, and the darker gray segments are the other boundaries of $R$. 
The black disks are points that could be in $V$.
Note that, unlike in case 2(a), there are now points in $V$ with smaller $j$ coordinate than $v^*$. 
We claim that the only possible valid $\remapu$ maps each point in $V$ with along the dashed blue line towards the vertical dashed line through $v^*$, and that this $\remapu$ is feasible if and only if there are no points on the green boundary segment (which is not inclusive of its black, top/right endpoint).
We now prove this formally. 

The if direction is very similar to the part of subcase (a) where $v^* = \left(v_i^{max}, v_j^{min}\right)$. 
We again claim that defining $\remapu$ by $\remapu_i(v) = (1-\eps)v_i + \eps v_i^*$ for some $\eps > 0$ and $\remapu_j(v) = v$ gives an SPI for $\eps >0$. 
This $\remapu$ is clearly Pareto improving, and strictly so for any $v'$ with $v'_i < v_i^*$.
And such $v'$ exist: Any $v$ with $v_j=v_j^{max}$ is such a $v'$, because $v_j^{max} > v_j^*$ and therefore, because $v$ cannot Pareto dominate $v^*$, $v_i < v_i^*$.
Hence, $\remapu$ is strictly Pareto improving.
The ``if" assumption, that for all $v$ in $V$ with $v_i \neq v_i^*$, $v + \eps \mathds{1}_i \in R$ for some $\eps_v>0$, is exactly the assumption that we showed was sufficient for $\remapu$ to be feasible in the proof of subcase (a). 
Hence, $\remapu$ is feasible and strictly Pareto improving, as desired. 

To show the only if direction, we prove the statement's contrapositive.
That is, we assume that there exists a point $v$ with $v_i \neq v_i^{max}$ where $v + \eps \mathds{1}_i \not\in R$ for any $\eps >0$ and we show that the instance does not admit the desired SPI. 
Observe that $\remapu_{-i}$ must be the identity since $v_{-i}^*$ must be an intermediate fixed point of $\remapu_{-i}$. 
In addition, $v_i^*$ must be a fixed point of $\remapu_i$.
By the contrapositive assumption, there exists a point $v$ with $v_i \neq v_i^{max}$ where $v + \eps \mathds{1}_i \not\in R$ for any $\eps >0$.
Since $\remapu_{-i}$ is the identity and this $v$ cannot be increased in the $i$ dimension while being held constant in the $-i$ dimension, $v_i$ must be a second fixed point of $\remapu_i$. 
Hence, $\remapu$ is the identity in both dimensions and no isomorphism token SPI exists. 

\textit{Subcase (c):} Suppose that for both $i$, $v^*_i \not\in \left\{ v_i^{min}, v_i^{max} \right\}$. We seek to show that $\game$ does not admit the desired SPI.
This follows immediately from the observations at the beginning of the proof: $v^*$ must be a fixed point of each $\remapu_i$, and since $v^*_i \not\in \left\{ v_i^{min}, v_i^{max} \right\}$ this implies that each $\remapu_i$ is the identity. 

\textit{Case 3: } Suppose $|V^*| \geq 2$. We seek to show that $\game$ does not admit the desired SPI.
Again, this follows immediately from the observations at the beginning of the proof. 
Each point in $V^*$ must be a fixed point of each $\remapu_i$, so each $\remapu_i$ has at least two fixed points and must be the identity. 

This concludes the proof of the characterization of correlated isomorphism token SPIs in $2$-player games. 

\textbf{Complexity: }%
Now, we prove the complexity claim of the theorem, i.e., that deciding whether such SPIs exist can be done in polynomial time. %
    We show this by reducing it to checking whether the optimal solution to the following polynomially sized linear program is strictly positive. 
    That is, we claim that there exists a correlated token SPI if and only if the LP admits a solution with objective value strictly greater than $0$. 
    By \Cref{lemma:nondegenerate-remap-SPI-iff-remapping-linear}, the desired SPI exists if and only if there exists a utility remapping function $\remapu$ which is a strictly Pareto improving entrywise positive affine function from $\u(A)$ to $\u(\C(A))$.

    Our linear program has variables $m_i, b_i \in  \R$ for all $i\in [n]$ representing the parameters of the positive affine transformations $\remapu_i(v_i) = m_i v_i + b_i$ as well as variables $p_a^{v^j}$ for all $v^j\in \u(\red A)$ and $a \in A$ representing each distribution in $\C(A)$ used to realize the payoffs $\remapu(v^j)$. 
    The payoff vectors $V = \{ v^1 ... v^k \} = \u(\red A)$ and utility function $\u: A \rightarrow \R^n$ are inputs to the program, i.e. constants. 

    The program searches over all possible parameters of this $\remapu$, but also allows $m_i =0$, while a positive affine function must have each $m_i > 0$.
    Similarly, it only constraints $\remapu$ to be weakly Pareto improving rather than strictly Pareto improving.
    This is necessary to make the program convex, and we'll show that we can nevertheless use the program to decide whether SPIs exist. 

\begin{align*}
    &\text{Maximize}\quad \sum_{v^j \in V} \sum_{i \in [n]} (m_i v_i^j + b_i - v_i^j) &\\
    &\text{Subject to:} &\\
    & p_a^{v^j} \geq 0 &\text{for all } a \in A, j \in [k]  \\
    &\sum_{a \in A} p_a^{v^j} = 1 &\text{for all } j \in [k] \\
    & m_i \geq 0 &\text{for all } i \in [n]\\
    & m_i v_i^j + b_i \geq v_i^j &\text{for all } i \in [n], j \in [k] \\
    & \sum_{a \in A} p_a^{v^j} u_i(a) = m_i v_i^j + b_i &\text{for all } i \in [n], j \in [k] %
\end{align*}

    It's easy to verify that the program is indeed linear. 
    It has $O(|A|^2)$ variables and $O(|A|^2)$ constraints. 
    Hence, it is polynomially sized, and so its feasibility can be decided and its objective optimized in polynomial time. 

    The first two constraints together ensure that the variables $\{p_a^{v^j}\}$ for each $v^j$ represents a valid distribution in distribution over $A$, i.e. correlated strategy profile. 
    The third constraint ensures that each parameter $m_i$ of the affine transformation $\remapu_i(v_i) = m_i v_i + b_i$ on each player $i$'s utilities is nonnegative, as previously discussed. 
    The fourth constraint ensures that $\remapu$ is weakly Pareto improving on each outcome in $\red A$. 
    The last constraint ensures that the remapping $\remapu$ defined by $\remapu_i(v_i) = m_i v_i + b_i$ is indeed a feasible mapping into $\u(\C(A))$ because each desired payoff profile $\remapu(v^j) = (m_i v_i^j + b_i)_{i \in [n]}$ is achieved by the strategy profile $\{p_a^{v^j}\}$.

    Therefore the linear program searches over all possible valid $\remapu$, as well as $\remapu$ that are invalid.
    These could be invalid because either they are not positive affine (i.e. some $m_i=0$) or because they are not strictly Pareto improving.
    Nevertheless, we claim that there exists a correlated token SPI if and only if the LP admits a solution with objective value strictly greater than $0$. 
    Note that the LP objective is player $i$'s payoff gain $\remapu_i(v_i^j) - v_i^j$, summed over players $i$ and payoffs $v^j$ in the reduced game. 
    Hence, the objective is strictly positive if and only if the (potentially non positive affine) $\remapu$ corresponding to the LP variables $\{m_i\}, \{b_i\}$ is strictly Pareto improving. 

    The \enquote{only if} direction is clear: since the LP searches over a superset of the set of valid $\remapu$, if the optimal LP solution does not correspond to a strictly Pareto improving $\remapu$, no valid $\remapu$ can be strictly Pareto improving. 

    To see the \enquote{if} direction, suppose the LP admits a strictly positive solution with variables $\{m_i\}, \{b_i\}$ corresponding to some $\remapu$ which may not be positive affine since some $m_i$ could be $0$.
    (Note that $\remapu$ is strictly Pareto improving because the objective is strictly positive.) 
    We claim the $\remapu': \u(\red A) \rightarrow \u(\C(A))$ defined by $\remapu'(v) = (1-\eps)\remapu(v) + \eps v$ is positive affine and strictly Pareto improving for any $\eps$ in the interval $(0,1)$, and hence shows the existence of an SPI. 
    First, since both $\remapu$ and the identity mapping $v \mapsto v$ are into $\u(\C(a))$, $\remapu'$ is also a valid function into $\u(\C(A))$ by the convexity of $\u(\C(A))$.
    Second, the $\remapu'$ is positive affine because $\remapu'_i(v_i) = (1-\eps)(m_i v_i + b_i) + \eps v_i = ((1-\eps)m_i + \eps)v_i + (1-\eps)b_i$, so the coefficient $(1-\eps)m_i + \eps > \eps$ of $v_i$ is strictly positive for each $i$. 
    Finally, $\remapu'$ is strictly Pareto improving because it's a convex combination of the identity and the strictly Pareto improving $\remapu$.

    Therefore, there exists a correlated token SPI if and only if the LP admits a solution with objective value strictly greater than $0$, which can be decided in polynomial time.

\end{proof}

Now, we turn our focus to isomorphism \textit{pure} token SPIs, i.e., those whose payoffs must be realized over \textit{pure} strategy profiles. 
The following theorem gives an algorithm for finding such SPIs. 
It runs in time $|A|^{O(n)}$, scaling polynomially in the number of outcomes but exponentially in the number of players.
Since we generally assume each player has at least two actions, $n \in O(\log(|A|))$ and the overall runtime is $|A|^{O\left(\log |A|\right)}$. %
Such runtimes are sometimes called quasipolynomial time.

\begin{theoremrep}\label{thm:complexity-decide-pure-remapping-SPI-n-player}
    It can be decided in time $|A|^{O(n)} \in  |A|^{O\left(\log |A|\right)}$, i.e. quasipolynomial time, whether a game $\game$ admits a pure isomorphism token SPI.
    For any fixed number of players $n$, this is polynomial time. 
\end{theoremrep}

\begin{proofsketch}
    By \Cref{lemma:nondegenerate-remap-SPI-iff-remapping-linear}, the desired SPI exists if and only if there exists a valid utility remapping function $\hat \remap: \u(\red A) \rightarrow \u(A)$, i.e. one which is a strictly Pareto improving, entrywise positive affine function on $\u(\red A)$.
    We give an algorithm, \Cref{alg:Pure_Token_SPI}, that decides whether such a valid $\remapu$ exists.
    Let $V = \u(\bar A)$ be the set of payoffs in the reduced game.
    First, we efficiently find a set $V' \subseteq V$ of at most $n+1$ payoffs that contains at least two distinct payoffs for every player whose utility is not constant in $V$. 
    
    Any choice of $\remapu': V' \rightarrow \u(A)$ for such a $V'$ determines the parameters $(m_i, b_i)_{i \in [n]}$ of any entrywise positive affine extension of $\remapu$ to $V$:
    If $V'$ contains two distinct values $v'$ and $v''$, then $\remapu'(v')$ and $\remapu'(v'')$ determine the positive affine function in dimension $i$. 
    If $V'$ is constant in dimension $i$, so is $V$ and thus any extension must map all $v_i$ to the same value $\remapu'_i(v')$.
    Hence, we can check whether a given $\remapu'$ has a valid extension to $V$ by checking whether $mv + b \in \u(A)$ for each $v \in V- V'$ and whether it is strictly Pareto improving on $V$. 
    This can be done in $O(|A|^2)$.

    Therefore, we can simply try each of the $O(|A|^{n+1})$ possible $\remapu'$. 
    This is polynomial for any constant number of players and quasipolynomial in general, since all players have at least two actions and therefore $n \in O(\log_2(|A|))$.
    Our algorithm enumerates all valid $\remapu$, so we can also optimize arbitrary (quasi-)polynomial-time computable objectives over SPIs by simply computing the objective value of each valid $\remapu$.
\end{proofsketch}

\begin{proof}

    Let $\bar A$ be the set of outcomes in the reduced game $\red \game$, and let $V = \u(\bar A)$ be the set of payoffs. 
    For each payoff $v \in V$, let $I(v)$ be the subset of $\u(A)$ which weakly Pareto improves on $v$.

    By \Cref{lemma:nondegenerate-remap-SPI-iff-remapping-linear}, there exists a pure isomorphism token SPI if and only if there exists a utility remapping function $\hat \remap: \u(\red A) \rightarrow \u(A)$ which is a strictly Pareto improving positive affine function on $\u(\red A)$.
    Therefore, the desired token SPI exists if and only if there exists some $\remapu: V \rightarrow \{I(v) | v \in V\}$ where (1) $\remapu(v) \in I(v)$ for all $v \in V$, (2) for all players $i$, there exist $m_i, b_i \in \R$ with $m_i > 0$ such that $\remapu_i(v) = m_i v_i + b_i$ for all $v \in V$, and (3) there exists some $v \in V$ such that $\remapu(v) \pge v$.

    We give a quasipolynomial algorithm that decides whether such a $\remapu$ exists, \Cref{alg:Pure_Token_SPI}. 
    Roughly, the algorithm works by finding a small set $V'$ of points such that for each choice of $\remapu': V' \rightarrow \u(A)$, there is at most one possible extension of $\remapu'$ into a valid $\remapu : V \rightarrow \u(A)$, and it can be efficiently checked whether such a $\remap$ exists.
    Overall then, checking each of the quasi-polynomially many choices of this $\remapu'$ lets us decide whether an SPI exists. 

        \begin{algorithm}
\caption{Deciding the existence of Pure Isomorphism Token SPIs}\label{alg:Pure_Token_SPI}
\begin{algorithmic}
\Require $\game = (A, \u)$
\State Let $V\gets \u(\red A)$ be the set of payoffs in the reduced game $\red G$
\For{$v \in V$}
    \State Let $I(v)$ be the subset of $\u(A)$ which Pareto improves on $v$ %
\EndFor
\LComment{Find a set of at most $n+1$ points $V' \subseteq V$ with the property that, for all dimensions $i$, either $V$ has only one value in dimension $i$ or or V' contains at least two distinct values in dimension $i$}
\State Let $V' \gets \{v\}$ for some arbitrary $v \in V$
\For{$i \in [n]$}
    \If{$V$ contains multiple values in dimension $i$ and $V'$ does not}
        \State $V' \gets V' \cup \{w\}$ for some $w \in V$ which different from $V'$ in dimension $i$ 
    \EndIf
\EndFor 
\LComment{Check if each possible choice of $\remapu': V' \rightarrow \u(A)$ can be extended into a valid $\remapu: V \rightarrow \u(A)$}
\For{each function $\remapu': V' \rightarrow \u(A)$}
    \State Let $\textit{valid} \gets \texttt{True}$
    \For{each dimension $i \in [n]$}
        \State Consider the relation $\left\{ \left(v_i', \remapu'_i(v') \right) : v' \in V' \right\}$
        \If{the relation is either not a (single-valued) function or not a positive affine function from $\R$ to $\R$}
            \State Set $\textit{valid} \gets \texttt{False}$
        \Else
            \State Let $\ell_i: \R \rightarrow \R$ be a positive affine function with $\ell_i (v_i')= \remapu'_i(v')$ for all $v' \in V'$. \Comment{$\ell_i$ is unique unless $V$ contains only one value in dimension $i$, in which case only $\ell_i(v_i)$ matters}
        \EndIf 
    \EndFor
    \If{not \textit{valid}} \textbf{Continue} \Comment{continue means skip the rest of the current iteration of the for loop}
    \EndIf
    \If{$(\ell_i)_{i \in [n]}$ is not strictly Pareto improving on $V$}
        \State Set $\textit{valid} \gets \texttt{False}$
    \EndIf
    \For{each $v \in V - V'$}
        \State Let $w$ be the vector defined by $w_i = \ell_i(v_i)$ for all $i \in [n]$
        \If{$w \not\in I(v)$}
            \State Set $\textit{valid} \gets \texttt{False}$
        \EndIf
    \EndFor
    \If{ \textit{valid} }
        \Return \enquote{Yes, an SPI exists}
    \EndIf
\EndFor
\Return \enquote{No SPI exists}
\end{algorithmic}
\end{algorithm}

    First, \Cref{alg:Pure_Token_SPI} finds a set $V' \subseteq V$ of at most $n+1$ payoffs such that, for every player $i$, either $v_i$ is equal for all $v \in V$ or $V'$ contains at least two distinct values $v_i'\neq v_i''$ for $v', v'' \in V'$.
    To do so, we initialize $V'$ with an arbitrary single point $v \in V$. 
    For each player $i \in [n]$, if $V'$ does not already contain a pair of points where Player $i$'s payoffs differ, we iterate through $V$ and find a point where Player $i$'s payoff different from that in $V'$ and add it to $V'$ (or conclude no such point exists). 
    This runs in $O(n*|A|) \subseteq O(|A|^2)$, which is dominated by the next step. 

    Next, \Cref{alg:Pure_Token_SPI} checks each possible choice of $\remapu': V' \rightarrow \u(A)$.
    If we find that that a particular $\remapu'$ can be extended into a valid $\remapu$, we have found an SPI and return \enquote{Yes, and SPI exists}. 
    Otherwise, we continue to the next choice of $\remapu'$. 
    If no $\remapu'$ can be extended into a valid $\remapu$, there can be no SPI (because any valid $\remapu$ has a valid restriction $\remapu' = \remapu\vert_{V'}$) and we return \enquote{No SPI exists.}

    For each dimension $i$, we consider the relation $\left\{ \left(v_i', \remapu'_i(v') \right) : v' \in V' \right\}$ and attempt to find the parameters of a valid extension of $\remapu'$ to $V$.
    If this relation is not a single-valued function, i.e. a single value $v_i$ is associated with multiple distinct values, then there is no valid extension of the choice of $\remapu'$. 
    In this case, we set \textit{valid} to false and continue to the next $\remapu'$. 
    Otherwise, we view the relation as a function from $\R$ to $\R$ and check whether it's positive affine. 
    If it's not, again there is no valid extension of $\remapu'$ and we continue to the next $\remapu'$. 
    If it is, we let $(\ell_i)_{i \in [n]}$ be the positive affine function. 
    For dimensions $i$ with at least two distinct values of $v_i'$ in $V'$, $\ell_i$ is unique; 
    For all other dimensions $i$, there is only value of $v_i$ in $V$, and so we define $\ell_i(v_i) = v_i$. 
    These checks can be done in time $O(n^3)$. 

    If these checks pass, we know that $\remapu'$ can be extended into a entrywise positive affine function from $\R$ to $\R$. 
    We still need to check that the extension is strictly Pareto improving and feasible, i.e. into $\u(A)$. 

    We first check whether $(\ell_i)_{i \in [n]}$ is strictly Pareto improving on $V$, which can be done in $O(|V|) \subseteq O(|A|)$ time, and set \textit{valid} to false if not. 
    Finally, we check in $O(|V|^2) \subseteq O(|A|^2)$ time whether $\remapu(v) \in \u(A)$ for the candidate $\remapu$ defined by the $\ell_i$ and for each $v \in V-V'$.
    Again, we set \textit{valid} to false and continue if not. 
    If it possible for all $v$, we have a valid $\remapu$ and return \enquote{Yes, an SPI exists}.

    Since $|V'|\leq n+1$, there are at most $|\u(A)|^{n+1} \in O(|A|^{n+1})$ possible choices of $\remapu'$. 
    For each of these choices, the above algorithm runs in time $O(n^3 + |A|^2 + |A|)$, which is $O(|A|^2)$ since all players have at least two actions and thus $n \in O(\log_2(|A|))$.
    Hence, our overall algorithm decides in time $O(|A|^{n+3})$ whether a pure isomorphism token SPI exists. 
    As claimed, this is polynomial for any constant number of players and quasipolynomial in general, again because $n \in O(\log_2(|A|))$.

    This concludes the proof of the decision problem complexity claim.
    To see the optimization claim, observe that \Cref{alg:Pure_Token_SPI} enumerates all valid utility remapping functions $\remapu$.
    To convert it into an optimization algorithm, one can simply modify the final for loop to track the best valid $\remapu$ found so far (according to the objective function) rather than returning \enquote{yes} when the first SPI is found. 
    Then, at the end we can simply return the best $\remapu$. 
    This optimization algorithm still runs in quasipolynomial time for any quasipolynomial time computable objective function. 
\end{proof}

The quasipolynomial efficiency of finding %
pure token isomorphism SPIs is in some sense an artifact of the fact that the representation size of normal-form games scales exponentially in the number of players.
In particular, an abstracted version of the underlying pure token SPI problem %
is \NPC{}.
In this problem, rather than a game, we're given a set of input payoff vectors and a set of target payoff vectors. 
These correspond to the sets of payoffs in the reduced game and full game, respectively. 
The problem asks whether there %
is
a strictly Pareto-improving, entry-wise positive affine mapping from the input set to target set. 
Of course, such a mapping exists if and only if there's a pure isomorphism token SPI in games with these reduced and full game payoffs.

\begin{theoremrep}\label{thm:abstracted-pure-remapping-NPC}
    The following problem is NP-complete.
    Given a set $S$ of input vectors and a set of target vectors $T$ in $\mathbb{R}^n$, decide whether there exists a strictly Pareto improving, entrywise positive affine mapping from $S$ to $S\cup T$. 
    That is, a function $\remap: S \rightarrow S\cup T$ such that 
    \begin{enumerate}
        \item $\remap(v) \pgeq v$ for all $v \in S$, 
        \item $\remap(v) \pge v$ for some $v \in S$, and
        \item For all players $i$, there exist $m_i, b_i \in \R$ with $m_i > 0$ such that $\remap_i(v) = m_i v_i + b_i$ for all $v \in S$.
    \end{enumerate}
\end{theoremrep}
\begin{toappendix}
    We refer to this problem as \GPRname{}.
\end{toappendix}

\begin{proofsketch}

    We reduce from the problem of graph $3$-coloring. 
    Given a graph $(V, E)$, we construct a remapping instance which has a satisfying remapping if and only if the graph admits a $3$-coloring. 
    Our vectors have one dimension corresponding to each vertex. 
    $S$ consists of the $\binom{n}{2}$ vectors with value $1$ in two dimensions and $0$ in all others. 
    We construct $T$ so that, in each dimension $i$, we must have $\remap_i (0) = .5$ and $\remap_i(1) \in \{1,2,3\}$. 
    $\remap_i(1)$ corresponds to the color of vertex $i$.
    Specifically, $T$ consists of vectors which have value $.5$ in all but two dimensions, and in those two dimensions can be some subset of $\{1,2,3\}$. 
    For each $(i, j) \not \in E$, $T$ contains all of the vectors $v$ with $(v_i, v_j) \in \{1, 2, 3\} \times \{1, 2, 3\}$, so that the colors of $i$ and $j$ do not constrain each other.
    For each $(i, j) \in E$, $T$ does not contain vectors $v$ with $v_i = v_j \in \{1,2,3\}$, encoding the constraint that vertices $i$ and $j$ cannot be the same color.
\end{proofsketch}

\begin{proof}
    The problem is clearly in \NP, as the remapping function is sub-linearly sized and can serve as a witness. 
    To show \NP-hardness, we reduce from the problem of graph 3-coloring, which is \NPH{} \cite{lovasz1973coverings_graph-3coloring-hardness}.
    In this problem, we are given a graph $(V, E)$ and must decide whether there exists a vertex 3-coloring, which is a function $c: V \rightarrow \{1, 2, 3\}$ such that for all pairs of adjacent vertices $(v_1, v_2) \in E$, $c(v_1) \neq c(v_2)$.

    Consider a graph $(V, E)$ with $|V| = n$. (We assume $n \geq 5$.)
    We construct a vector remapping instance with $n$-dimensional vectors as follows. 
    The input vectors are the $\binom{n}{2}=\frac{n(n-1)}{2}$ vectors that are $1$ in two dimensions and zero everywhere else: $S = \{ \mathds{1}_{i,j} : 1 \leq i < j \leq n \}$.
    The target vectors each have value $.5$ in all but two entries. 
    They can have any value in $\{1, 2, 3\}$ in each of these other two entries, except that if $(i, j) \in E$, then no target vector can have equal values aside from $.5$ at indices $i$ and $j$. 
    In other words, $T$ is the set of vectors of the form $a_i \mathds{1}_i + a_j \mathds{1}_j + .5 \mathds{1}_{-i,j}$ for $1 \leq i < j \leq n$ and $a_i, a_j \in \{1, 2, 3\}$ where $(i,j) \in E \Rightarrow a_i \neq a_j $. 
    Observe that this instance is polynomially sized in the graph coloring instance: it has $\Theta(|V|^2)$ input vectors and $\Theta(|V|^2)$ output vectors and each vector is $|V|$-dimensional. 
    
    Now, we claim that there exists a pure token SPI in our constructed instance if and only if $G$ admits a 3-coloring.

    Note that, for each input vector $\mathds{1}_{i,j}$, the only Pareto improving elements of $S\cup T$ are itself and those vectors of $T$ whose non-$.5$ values are those same indices $i$ and $j$. 
    Next, observe that no input vector $\mathds{1}_{i,j}$ can be mapped to itself in a strictly Pareto improving remapping. 
    If $\remap(\mathds{1}_{i,j}) = \mathds{1}_{i,j}$, then $\remap_k(0) = 0$ for all $k \in [n]\setminus\{i,j\}$.
    But then no vector in $T$ can be used in the image of $\remap$ (since $n\geq 5$), because there would be at least one dimension where $0$ is mapped to both $0$ and $.5$. 
    The only option for $\remap$ is then the identity. 

    Hence, we assume that the image of $\remap$ is a subset of $T$. 
    This implies that $\remap_i(0) = .5$ for all $i \in [n]$. 
    Each $\remap_i$ is therefore fully defined by the value in $\{1, 2, 3\}$ to which it maps $1$. 

    $\Rightarrow:$
    Suppose there exists a satisfying $\remap$. 
    Let $c_i$ to be the value $\remap_i(1)$ to which the positive affine transformation for dimension $i$ maps $1$. 
    Then we claim that coloring each vertex $v_i$ color $c_i$ is a proper 3-coloring. 
    To see this, consider an arbitrary pair of adjacent vertices $v_i$ and $v_j$. 
    Then $\remap(\mathds{1}_{i,j})$ is some target vector $a_i \mathds{1}_i + a_j \mathds{1}_j$, where $a_i = c_i$ and $a_j = c_j$.  
    By construction, the possible target vectors have differing values in dimension $i$ and $j$: they are $\{a_i' \mathds{1}_i + a_j' \mathds{1}_j : a_i', a_j' \in \{1, 2, 3\}, a_i' \neq a_j'\}$. 
    Hence, $c_i \neq c_j$, and the coloring is proper. 
    
    $\Leftarrow:$
    Suppose the graph $G$ has a 3-coloring $c: V \rightarrow \{1, 2, 3\}$.
    Then consider the remapping $\remap$ defined by $\remap( \mathds{1}_{i,j}) = c(v_i)\mathds{1}_i + c(v_j)\mathds{1}_{j}$.
    This is a well-defined function (ie not a multifunction) because in each dimension $i$ it maps $0$ to $.5$ and $1$ to $c_i$.
    Since these values ($0$ and $1$) are the only two values in dimension $i$ in the input set, the remapping is linear, and since $c_i \in \{1,2,3\}$, it's strictly Pareto improving. 
    Finally, $\remap$ maps only into $S\cup T$ since, by the definition of a proper coloring, $c(v_i) \neq c(v_j)$ for any pair of adjacent vertices $(v_i, v_j)$, we have that each $\remap( \mathds{1}_{i,j}) \in T$.
\end{proof}

This also shows that the pure token SPI problem becomes \NPH{} if the game is represented in a more succinct form, e.g. as a dictionary which only stores payoffs that aren't uniformly zero.%

\begin{toappendix}
    \subsection{Details of Complexity Claims about Pure Token SPIs}\label{subsec:details-pure-token-SPI-complexity}

    We claimed in the main text that, as a corollary of \Cref{thm:abstracted-pure-remapping-NPC}, the problem of finding pure token SPIs profiles becomes \NPH{} in many representations of normal form games.
    We now formally show this by demonstrating that the instances $S, T$ of \GPRname{} to which we reduce graph 3-coloring in \Cref{thm:abstracted-pure-remapping-NPC} are equivalent to the pure token SPI problem in actual games. 
    Consider an instance $S, T$ of \GPRname{}.

    Let $T= \{ t^0, t^1, ... t^{|T|-1} \}$ index the vectors of $T$ in some arbitrary order.
    We'll also define $t^{|T|} = .5\mathds{1}$ and $t^i = \vec 0$ for $i > |T|$.

    Consider the following $n$-player game.
    For $i \in \{1, 2\}$, $A_i = \{0, 1, d\}$. ($d$ stands for dominated.)
    For $i >2$, $A_i = \{0, 1\}$. 
    The payoffs are as follows:

    \begin{itemize}
        \item If $a_1\neq d$ and $a_2 \neq d$, \\$\u(a) = \begin{cases}
            \mathds{1}_{i,j} \text{ if } a_i = a_j \neq a_k \text{ for some } i \neq j \text{ and all } k \in [n]-\{i,j\} \\
            \vec{0} \text{ otherwise }
        \end{cases}$
        \item If $a_1 = d$, $a_2 \neq d$, $\u(a) = (-10, 10, 0, ..., 0)$. 
        \item If $a_1 \neq d$, $a_2 = d$, $\u(a) = (10, -10, 0, ..., 0)$. 
        \item If $a_1 = a_2 = d$, then $\u(a) = t^i$, where $i \coloneqq a_{-1,2}$ when viewed in binary and $t^i \coloneqq \vec 0$ for $i \geq |T|$.
    \end{itemize}

    Essentially, assuming no players play a dominated action, each player picks an action in $\{0,1\}$, with the goal of picking the same action as exactly one other player. 
    If some pair of players succeeds, they each get a payoff of $1$ while all other players get $0$. 
    If no pair succeeds, all players receive a payoff of $0$. 
    The first two players also each have a strictly dominated action $d$, and if both play $d$ then any payoff in $T$ can be achieved. 

    Now, observe that for both Players $1$ and $2$, the $d$ action is strictly dominated by both $0$ and $1$. 
    All other actions are undominated, as any action $b \in \{0,1\}$ is the unique best response for Player $i$ if $a_{-i}$ contains exactly $b$. 
    Hence, $\red A = \{0,1\}^n$ and $\u(\red A) = \{\mathds{1}_{i,j} : i \neq j\} \cup \{\vec 0\} = S \cup \{\vec 0\}$. 
    Clearly $\u(A) = S \cup T \cup \left\{\vec{0}, .5\mathds{1}, (-10, 10, 0, ..., 0), (10, -10, 0, ..., 0) \right\}$.
    
    We now show that this is equivalent to the \GPRname{} instance as constructed in \Cref{thm:abstracted-pure-remapping-NPC}. 
    Since $(-10, 10, 0, ..., 0)$ and $(10, -10, 0, ..., 0)$ are not Pareto-improving on any of $S$, their inclusion in $\u(A)$ makes no difference. 
    So this game is equivalent to having $\u(\red A) = S \cup \{\vec{0}\}$ and $\u(A) = S \cup T \cup \{\vec{0}, .5\mathds{1} \}$. 
    In the \Cref{thm:abstracted-pure-remapping-NPC}, we argue that any strictly Pareto improving mapping from $S$ to $S \cup T$ must map $0$ to $.5$ in each dimension. 
    This is still true after constraining the remapping further by adding adding $\vec 0$ to its domain, and this addition doesn't change the decision problem because we can have $\remap(\vec 0) = .5\mathds{1}$.
    Hence, there is a strictly Pareto improving token game on this game $G$ if and only if the corresponding \GPRname{} instance is satisfiable, as desired.

    Observe that our constructed game $G$ has $O(n^2)$ nonzero payoffs, and so is polynomially sized in the original vertex cover instance. 
    This shows that finding pure token SPIs becomes \NPC{} in game representations that only store the nonzero payoffs.

\end{toappendix}

\textbf{Optimization.}~~In addition to deciding whether SPIs exist in various settings, we'll also consider optimizing over SPIs.
    We define our objectives on valid utility remapping functions, which \Cref{lemma:nondegenerate-remap-SPI-iff-remapping-linear} shows specify the effect of an SPI on the players' payoffs. 
    We define a class of \textit{linear} objective functions, which includes (weighted) utilitarian social welfare under some beliefs about the outcome of the reduced game.
    It generally does not include either Nash social welfare or maximizing one player's benefit subject to a minimum on the other player's benefit (under beliefs about the outcome of the reduced game). 
    Roughly, we show that linear objectives over correlated token SPIs can be efficiently optimized. 
    In contrast, \textit{arbitrary} objectives over pure token SPIs can be optimized in quasipolynomial time, the same time complexity as our algorithm for the decision problem. 

    There's an important %
    complication %
    regarding what it means to optimize over correlated token SPIs. 
    Valid $\remapu$ must have strict inequalities $m_i > 0$, so the space of valid $\remapu$ is not closed.
    Thus, there may not be an optimal SPI: It might be possible to get arbitrarily close to a particular objective value but not to achieve it.
    We show that we can optimize linear objectives over correlated token SPIs in the strongest sense one could hope for given this issue.
    For details, see \Cref{appendix:Optimizing_Token_SPIs}.

\begin{toappendix}
    \subsection{Optimizing over Token SPIs}
    \label{appendix:Optimizing_Token_SPIs}

    In addition to deciding whether SPIs exist in various settings, we'll also consider optimizing (w.l.o.g. maximizing) objective functions over SPIs. %
    We define our objectives on valid utility remapping functions, which \Cref{lemma:nondegenerate-remap-SPI-iff-remapping-linear} shows specify the effect of an SPI on the players' payoffs. 
    An objective assigns a real-valued quality to each SPI, and can be viewed as a function $f: \{\remapu\} \rightarrow \R$, where $\{\remapu\}$ is the set of valid remapping functions from $\u(\red A)$ into $\u(\F(A))$.
    One particular class of objectives we'll consider is \textit{linear} objectives. 
    These are characterized by a set of linear functions $f^{\red v}: \R^{n} \rightarrow \R$ for each each value $\red v \in \u(\red A)$, such that $f(\remapu) = \sum_{\red v} f^{\red v}(\remapu(\red v))$. 
    Here, linear means that each $f^{\red v}$ is a weighted sum of the components of $\remapu(\red v)$, i.e. can be computed by $f^{\red v}(\remapu) = \sum_i w_i^{\red v} \remapu_i(\red v)$ for some $w_i^{\red v}$.
    We assume w.l.o.g. that $f^{\red v}$ doesn't have an additive term, since it would add the same constant to the objective value of each $\remapu$.

    One important linear objective is utilitarian social welfare under some belief $P$ over outcomes of the reduced game: $f^{\red v} = P(\red v) \sum_i \remapu_i(\red v)$.
    Another is \enquote{subjective} utilitarian social welfare, where each player's expected utility is computed with respect to their own beliefs $P^i$ over the outcome of the reduced game: 
    $f^{\red v} = \sum_i P^i(\red v) \remapu_i(\red v)$.
    These are of course equivalent to the (subjective) utilitarian welfare \textit{gain} relative to the default, since that only subtracts a constant term. 
    We can also easily obtain weighted versions of utilitarian SW, which include maximizing a single player's utility as a special case. 
    However, there are also natural non-linear objective functions. 
    For example, given some beliefs $P$, maximizing one player's benefit from the SPI subject to a minimum constraint on another player's benefit is nonlinear. 
    In addition, maximizing non-utilitarian notions of social welfare, such as Nash or egalitarian social welfare, under some belief $P$ over the outcomes of the reduced game generally requires a non-linear objective. 

    There's an important %
    complication %
    regarding what it means to optimize over correlated token SPIs. 
    Valid $\remapu$ must have strict inequalities $m_i > 0$, so the space of valid $\remapu$ is not closed.
    SPIs must also be strictly Pareto improving, which is another strict inequality.
    Thus, there may not be an optimal SPI: It might be possible to get arbitrarily close to a particular objective value but not to achieve it.
    We show that we can optimize linear objectives over correlated token SPIs in the strongest sense one could hope for given this issue: 
    We can efficiently decide whether the instance admits an optimal solution.
    If so, we find the value of the optimal solution and a valid $\remapu$ achieving it. 
    If not, we find the supremum over SPIs of the objective value and a family of valid $\remapu$ whose objective values approach this supremum.

\begin{theoremrep}
    For a given game $\game$, any linear objective over valid, correlated utility remapping functions $\remapu: \u(\red A) \rightarrow \u(\C(A))$ can be optimized in polynomial time. %
\end{theoremrep}
\begin{proof}
    We'll assume that the instance admits an SPI, which can of course be decided by the approach in the previous (\enquote{complexity}) part, as otherwise this optimization problem is infeasible. 
    Note that there is some subtlety about what optimizing over these SPIs means.
    Because the space of valid $\remapu$ is not closed due to the $m_i>0$ constraints and strictly Pareto improving requirement, there may be no optimal correlated isomorphism SPI. 
    However, we show that we can decide in polynomial time whether the instance admits an optimal solution. 
    If so, we find the value of the optimal solution and a valid $\remapu$ achieving it.
    If not, we find the supremum of the objective value and a parameterized $\remapu$ whose objective value approaches the supremum.

    We do this via a linear program very similar to the one from the proof of \Cref{thm:complexity-decide-correlated-remapping-SPI}. 
    Say the linear objective has parameters $\{w_i^{v^j}\}_{i,j}$. 
    Consider the following linear program with a lexicographic objective. 

    \begin{align*}
    &\text{Maximize Lexicographically}&\quad \left( \sum_{v^j \in V} \sum_{i \in [n]} w^{v^j}_{i} (m_i v_i^j + b_i), \mu \right)&\\
    &\text{Subject to:} &\\
    & p_a^{v^j} \geq 0 &\text{for all } a \in A, j \in [k]  \\
    &\sum_{a \in A} p_a^{v^j} = 1 &\text{for all } j \in [k] \\
    & m_i \geq 0 &\text{for all } i \in [n]\\
    & m_i v_i^j + b_i \geq v_i^j &\text{for all } i \in [n], j \in [k] \\
    & \sum_{a \in A} p_a^{v^j} u_i(a) = m_i v_i^j + b_i &\text{for all } i \in [n], j \in [k] \\
    & \mu \leq m_i &\text{for all } i \in [n] \\
    & \mu \leq \sum_{v^j \in V} \sum_{i \in [n]} (m_i v_i^j + b_i - v_i^j)
\end{align*}

    Aside from the objective, the program is identical to that from the proof of \Cref{thm:complexity-decide-correlated-remapping-SPI} except it has one additional variable $\mu$.
    This $\mu$ is constrained to be weakly less than each $m_i$.
    It's also constrained to be weakly less than the objective from the previous LP, which is essentially the social welfare gain under the $\remapu$, so is positive if and only if $\remapu$ is strictly Pareto improving. 
    The objective is first maximizing the given objective $\sum_{v^j \in V} \sum_{i \in [n]} w^{v^j}_{i} (m_i v_i^j + b_i)$ and secondarily maximizing $\mu$.

    As before, the program searches over all valid $\remapu$ as well as those with some $m_i=0$ and those which are not strictly Pareto improving.
    It is polynomially sized and can therefore be solved in polynomial time. 
    The lexicographic objective can be handled, for example, by first solving the LP with the primary objective to find optimal value $o^*$, and then maximizing the secondary objective with the additional constraint that the primary objective value is at least $o^*$.

    If the program finds an optimal solution with $\mu > 0$, then the problem instance admits an optimal valid $\remapu$, one with all $m_i>0$ and which is strictly Pareto improving, which the LP finds. 
    If the program finds an optimal solution with $\mu=0$, then the problem instance does not admit an exactly optimal solution, but the LP finds the supremum of the objective values over the space of valid $\remapu$.
    Let $\remapu^*$ be the the utility remapping function represented by the values of LP variables at optimality, and let $\remapu^s$ be any strictly Pareto improving utility remapping function. 
    Such a $\remapu^s$ can be found efficiently using the LP from the proof of \Cref{thm:complexity-decide-correlated-remapping-SPI}.
    Define $\remapu'$ by $\remapu'(v) = (1-2\eps)\remapu^*(v) + \eps\remapu^s(v) + \eps v$.
    We first show that $\remapu'$ is valid.
    Observe that $\remapu'$ is feasible by the convexity of $\u(\F(A))$. 
    It is strictly Pareto improving because its the convex combination of two weakly Pareto improving functions $\remapu$ and one strictly Pareto improving one. 
    Finally, $\remapu'$ is positive affine for the same reason as before: its the convex combination of one $\remapu$ with all $m_i=1 >0$ and two $\remapu$ with all $m_i \geq0$.
    This shows that $\remapu'$ is valid.
    Lastly, by the linearity of the objective, the objective value of $\remapu'$  approaches the value of the LP solution as $\eps\rightarrow 0$, as desired. 
    \end{proof}

    In contrast to the correlated case,  
    with \textit{pure} token isomorphism SPIs, non-linear functions of $\remapu$ (e.g., Nash or egalitarian welfare) can also be optimized in quasipolynomial time. 
    In addition, the technicalities we faced there regarding optimization don't apply here: 
    The space of valid $\remapu$ into $\u(A)$ is finite, so optimal pure SPIs always exist.

\begin{theoremrep}\label{thm:complexity-optimize-pure-remapping-SPI-n-player}
    Consider a game $\game$. 
    Arbitrary (quasi-)polynomial time computable objective functions over the space of valid, pure utility remapping functions $\remapu: \u(\red A ) \rightarrow \u(A)$ can be optimized in quasi-polynomial time. 
    For any fixed number of players $n$, this is polynomial time. 
\end{theoremrep}

\begin{proof}
    This proof follows almost immediately from the proof of \Cref{thm:complexity-decide-pure-remapping-SPI-n-player}.
    Observe that \Cref{alg:Pure_Token_SPI} from that proof enumerates all valid utility remapping functions $\remapu$.
    To convert it into an optimization algorithm, one can simply modify the final for loop to track the best valid $\remapu$ found so far (according to the objective function) rather than returning \enquote{yes} when the first SPI is found. 
    Then, at the end we can simply return the best $\remapu$. 
    This optimization algorithm still runs in quasipolynomial time for any quasipolynomial time computable objective function. 
\end{proof}

\end{toappendix}

\section{Default-Remapping \SPIs}
\label{sec:default-remapping-spis}

In this section, we consider what \SPIs{} can be achieved if the players can credibly reveal their default strategy $\policy_i(\game)$ and thus \textit{ex post} verifiably commit to play according to some function $\remap$ of this default policy.
We call this default-remapping commitment and refer to $\remap$ as a (default-) remapping function. 
Unilateral default-remapping commitment involves committing to some $\remap_i: A_i \rightarrow \F(A_i)$. 
In the omnilateral case, when all players can credibly reveal their default, the players can choose a function $\remap: A \rightarrow \F(A)$ and commit to play $\remap(a)$ whenever the default policy $\policy(\game)$ results in outcome $a$. 
We'll consider the unilateral and omnilateral versions of these commitments, as well as the pure and correlated versions, where $\F(A)$ is either $A$ or $\C(A)$.

Our reuse of the $\remap$ notation highlights the relationship of omnilateral default-remapping to the token game SPIs of the previous section.  
In the token game setting, the players commit to play the strategy profile $\remap(t)$, where $t$ is the outcome of a token game $\tgame$.
In the omnilateral default-remapping setting, players commit to play the strategy profile $\remap(a)$, where $a$ is the outcome the players \textit{would have reached} had they played the original game $\game$ as usual. 

Though the ability to credibly reveal one's default policy is a strong assumption, it applies in some scenarios. %
For example, a player might intend to play a future game by copying the strategy of some public figure or taking the recommendation of a forthcoming paper. %
If this fact is common knowledge, the player could unilaterally commit to an \textit{ex post} verifiable remapping of their default action. 

The complexity of finding default-remapping SPIs depends on whether or not all players' default actions can be credibly revealed. 
As such, we'll consider these two cases separately.%

\subsection{Unilateral Default-Remapping SPIs}
\label{subsec:unilateral_default_remapping}
\begin{toappendix}
    \subsection{Proofs for \Cref{subsec:unilateral_default_remapping} (Unilateral Default-Remapping SPIs)}
    \label{sec:unilateral-default-remapping}
\end{toappendix}

We first consider the case where only a strict subset of the players can commit to a strategy remapping. For notational simplicity, we specifically assume that only Player 1 can commit to a strategy remapping $\Psi_1(\Pi_1(\tilde \game))$. %
Call the resulting interaction $\game ^{\Psi_1 \circ \Pi_1(\tilde \game)}$.
For SPI purposes, unilateral default remapping is similar to unilateral utility function commitments (although we will see some differences below and especially in \Cref{appendix:why-all-isomorphisms-are-relevant-for-unilateral-default-remapping}). Therefore, we illustrate it with the ``Complicated Temptation Game'', the same example that \citet[Table 4]{SPI-original} use to illustrate unilateral utility function SPIs, see our \Cref{table:complicated-temptation-game}.

By default, this  game reduces to its top-left quadrant, where Player 1 chooses between $T_1,T_2$ and Player 2 chooses between $C_1,C_2$. Player 1 can unilaterally Pareto-improve by committing to %
choose $R_{1}$/$R_{2}$ had she chosen $T_1/T_2$ in the default.%

\begin{table}
\begin{center}
\begin{tabular}{cc|c|c|c|c|}
  & \multicolumn{1}{c}{} & \multicolumn{4}{c}{Player 2}\\
  & \multicolumn{1}{c}{} & \multicolumn{1}{c}{$C_1$}  & \multicolumn{1}{c}{$C_2$} &   \multicolumn{1}{c}{$F_1$} &   \multicolumn{1}{c}{$F_2$}  \\\cline{3-6}
  \multirow{4}*{Player 1} & $T_1$ &  $4,2$ & $1,1$ & $6,0$ & $6,0$ \\\cline{3-6}
  & $T_2$ & $1,1$ & $2,4$ & $6,0$ & $6,0$ \\\cline{3-6}
 & $R_1$ & $0,0$ & $0,0$ & $5,3$ & $3,2$ \\\cline{3-6}
  & $R_2$ & $0,0$ & $0,0$ & $2,2$ & $3,5$ \\\cline{3-6}
\end{tabular}
\end{center}
\caption{Complicated Temptation Game}
    \label{table:complicated-temptation-game}
\end{table}

To allow for a formal analysis of unilateral default-remapping commitments, we need assumptions about outcome correspondence for interactions of the form $\game ^{\Psi_1 \circ \Pi_1(\game)}$. Specifically, we make three assumptions. The first two parallel the elimination of dominated strategies: Actions not in the image of $\Psi_1$ can be removed; and elimination of dominated actions for Players $-i$ works as before. We also need an analog of the isomorphism assumption (\Cref{assumption:isomorphism-assumption}). %
We also show (\Cref{proposition:example-unilateral-disarmament}) how the assumptions can be used to prove the SPI for the example in \Cref{table:complicated-temptation-game}.

First, we need two elimination assumptions. The first is that we can eliminate dominated actions for players \textit{other} than $i$.

\begin{assumption}
\label{assumption:remove-dominated-under-instructions}
If in $\game$ some action $\bar a_i$ of some Player $i\neq 1$ is strictly dominated, then 
\begin{equation*}
\Pi(\game^{\Psi_1 \circ \Pi_1(\game)}) %
\sim_{(\mathrm{id},\Xi_i)}
\Pi((\game-\{\bar a_i\})^{\Psi_1 \circ \Pi_1(\game)})
\end{equation*}
where $\Xi_i$ is the identity function except that it maps $\bar a_i$ to the empty set, i.e., $\Xi_i(a_i)=\{a_i\}$ whenever $a_i\neq \bar a_i$ and $\Xi_i(\bar a_i) = \emptyset $.
\end{assumption}

The second assumption is an elimination assumption for Player 1. It says that if $\Psi_1$ never maps to some action $\bar a_1$, then we can remove $\bar a_1$. This is important primarily by allowing us to apply \Cref{assumption:remove-dominated-under-instructions} more often.

\begin{assumption}\label{assumption:disarm-effective}
Let $\game,\hat\game$ be games and let $\Psi^{-1}(\hat a_1)=\emptyset$, i.e., let $a_1$ be an action that is not in the image of $\Psi_1$. %
Then
\begin{equation*}
\Pi(\hat\game^{\Psi_1 \circ \Pi_1(\game)})
\sim_{(\Xi_1,\mathrm{id})}
\Pi((\hat\game - \hat a_1)^{\Psi_1 \circ \Pi_1(\game)})
\end{equation*}
where $\Xi_1(a_1)=\{a_1\}$ for all $a_1\neq \hat a_1$ and $\Xi(\hat a_1) = \emptyset$.
\end{assumption}

Third, we need a sort of isomorphism assumption to connect interactions of the form $\Pi(\game^{\Psi_1 \circ \Pi_1(\game')})$ to interactions that are just normal form games. Roughly, the following assumption states: If P1 announces that she'll play like $\Pi(\game')$ but mapped into $\hat \game$ and moreover $\game'$ and $\hat \game$ are isomorphic under $\Psi_1,\phi_{2},...,\phi_{n}$ in terms of the other players' utilities (for some $\phi_{2},...,\phi_{n}$), then $\Pi(\hat\game^{\Psi_1 \circ \Pi_1(\game')})$ will be played isomorphically to $\game'$.

\begin{assumption}\label{assumption:player-2-isomorphism-against-play}
Let $\game'$ be a fully reduced game. Let $\hat\game$ be a game %
in which Players $-1$ have %
no strictly dominated strategies. Let $\Psi_1\colon A_1'\rightarrow \hat A_1$. Let $\phi_i\colon A_i' \rightarrow \hat A_i$ s.t.\ $(\Psi_1,\phi_2,...,\phi_n)$ is an isomorphism in terms of the other players' utilities (i.e., for each $i\neq 1$, $\phi_i$ is a bijection and there are $m_i\in \mathbb R_+,b_i\in \mathbb R$ s.t.\ $u_i \circ (\Psi_1,\phi_2,...,\phi_n) = m_i u_i + b_i$). %
Then $\game' \sim_{\Phi} \Pi(\hat\game^{\Psi_1 \circ \Pi_1(\game')})$, where $\Phi$ is the union of all isomorphisms $(\Psi_1,\phi_2,...,\phi_n)$ of the form above, i.e., $\Phi(\a')= \{ (\Psi_1,\phi_2,...,\phi_n) (\a') \mid (\Psi_1,\phi_2,...,\phi_n) \}$.
\end{assumption}

Using these assumptions, we can now formally prove that the unilateral default-conditional SPI for the Complicated Temptation Game is indeed an SPI.

\begin{proposition}
    \label{proposition:example-unilateral-disarmament}
    Let $\game$ be the game of \Cref{table:complicated-temptation-game}. Let $\red{\game}$ be the top-left quadrant of $\game$. Let $\Psi_1\colon T_1\mapsto R_1, T_2 \mapsto R_2$. 
    From \Cref{assumption:disarm-effective,assumption:elimination-of-dominated-strategies,assumption:player-2-isomorphism-against-play,assumption:remove-dominated-under-instructions}, it follows that $\Pi(\game^{\Psi_1 \circ \Pi_1(\bar\game)})$ is an SPI on $\game$.
\end{proposition}

\begin{proof}
By repeated application of the dominance assumption, we get that $\game\sim_{\Xi}\bar\game$, where $\Xi$ maps outcomes including $F$ or $R$ actions to $\emptyset$ and otherwise maps outcomes onto themselves.

Now let $\hat G$ be the bottom-right game. Consider $\phi_2\colon C_1 \mapsto F_1, C_2 \mapsto F_2$. Note that $(\Psi_1, \phi_2)$ is isomorphism for Player 2's utility (with coefficients $a=1,b=1$). Note further that $\phi_2$ thus defined is the only such function. Thus, we have that $\bar\game \sim_{(\Psi_1,\phi_2)} \Pi(\hat\game^{\Psi_1 \circ \Pi_1(\bar\game)})$.

By \Cref{assumption:disarm-effective}, we have
\begin{equation*}
\Pi(\game^{\Psi_1\circ \Pi_1(\bar \game)}) \sim \Pi ((\game - \{ T_1,T_2\})^{\Psi_1 \circ \Pi_1(\bar \game)}).
\end{equation*}
By \Cref{assumption:remove-dominated-under-instructions},
\begin{equation*}
\Pi((\game - \{T_1,T_2\}) ^ {\Psi_1 \circ \Pi_1(\bar \game)}) \sim \Pi(\hat G ^ {\Psi_1 \circ \Pi_1(\bar \game)}).
\end{equation*}

Putting it all together using the transitivity rule, we get that $\game \sim \Pi(\game^{\Psi_1\circ \Pi(\bar \game)})$. It is easy to verify that the resulting outcome correspondence is Pareto improving.
\end{proof}

One can characterize unilateral default-remapping SPIs in a way that's analogous to the characterization in \Cref{lemma:SPI-iff-isomorphism-on-reduced-game-or-degenerate}. That is, if $\Psi_1$ is a default-remapping SPI, then we can see this by first reducing $\game$ and $\game^{\Psi_1 \circ \policy_1(\game)}$ (i.e., the game under remapping $\Psi_1$)%
. We can then have two types of SPIs (simple and isomorphism SPIs): The first is that every possible outcome of the prospective SPI $\game^{\Psi_1 \circ \policy_1(\game)}$ is better than every outcome of the default $\game$. 
The second is that the reduced games are isomorphic. %
An important difference between \Cref{lemma:SPI-iff-isomorphism-on-reduced-game-or-degenerate} and the present result is that we need the condition to state that \textit{all} the isomorphisms are Pareto-improving. It does not suffice to consider one of the isomorphisms. We explain this in detail in \Cref{appendix:why-all-isomorphisms-are-relevant-for-unilateral-default-remapping}.

\begin{toappendix}
\begin{lemmarep}\label{lemma:characterization-unilateral-default-remapping-spis}
Let $\game=(A,\mathbf u)$ be a normal-from game that reduces to $\bar\game$. Then there is a unilateral default-remapping SPI on $\game$ under %
\Crefrange{assumption:elimination-of-dominated-strategies}{assumption:player-2-isomorphism-against-play}
if and only if at least one of the following two conditions holds:
(1) There is an action $a_1$ and sets of actions $\hat A_{2},...,\hat A_n$ such that $(\{a_1\},  A_2,...,  A_n, \mathbf u)$ reduces by strict dominance to $(\{a_1\}, \hat A_2,..., \hat A_n, \mathbf u)$ and $\{a_1\} \times \hat A_2 \times ... \times \hat A_n $ Pareto dominates all outcomes in $\bar\game$ with at least one strict Pareto relation.
(2) There is a subgame $\hat \game$ of $\game$ such that:
        \begin{inparaitem}[]
            \item $(\hat A_1, A_2, ..., A_n, \mathbf u)$ reduces by strict dominance to $(\hat A_1, \hat A_2, ..., \hat A_n, \mathbf u)$;
            \item there exists $\Psi_1\colon \bar A_1 \rightarrow \hat A_1$ and $\iso_i\colon \bar A_i\rightarrow \hat A_i$ s.t.\ $(\Psi_1,\iso_2,...,\iso_n)$ is an isomorphism from $\bar G$ to $\hat G$ in terms of the utilities of Players $2,...,n$; and
            \item for all such $\iso_2,...,\iso_n$, the isomorphism $(\Psi_1,\iso_2,...,\iso_n)$ is Pareto improving.
        \end{inparaitem}
\end{lemmarep}

\begin{proof}
    \underline{$\Rightarrow$:} %
    We prove this direction by proving its contrapositive. That is, we show that if the conditions aren't satisfied, then there can be no unilateral default-remapping SPI.

    Let $\game^{\Psi_1 \circ \Pi_1(\game')}$ be a proposed unilateral disarmament. We need to show that there is an assignment of outcomes to games (including games under remapping instructions) s.t.\ the outcome assigned to $\game^{\Psi_1 \circ \Pi_1(\game')}$ is not Pareto-better than $\game$.
    
    Now let $\tilde \game ^ {\Psi_1 \circ \Pi_1(\game ')}$ be the game resulting from $\game ^ {\Psi_1 \circ \Pi_1(\game ')}$ by repeated elimination of disarmed actions for Player 1 and dominated actions for Players $-1$ as per \Cref{assumption:disarm-effective,assumption:remove-dominated-under-instructions}.%

    We distinguish three cases:
    \begin{enumerate}
        \item \label{proof-of-lemma-5-1-item:1} $\game'$ is not isomorphic to $\tilde \game$ w.r.t.\ Player $-1$'s utilities or $\game'$ is not fully reduced.
        \item \label{proof-of-lemma-5-1-item:2} $\game'$ is fully reduced and is isomorphic to $\tilde \game$ in terms of Player $-1$'s utilities, but $\game'$ is not isomorphic to $\bar G$.
        \item \label{proof-of-lemma-5-1-item:3} $\game'$ is fully reduced and is isomorphic to $\tilde \game$ in terms of Player $-1$'s utilities, and $\game'$ is also isomorphic to $ \bar \game$.
    \end{enumerate}

    We now consider these in turn (in the order 2, 1, 3).

    \ref{proof-of-lemma-5-1-item:2}: %
    By the assumption that condition 2 doesn't hold, 
    one of the following two holds:
    \begin{itemize}
        \item There is an outcome $\tilde \a$ in $\tilde \game$ and an outcome in $\bar \a$ in $\bar \game$ s.t.\ $\tilde \a$ is not even weakly Pareto-better than $\bar \a$.
        \item All outcomes of $\tilde \game$ and $\bar \game$ have the same utility.
    \end{itemize}
    In the second case, it's clear that $\game^{\Psi_1 \circ \Pi_1(\game')}$ cannot be an SPI on $\game$. (It fails the strictness condition.)
    
    So let's consider the first case. We will show that we can construct an assignment $\Pi$ of outcomes to games that violates the SPI claim.
    First assign $\bar \a$ to $\bar \game$ and $\tilde \a$ to $\game^{\Psi_1 \circ \Pi_1(\game')}$ and $\tilde\game^{\Psi_1 \circ \Pi_1(\game')}$.
    Now we have left to show that we can complete the assignment in a way that satisfies all the assumptions.

    We proceed with the construction as follows. Note that game isomorphism forms an equivalence relationship on games. Thus, we can partition games into sets of isomorphic games. So consider each set of isomorphic games $\Gamma$ %
    and the set of associated $\hat \Gamma ^ {\Xi_1 \circ \Pi (\Gamma)}$, where $\hat\Gamma$ is isomorphic in terms of Player $-1$'s utilities to $\Gamma$ via $\Xi_1$ (and some $\phi_2,...,\phi_n$).

    Now if $\game'$ is in this class of games, then we assign to $\game'$ an outcome $(\Psi_1,\phi_2,...,\phi_n)^{-1}(\tilde \a)$ (where $\phi_2,...,\phi_n$ are s.t.\ $(\Psi_1,\phi_2,...,\phi_n)$ is an isomorphism). Then assign outcomes in the rest of the class (of games isomorphic to $\game'$) to be isomorphic to this outcome. Finally, assign outcomes to the $\hat \Gamma ^ {\Xi_1 \circ \Pi (\Gamma)}$ by picking any isomorphism functions $\phi_2,...,\phi_n$ and applying $(\Xi_1,\phi_2,...,\phi_n)$ to the outcome assigned to $\Gamma$.

    If $\bar \game$ is in the class of games, then we similarly ``expand'' our assignment from the assignment of $\bar a$ to $\bar \game$.

    Otherwise, we simply assign any outcome to one of the games and then proceed as above.

    It is easy to verify that this assignment satisfies the assumptions.

    \ref{proof-of-lemma-5-1-item:1}: This case can be handled just like case \ref{proof-of-lemma-5-1-item:2}. The only difference is that none of the reduced game classes are constrained by having assigned $\tilde \a$ to $\game^{\Psi_1 \circ \Pi_1(\game')}$ and $\tilde\game^{\Psi_1 \circ \Pi_1(\game')}$.

    \ref{proof-of-lemma-5-1-item:3}: By assumption, there is an isomorphism between $\bar\game$ and $\game'$. Let $\phi^{\bar\game,\game'}$ be one such isomorphism.
    We also know that there are $\phi_2^{\game', \tilde \game},...,\phi_n^{\game', \tilde\game}$ s.t.\ $(\Psi_1,\phi_2^{\game',\tilde\game},...,\phi_n^{\game',\tilde\game})$ is an isomorphism for Players $-1$ from $\game'$ to $\tilde G$. Thus, we also get that $(\Psi_1\circ \phi^{\bar\game,\game'}_1,\phi_2^{\game', \tilde \game}\circ\phi^{\bar\game,\game'}_2,...,\phi_n^{\game', \tilde \game}\circ\phi^{\bar\game,\game'}_n)$ is an isomorphism in terms of utilities of Players $-1$ from $\bar \game$ to $\tilde\game$.

    By the assumption that the second condition in the lemma is violated, we know that there must be $\phi_2^{\bar \game, \tilde \game},...,\phi_n^{\bar \game, \tilde G}$ s.t.\ $(\Psi_1\circ \phi^{\bar\game,\game'}_1,\phi_2^{\bar\game, \tilde \game},...,\phi_n^{\bar\game, \tilde \game})$ is an isomorphism in terms of Players $-1$'s utilities between $\bar G$ and $\tilde G$ and is not Pareto-improving.

    Again, this can mean two things:
    \begin{itemize}
        \item $(\Psi_1\circ \phi^{\bar\game,\game'}_1,\phi_2^{\bar\game, \tilde \game},...,\phi_n^{\bar\game, \tilde \game})$ keeps all utilities constant.
        \item There is $\bar \a$ s.t.\ $\tilde \a=(\Psi_1\circ \phi^{\bar\game,\game'}_1,\phi_2^{\bar\game, \tilde \game},...,\phi_n^{\bar\game, \tilde \game})(\bar{\a})$ is not at least as good for all players as $\bar \a$.
    \end{itemize}

    Consider the first case. We have $\bar \game \sim_{\Phi^{\bar \game, \game'}} \game'$ by \Cref{assumption:isomorphism-assumption} and $\game' \sim _ {(\Psi_1, \Phi_{-1}^{\game',\tilde \game})} \tilde \game ^ {\Psi_1 \circ \Pi_1(\game')}$. The difficulty is that $\Phi^{\bar \game, \game'}$ and $ \Phi_{-1}^{\game',\tilde \game}$ may contain multiple isomorphisms and we only know \textit{one} of these to keep the utilities constant. Of course, if \textit{all} keep utilities constant, then we're done. But now note that all the functions aggregated in $\Phi^{\bar \game, \game'}$ are bijections. Further note that all the functions aggregated in $(\Psi_1, \Phi_{-1}^{\game',\tilde \game})$ are bijections up to grouping Player 1 actions that are mapped identically by $\Psi_1$ (i.e., $a_1,a_1'$ s.t.\ $\Psi_1(a_1)=\Psi_1(a_1')$). Additionally all the outcomes grouped by $\Psi_1$ must have the same utility. It follows that if one of the aggregated functions (namely $(\Psi_1\circ \phi^{\bar\game,\game'}_1,\phi_2^{\bar\game, \tilde \game},...,\phi_n^{\bar\game, \tilde \game})$) keeps all utilities constant, then any other of the functions must either also keep utilities constant or must increase some utilities and decrease other utilities. (This is because all the functions must have the same average effect on utilities, by virtue of the functions being (essentially) bijections. So if any of the functions had a (Pareto-)improving effect on some outcomes, it must have a negative effect on other utilities to compensate.) We can then address this in the way we are addressing the second case.

    Now let's consider the second case.
    Assign $\bar \a$ to $\bar\game$ and $\game$ and $\tilde \a$ to $\tilde\game^{\Psi_1 \circ \Pi_1(\game')}$ and $\game^{\Psi_1 \circ \Pi_1(\game')}$.
    Next, assign $\a'$ to $\game'$, where $a_i'=\phi_i^{\bar \game, \game'}(\bar a_i)$. Clearly, this satisfies the isomorphism assumption between $\bar\game$ and $\game'$. 

    We now show that $\a'$ thus chosen also satisfies the assumption between $\game'$ and $\tilde\game^{\Psi_1 \circ \Pi_1(\game')}$. That is, we need to show that there are $\phi_2',...,\phi_n'$ s.t.\ $(\Psi_1,\phi_2',...,\phi_n')$ is an isomorphism in terms of Player $-1$'s utilities from $\game'$ to $\tilde\game$ and $(\Psi_1,\phi_2',...,\phi_n')(\a')=\tilde \a$. (Note that this isomorphism may well be Pareto-improving if we make \Cref{assumption:isomorphism-assumption}. For instance, $\game'$ may arise from $\game$ by adding a large constant to all utilities.)

    Choose $\phi_i'=\phi_i^{\bar \game,\tilde\game} \circ (\phi_i^{\bar\game,\game'})^{-1}$ for $i=2,...,n$. 
    We then have that
    \begin{equation*}
        \Psi_1 (a_1') = \Psi_1(\phi_1^{\bar \game, \game'}(\bar a_1)) = \tilde a_1.
    \end{equation*}
    Further, for $i=2,3...$ we have
    \begin{eqnarray*}
        \phi_i' (a_i') &=& \phi_i' ( \phi_i^{\bar \game, \game'} (\bar a_i))\\
        &=& \phi_i^{\bar \game,\tilde\game}((\phi_i^{\bar\game,\game'})^{-1}( \phi_i^{\bar \game, \game'} (\bar a_i)))\\
        &=& \phi_i^{\bar \game,\tilde\game}(\bar a_i)\\
        &=& \tilde a_i.
    \end{eqnarray*}
    
    Finally, we need to show that $(\Psi_1,\phi_2',...,\phi_n')$ is an isomorphism in terms of the utilities of Players $-1$. Recall that the $(\phi_i^{\bar \game, \game'})$ form an isomorphism from $\bar \game $ to $\game'$. Thus $(\phi_i^{\bar \game, \game'})^{-1}$ form an isomorphism from $\game'$ to $\bar \game$ and in particular an isomorphism in terms of the utilities in terms of Players $-1$. Next note that $(\Psi_1 \circ \phi_1^{\bar \game, \game'}, \phi_2^{\bar\game,\tilde \game}, ..., \phi_n^{\bar\game,\tilde \game})$ is an isomorphism in terms of Players $-1$'s utilities from $\bar \game$ to $\tilde\game$. It follows that the composition of the two is an isomorphism in terms of Players $-1$'s utilities from $\game'$ to $\tilde\game$. Note that this composition is $(\Psi_1,\phi_2',...,\phi_n')$.

    We complete the assignment as in case \ref{proof-of-lemma-5-1-item:2}.

    \underline{$\Leftarrow$}:
    We need to show that each of the two conditions suffices to imply the existence of an SPI.
    
    Let's first consider the case that the first condition holds. Then consider $\Psi_1 = a_1$, i.e., the function that maps everything onto $a_1$, and the prospective SPI $\game ^ {\Psi_1 \circ \Pi_1(\bar G)}$.
    
    By \Cref{assumption:disarm-effective},
    \begin{equation*}
        \game^{\Psi_1\circ \Pi_1(\bar\game)}
        \sim
        (\{a_1\}, A_2,...,A_n, \mathbf{u})^{\Psi_1\circ \Pi(\bar\game)}
    \end{equation*}
    By \Cref{assumption:remove-dominated-under-instructions},
    \begin{equation*}
        (\{a_1\}, A_2,..., A_n, \mathbf{u})^{\Psi_1\circ \Pi(\bar\game)}
        \sim
        (\{a_1\}, \hat A_2,..., \hat A_n, \mathbf{u})^{\Psi_1\circ \Pi(\bar\game)}
    \end{equation*}

    By \Cref{assumption:elimination-of-dominated-strategies}, we have $\game \sim \bar\game$. Finally, we have
    \begin{equation*}
        \bar{\game}\sim_{\mathrm{all}} (\{a_1\}, \hat A_2,..., \hat A_n, \mathbf{u})^{\Psi_1\circ \Pi(\bar\game)},
    \end{equation*}
    where $\mathrm{all}$ is the trivial correspondence.
    
    By transitivity and reflexivity, we thus obtain $\game \sim \game ^ {\Psi_1 \circ \Pi_1(\bar G)}$ and it is easy to see that the mapping is Pareto-improving. To prove that $\game ^ {\Psi_1 \circ \Pi_1(\bar G)}$ is an SPI on $\game$, we also need to prove that there exists an assignment $\Pi$ of outcomes to games that assigns a strictly Pareto-better outcome to $\game ^ {\Psi_1 \circ \Pi_1(\bar G)}$ than to $\game$. This $\Pi$ can be constructed analogously to the above construction.

    The proof for the second condition works in the same way. The only difference is that instead of the ``$\mathrm{all}$ assumption'', we need to use \Cref{assumption:player-2-isomorphism-against-play}.
\end{proof}
\end{toappendix}

We can now prove our main result about the complexity of finding unilateral default-remapping SPIs.

\begin{theoremrep}\label{theorem:unilateral-default-conditional-SPI-NP-hard}
    Deciding whether a game admits a unilateral default-remapping action remapping SPI (for Player $1$) under \Cref{assumption:disarm-effective,assumption:elimination-of-dominated-strategies,assumption:player-2-isomorphism-against-play,assumption:remove-dominated-under-instructions} is \NPH, even for two players.
\end{theoremrep}

As usual, the key difficulty is finding some part of the full game that is isomorphic to the original game. However, contrary to the other proof, only Player 2's utilities are relevant. Thus, (in the two-player case) we cannot straightforwardly reduce from any of the subgame isomorphism problems that we consider in our other proofs (see \Cref{appendix:game-isomorphism-graph-isomorphism-complete}).
Instead, we will directly reduce from the subgraph isomorphism problem.
Nonetheless, the proof is a straightforward adaption of earlier proofs. In fact, \citeauthor{SPI-original}'s (\citeyear{SPI-original}) construction in the proof of their NP-hardness result can be used to show \Cref{theorem:unilateral-default-conditional-SPI-NP-hard}. To be complete and self-contained, we show how the proof of our \Cref{thm:disarmament-SPI-NP-complete} can be adapted to prove \Cref{theorem:unilateral-default-conditional-SPI-NP-hard}.

\begin{toappendix}
We now formally define directed graphs and subgraph isomorphisms for directed graphs for use in the proof of \Cref{theorem:unilateral-default-conditional-SPI-NP-hard}.
A \textit{directed graph} (as represented by an adjacency matrix) is a pair of a natural number $n$ (representing the number of nodes) and a function
\begin{equation*}
\mathrm{adj}\colon \{1,...,n\}\times \{1,...,n\} - \{(1,1), ..., (n,n)\} \rightarrow \{0,1\}
\end{equation*}
that maps each pair of nodes $(i,j)$ onto $1$ if there's an edge from $i$ to $j$ and $0$ otherwise.

A \textit{subgraph isomorphism} from $(n,\mathrm{adj})$ to $(n',\mathrm{adj}')$ is an injection $\phi\colon \{1,...n\} \rightarrow \{1,...,n'\}$ s.t.\ for all $i,j\in \{1,...,n\}$ with $i\neq j$ we have $\mathrm{adj}(i,j)=\mathrm{adj}'(\phi(i),\phi(j))$. %

\begin{lemma}[\citealt{cook1971complexity}; 
{\citealt[][Sect.\ A1.4, Problem GT48]{gary1979computers}}]
\label{lemma:subgraph-iso-hard}
    The following problem is NP-complete: Given two directed graphs $\mathrm{adj},\mathrm{adj}'$, decide whether there is a subgraph isomorphism from $\mathrm{adj}$ into $\mathrm{adj}'$.
\end{lemma}

Note also that the clique problem is a special case of the subgraph isomorphism problem and the clique problem is well known to be NP-complete \cite{karp1972reductibility}. 
\end{toappendix}

\begin{proof}[Proof of \Cref{theorem:unilateral-default-conditional-SPI-NP-hard}]
We reduce from the subgraph isomorphism problem. So let $\mathrm{adj}$ and $\mathrm{adj}'$ be adjacency matrices for graphs of $n$ and $n'$ nodes, respectively. %

Then consider the game in \Cref{table:construction-for-proof-of-thm:disarmament-SPI-NP-complete-new} from the proof of \Cref{thm:disarmament-SPI-NP-complete}, where we let $A_1=A_2=\{1,...,n\}$ and $A_1'=A_2'=\{1,...,n'\}$ and for $\game$ we insert a game in which both players receive a payoff of $\mathrm{adj}$ and for $\game'$ we insert a game in which both players receive a payoff of $\mathrm{adj}'$. On the diagonal, we insert payoffs $(1.5,1.5)$. The same argument as in the proof of \Cref{thm:disarmament-SPI-NP-complete} shows that there is an SPI if and only if there is subgame isomorphism from the bottom-left to the bottom-right corner and in particular one that maps $\mathrm{adj}$ into $\mathrm{adj}'$. We here only give a few notes on the differences between the proofs.

$\Leftarrow$: For constructing the SPI, let $\phi$ be the subgraph isomorphism. Then consider $\Psi_1\colon (T,a_1)\mapsto (R,\phi(a_1)), (T,a_2)\mapsto (R,\phi(a_2))$. It's easy to see that the resulting remapping is an SPI.

$\Rightarrow$: The argument works as before. It is further easy to verify that this isomorphism must map the actions for Players 1 and 2 identically (i.e., maps Player $1$'s action $(T,i)$ and Player $2$'s $(D,i)$ onto some $(R,j)$ and $(P,j)$ for some $j$) so that the isomorphism induces a subgraph isomorphism.
\end{proof}

\begin{toappendix}
It's unclear whether deciding the existence of a unilateral remapping SPI is also \textit{in} NP (and thus NP-complete) in full generality. The problem is that as per \Cref{lemma:characterization-unilateral-default-remapping-spis} (and the notes in \Cref{appendix:why-all-isomorphisms-are-relevant-for-unilateral-default-remapping}), we need to verify not just that there is one Pareto--improving isomorphism, but there doesn't exist also exist another isomorphism that is not Pareto improving. Under this formulation, the problem thus looks more like a member of 2QBF (sometimes also called $\mathrm{NP}^\mathrm{NP}$) \cite{balabanov20162qbf}, which is expected to be much harder.  That said, if we assume that the reduced game doesn't have action symmetries (considering only the utility functions of Players $-i$), then the problem does immediately come to be in NP, because we then only need to find one isomorphism.
\end{toappendix}

\begin{toappendix}
\subsubsection{Why \Cref{lemma:characterization-unilateral-default-remapping-spis} needs to consider all isomorphisms}
\label{appendix:why-all-isomorphisms-are-relevant-for-unilateral-default-remapping}

\begin{table}[]
    \centering
    
        \begin{tabular}{|c|c|c|c|c|c|}
        \hline
        $-5,-5$ & $-5,-5$ & $-5,-5$ & $-3,-3$ & $-3,-3$ & $0,5$ \\
        \hline
        $-5,-5$ & $-5,-5$ & $-5,-5$ & $2,1$ & $4,1$ & $-3,-3$ \\
        \hline
        $-3,-3$ & $-3,-3$ & $0,5$ & $\phantom{-}5,-5$ & $\phantom{-}5,-5$ & $\phantom{-}5,-5$ \\
        \hline
        $1,1$ & $3,1$ & $-3,-3$ & $\phantom{-}5,-5$ & $\phantom{-}5,-5$ & $\phantom{-}5,-5$ \\
        \hline
    \end{tabular}
    \caption{A game to illustrate why the second condition in \Cref{lemma:characterization-unilateral-default-remapping-spis} is over \textit{all} isomorphisms, as opposed to the existence of one isomorphism.}
    \label{tab:why-all-isomorphisms}
\end{table}

Note that in \Cref{lemma:characterization-unilateral-default-remapping-spis} the condition states that the union of isomorphisms between the default and the new game must be Pareto-improving. Perhaps it's enough to require that there is \textit{one} Pareto-improving isomorphism? In particular,
note that our \Cref{lemma:SPI-iff-isomorphism-on-reduced-game-or-degenerate} only requires the existence of a single Pareto-improving isomorphism and \cite[][Lemma 4]{SPI-original},
too, show that in their setting it's sufficient to find one Pareto-improving isomorphism. In both cases, it is shown that if one isomorphism is Pareto improving, all are.

Unfortunately, the same is not true in the case of unilateral default-remapping SPIs. Consider the game in \Cref{tab:why-all-isomorphisms}. This game reduces by strict dominance to its bottom-left $2$-by-$3$ subgame. One might think that Player 1 can unilaterally Pareto-improve by remapping the third and fourth rows onto the first two rows. Call this $\Psi_1$. By dominance, Player 2 will then choose from the right-most three rows. Now there are two bijections $\phi_2$ from the left three columns to the right three columns that make $(\Psi_1,\phi_2)$ an isomorphism in terms of Player 2's utilities: one that maps the first onto the third, and the second onto the fourth column; and one that maps the first onto the fourth and the second onto the third column. The first of the two is Pareto-improving, but the second is not.

At a high-level, the problem is that our isomorphism assumption essentially ignores Player 1's utilities in the target game. We assume that when Player 2 decides whether to play the third or the fourth column against Player 1's commitment to play according to $\Psi_1$, Player 2 does not take into account Player 1's utilities. (Or more precisely, we do \textit{not} assume that Player 2 \textit{does} take Player 1's utilities into account.) So we assume that it's possible (or: we don't assume that it's impossible) that Player 2 tie-breaks, say, in favor of Player 1 in the default, but doesn't tie-break in Player 1's favor in the new game.

This contrasts with our regular isomorphism assumption (\Cref{assumption:isomorphism-assumption}, shared with \cite{SPI-original}).
Under this assumption, the outcome correspondences are constrained by all players' utilities.

Interestingly, this can mean that under the assumptions as stated, some seemingly weaker forms of commitment can be more powerful than unilateral default-conditional utility. For instance, in the game in \Cref{tab:why-all-isomorphisms}, unilateral disarmament of the third and fourth row is an SPI, precisely because it keeps Player 1's utilities in play.

Similarly, a unilateral utility function-based commitment, as studied in \cite{SPI-original}, 
can achieve an SPI in \Cref{tab:why-all-isomorphisms}.

In the specific game of \Cref{tab:why-all-isomorphisms}, one might argue that Player 2 \textit{should} take Player 1's utilities into account analogously. But we can consider a version of \Cref{tab:why-all-isomorphisms} where we slightly perturb Player 1's utilities in the top-right corner. The unilateral utility function-based commitment of the earlier paper can still achieve a safe Pareto improvement by specifying that the choice between the third and fourth column should be treated in the same way as the choice between the first two columns. In contrast, a default-remapping function simply has no way to specify anything of this sort. Therefore, we think that there's no natural variant of \Cref{assumption:player-2-isomorphism-against-play} under which default-conditional commitments become more similar to unilateral utility function commitments.

\end{toappendix}

\subsection{Omnilateral Default-Remapping SPIs}
\label{subsec:multilateral-default-remapping-spis}

\begin{toappendix}
    \subsection{Proofs for \Cref{subsec:multilateral-default-remapping-spis} (Omnilateral Default-Remapping SPIs)}
\end{toappendix}

Now, we consider the case where all players can commit to strategies as a function of the default outcome of the game. %
In this case, the players' default-remapping commitment $\remap$, along with the default policy, fully determines the outcome of the game. 
That is, $G \sim_{\remap} G^{\remap}$.
Because of this, we don't need to reason about outcome correspondence, or strategic dynamics in general, when proving SPIs; SPIs occur whenever the remapping function is Pareto improving. 
\cameraReady{Consequently, finding default-remapping SPIs requires only finding an outcome in the reduced game which can be feasibly Pareto improved.}

\begin{theoremrep}
    Suppose the players can make omnilateral commitments to remap outcomes of the default policy to any feasible $\F(A)$ strategy profile. 
    A default-remapping SPI exists under \Cref{assumption:elimination-of-dominated-strategies} if and only if there exists an outcome in $\red A$ which is Pareto sub-optimal in $\F(A)$. 
    For both pure and correlated default-remapping, it can be decided in polynomial time whether an SPI exists. 
\end{theoremrep}
\begin{proof}
    We divide the proof into three subclaims.
    \Cref{lemma:default-remapping-SPI-iff-Pareto-improvable-outcome} proves the claim in the second sentence of the theorem. 
    Given this, \Cref{thm:omnilateral_correlated_default-remapping_SPI} proves the claim in the last sentence for correlated default-remapping SPIs, and \Cref{thm:omnilateral_pure_default-remapping_SPI} proves the claim in the last sentence for pure default-remapping SPIs. 
\end{proof}

The setting and result of correlated default-remapping part of the theorem above is related to that of ``SPIs under improved coordination" from \cite{SPI-original}. 
Our setting is more restrictive on the ability of the players to commit, but allows the same class of SPIs to be achieved.
(In our language, they essentially allow the players to make an omnilateral default-remapping commitment to remap the default outcome of an arbitrary game into $\C(A)$, while we only allow remapping the default of the original game.)%

As in \Cref{sec:token-game-SPIs}, we additionally consider \textit{optimizing} over omni-lateral default-remapping SPIs. %
For details, see \Cref{appendix:subsec:optimizing_default-remapping_SPIs}.

\begin{toappendix}

\begin{lemmarep}\label{lemma:default-remapping-SPI-iff-Pareto-improvable-outcome}
    Suppose the players can make omnilateral commitments to remap outcomes of the default policy to any feasible $\F(A)$ strategy profile. %
    A default-remapping SPI exists under \Cref{assumption:elimination-of-dominated-strategies} if and only if there exists an outcome in $\red A$ which is Pareto sub-optimal in $\F(A)$. 
\end{lemmarep}
\begin{proof}
    As we observed in the main text, the result of the game $\game'$ resulting from an omnilateral default-remapping commitment $\remap$ is fully determined by $\remap$ and the default policy profile $\policy(G)$. %
    Hence, $G^{\remap}$ is an SPI on $G$ if and only if for all outcomes $a \in \red A$, $\u(\remap(a)) \pgeq \u(a)$ and there exists $a \in \red A$ such that $\u(\remap(a)) \pge \u(a)$.

    If all outcomes $a \in \red A$ are Pareto optimal in the feasible set, then clearly no strict SPI can exist. 
    If there does exist an outcome $a \in \red A$ which can be strictly Pareto improved in the feasible set, then assigning $\remap(a) \pge a$ to be such a strictly Pareto improving element of the feasible set and $\remap(a') = a'$ for all other outcomes $a' \in \red A$ constitutes an SPI. 
\end{proof}

\begin{lemmarep}\label{thm:omnilateral_correlated_default-remapping_SPI}
    It can be decided in polynomial time whether an $n$-player game $\game$ admits a omnilateral default-remapping SPI into correlated strategies. 
    Furthermore, any linear objective over such SPIs can be optimized efficiently. 
\end{lemmarep}

\begin{proof}
    First, we show that the existence of default-remapping SPIs can be decided in polynomial time. 
    By \Cref{lemma:default-remapping-SPI-iff-Pareto-improvable-outcome}, a strict default-remapping SPI on $\game$ exists if and only if there exists a outcome in $\red A$ that can be strictly Pareto improved in $\C(A)$. 
    We show that checking this condition is equivalent to checking whether the optimal objective value of the following polynomially sized linear program is strictly positive.%

\begin{align*}
    &\text{Maximize}\quad \sum_{i \in [n]} \sum_{\red{a} \in \red{A}} \sum_{a \in A} \left[ p_a^{\red{a}} u_i(a) - u_i(\red{a}) \right] &\\
    &\text{Subject to:} &\\
    & p_a^{\red{a}} \geq 0 &\text{for all } a \in A, \red{a} \in \red{A}  \\
    &\sum_{a \in A} p_a^{\red{a}} = 1 &\text{for all } \red{a} \in \red{A} \\
    &\sum_{a \in A} p_a^{\red{a}} u_i(a) \geq u_i(\red{a}) &\text{for all } i \in [n], \red{a} \in \red{A}\\ 
\end{align*}

    The variables $\{ p_a^{\red a }\}$ for each $\red a \in \red A$ correspond to the distribution over $\C(A)$ corresponding to $\remap(\red a)$.
    The utilities and sets of outcomes are parameters of the problem instance, i.e. constants from the perspective of the LP. 
    It's easy to verify that the program is indeed linear. 
    The LP is polynomially sized, with $|A|^2$ variables and $O(|A|^2)$ constraints, and hence can be solved or optimized in polynomial time.

    The first two constraints ensure that each $\{p_a^{\red{a}}\}$ is a valid probability distribution over outcomes, i.e. correlated strategy profile. 
    The third ensures that the expected utility of the each $\{p_a^{\red{a}}\}$ (weakly) Pareto improves on $\red{a}$.
    The objective is, summing over all players and outcomes in $\red A$, the player's expected utility gain from playing the remapped strategy $p_a^{\red{a}}$ profile rather than the original outcome $\red a$.
    Hence, there is an SPI if and only if the objective value can be strictly greater than $0$.

\end{proof}

\begin{lemmarep}\label{thm:omnilateral_pure_default-remapping_SPI}
    It can be decided in polynomial time whether an $n$-player game $\game$ admits a strict, omnilateral default-remapping SPI into pure strategies.
\end{lemmarep}

\begin{proof}
    By \Cref{lemma:default-remapping-SPI-iff-Pareto-improvable-outcome}, it suffices to check whether there exists an outcome in $\red A$ which can be Pareto improved. 
    This can trivially be done in polynomial time by checking, for each $\red a \in \red A$, whether any of the outcomes in $A$ is strictly Pareto improving over $\red a$. 
\end{proof}

\subsection{Optimizing over Omnilateral Default-Remapping SPI}\label{appendix:subsec:optimizing_default-remapping_SPIs}

As in with the token SPIs from \Cref{sec:token-game-SPIs}, we'll consider optimizing over default-remapping SPIs on a given game in addition to deciding whether they exist. 
We show that linear objectives can be optimized efficiently.
Linear objectives over default-remapping SPIs are defined very similarly to linear objectives over token SPIs, except that they operate on the payoffs induced by the default-remapping function $\remap$. %
A \textit{linear} objective is characterized by a linear $f^{\red a}: \R^{n} \rightarrow \R$ for each each outcome $\red a \in \red A$, such that $f(\remap) = \sum_{\red a} f^{\red a}(\u(\remap(\red a)))$. %
The notion of linearity for $f^{\red a}$ is the same as in the previous section, and so as before, the class of linear objectives includes utilitarian social welfare (gain) and its weighted and subjective variants. 

Optimizing linear objectives is easy in the case of both pure and correlated default-remapping SPIs.
As was the case with correlated token SPIs, the requirement that SPIs be strictly Pareto improving causes the space of correlated default-remapping SPIs to be open, complicating the definition of optimization. 
We address this in the same way as before.

\begin{theoremrep}
    Any linear objective over omnilateral default-remapping SPIs into correlated strategies can be optimized in polynomial time. 
\end{theoremrep}

\begin{proof}
    We'll assume that default-remapping SPIs exist (which can be checked efficiently \Cref{thm:omnilateral_correlated_default-remapping_SPI}), as otherwise this optimization problem is undefined.
    As was the case with correlated token SPIs (\Cref{thm:complexity-decide-correlated-remapping-SPI}), there is some subtlety about what optimization means here.
    This is because the definition of SPIs requires them to be strictly Pareto improving, which is implicitly a strict inequality that makes the space of SPIs open. 
    (For example, there is no optimal SPI for the objective of minimizing the players' total utility gain.)
    As before, we show that we can decide in polynomial time whether the instance admits an
    optimal solution.
    If so, we find the value of the optimal solution and a $\remap$ achieving it. 
    If not, we find the supremum of the objective value and a parameterized $\remap$ whose objective value approaches the supremum.

    Let $\{w_i^{\red a}\}$ be the weights of some linear objective over default-remapping functions $\remap$.
    Consider the following lexicographic linear program. 

\begin{align*}
    &\text{Maximize Lexicographically}\quad &\left(\sum_{\red a \in A} \sum_i w_i^{\red a} \sum_{a \in A} p_a^{\red{a}} u_i(a), \mu \right) &\\
    &\text{Subject to:} &\\
    & p_a^{\red{a}} \geq 0 &\text{for all } a \in A, \red{a} \in \red{A}  \\
    &\sum_{a \in A} p_a^{\red{a}} = 1 &\text{for all } \red{a} \in \red{A} \\
    &\sum_{a \in A} p_a^{\red{a}} u_i(a) \geq u_i(\red{a}) &\text{for all } i \in [n], \red{a} \in \red{A}\\ 
    & \mu \leq \sum_{v^j \in V} \sum_{i \in [n]} (m_i v_i^j + b_i - v_i^j)
\end{align*}

    The variables $\{ p_a^{\red a }\}$ for each $\red a \in \red A$ correspond to the distribution over $\C(A)$ corresponding to $\remap(\red a)$.
    The utilities, sets of outcomes, and weights $w_i^{\red a}$ are instances of the problem instance, i.e. constants from the perspective of the LP. 
    It's easy to verify that the program is indeed linear. 
    As before, the lexicographic objective can be handled by first solving the LP with the first objective, finding the optimal value $o^*$, and then solving the LP again with the second objective and the additional constraint that the value of the first objective is at least $o^*$.
    The LP is polynomially sized, with $|A|^2$ variables and $O(|A|^2)$ constraints, and hence can be optimized in polynomial time.
    
    The first part of the objective is the linear objective represented by $\{w_i^{\red a}\}$ over the remapping function $\remap$, as $\u(\remap(\red a))$ is defined by $\sum_{a \in A} p_a^{\red{a}} u_i(a)$. 
    The first two constraints ensure that each $\{p_a^{\red{a}}\}$ is a valid probability distribution over outcomes, i.e. correlated strategy profile. 
    The third ensures that the expected utility of the each $\{p_a^{\red{a}}\}$ (weakly) Pareto improves on $\red{a}$.
    The final constraint ensures that $\mu$ is upper bounded by the sum, over all players and outcomes in $\red A$, of the players' expected utility gains from playing the remapped strategy $p_a^{\red{a}}$ profile rather than the original outcome $\red a$.
    As such, $\mu$ can be strictly positive if and only if the default-remapping represented by the LP variables is strictly Pareto improving.

    Therefore, the linear program searches over all possible default-remapping SPIs, as well as default-remapping functions which are not strictly Pareto improving and so are technically not SPIs. 
    If the linear program admits an optimal solution where $\mu >0$, then the problem instance admits an optimal SPI, which we return. 
    Otherwise, if the program returns an optimal solution where $\mu =0$, then there is no optimal SPI. 
    In this case, let $\remap^*$ be the default-remapping represented by the LP variables, which is not strictly Pareto improving.
    Let $\remap^s$ be any default-remapping SPI, which in particular is strictly Pareto improving and can be found efficiently using the LP from the previous part of the proof. 
    Then we claim the default-remapping function defined by $\remap(\red a) = (1-\eps)\remap^*(\red a) + \eps \remap^s(a)$ is the desired SPI. 
    It is feasible by the convexity of $\F(A)$ and strictly Pareto improving because its the convex combination of the strictly Pareto improving $\remap^s$ and the weakly Pareto improving $\remap^*$. 
    Hence, it is in fact a valid default-remapping SPI. 
    Finally, as $\eps \rightarrow 0$, since the objective is linear, its objective value approaches that of $\remap^s$, i.e. the supremum over all SPIs, as desired. 
\end{proof}

\begin{theoremrep}\label{thm:optimizing_pure_default-remapping_SPIs}
    Any linear objective over omnilateral default-remapping SPIs into pure strategies can be optimized in polynomial time. 
\end{theoremrep}
\begin{proof}
    Linear objectives over these SPIs can be optimized efficiently because they can be optimized greedily by outcome. 
    That is, an optimal SPI according to a linear objective $f$ can be found by iterating over all outcomes $a \in A$ and assigning $\remap(\red a)$ to be a Pareto improving outcome in $A$ that maximizes $f_a(\remap(\red a))$. 
    This can be done in polynomial time because there are $|A|$ outcomes and thus $|A|$ possibilities for each $\remap(\red a)$, and $f_a(\remap(\red a))$ is polynomial time computable. 
\end{proof}

\end{toappendix}

\section{Conclusion}
In this paper, we've studied safe Pareto improvements (SPIs) achieved through disarmament, commitment to token games, and default-remapping commitment. 
In each setting, we've characterized the computational complexity of finding and optimizing over SPIs. 
By considering forms of commitment which are \textit{ex post} verifiable and thus easier to make credible and enforce, we hope to work towards SPIs that can be more readily applied in practice.

\section*{Acknowledgments}

Nathaniel Sauerberg was supported by a grant from the CLR fund and funding from the Cooperative AI Foundation (CAIF). 
Some of this work was carried out as part of the ML Alignment \& Theory Scholars (MATS) program.
Caspar Oesterheld was supported by an FLI PhD Fellowship.

\bibliography{refs}
\nosectionappendix

\appendix

\begin{toappendix}

\section{Complexity of deciding game and subgame isomorphism}
\label{appendix:game-isomorphism-graph-isomorphism-complete}

We here state and prove some results about deciding whether games are isomorphic. We will use these results for proving some of the complexity results in this paper.

\begin{theorem}\label{thm:GI-completeness}
The following problem is \GI-complete. Given two normal-form games $G$ and $G'$, decide whether $G$ is isomorphic to $G'$. The problem remains \GI-complete if we restrict attention to two-player games. Further, the problem remains \GI-complete if we restrict it to fully reduced games $G$ and $G'$. 
\end{theorem}

    \begin{table*}[]
    \centering
    \begin{tabular}{c|c|c|c|c|c|c|c|}
    \multicolumn{1}{c}{} & \multicolumn{1}{c}{$a_1$} & \multicolumn{1}{c}{$\dots$} & \multicolumn{1}{c}{$a_n$} & \multicolumn{1}{c}{$\hat a_1$} & \multicolumn{1}{c}{$\dots$} & \multicolumn{1}{c}{$\hat a_n$} \\
     \cline{2-7}
    \multirow{1}{*}{$a_1$ } & \multicolumn{3}{c|}{\multirow{3}{*}{$\mathbf u$}}  & \multicolumn{1}{c}{$1, 1$} & \multicolumn{1}{c}{} & $-1,-1$ \\
    $\vdots$ &\multicolumn{3}{c|}{}  & \multicolumn{1}{c}{} & \multicolumn{1}{c}{$\ddots$} &  \\
    $a_n$ &\multicolumn{3}{c|}{}  & \multicolumn{1}{c}{$-1,-1$}& \multicolumn{1}{c}{} & $1,1$ \\
    \cline{2-7}
    \cline{2-4}
    \multirow{1}{*}{$\hat a_1$} & \multicolumn{1}{c}{$1,1$} & \multicolumn{1}{c}{} & $-1,-1$ & \multicolumn{3}{c|}{\multirow{3}{*}{$0,0$}} \\
    $\vdots$ & \multicolumn{1}{c}{} & \multicolumn{1}{c}{$\ddots$} &  & \multicolumn{3}{c|}{} \\
    $\hat a_n$& \multicolumn{1}{c}{$-1, -1$}& \multicolumn{1}{c}{} & $1,1$ & \multicolumn{3}{c|}{}\\
    \cline{2-7}
    \end{tabular}
    \caption{Construction for the hardness part of the proof of \Cref{thm:GI-completeness}.}
    \label{table:construction-for-proof-of-thm:GI-completeness}
    \end{table*}

\begin{proof}
The first two claims (GI-completeness even for two-player games) were proved by \citet{Gabarro_Complexity_Game_isomorphism}.

Since the last claim (about games without dominated actions) considers a narrower problem, all we need to show is that the problem is still \GI-hard. We do this by reducing the 2-player-game isomorphism problem onto the 2-player-game isomorphism problem restricted to games without strictly dominated actions%
.

So let $G=(\{a_1,...,a_n\},\{a_1,...,a_n\},u_1,u_2)$ and $G'$ be games. Let $u_1,u_2$ be normalized to have values in $[0,1]$, with both the minimum $0$ and the maximum $1$ being achieved for both players. (This is almost without loss of generality because renormalizing the utilities doesn't affect whether an isomorphism exists. The only special case that we need to deal with separately is the case where one player's utility function is constant across outcomes. Note that for there to be an isomorphism, the player's utility must be constant in both games. We then normalize the utility to be, say, $\nicefrac{1}{2}$ in both games. Because this case is easy to deal with, we don't treat it in more detail here.) Construct game $\hat G$ as in \Cref{table:construction-for-proof-of-thm:GI-completeness}, and construct $\hat G '$ analogously with action labels $a_1',...,a_n',\hat a_1',...,\hat a_n'$. It is easy to see that $\hat G, \hat G'$ contain no dominated actions. It is left to show that $G,G'$ are isomorphic if and only if $\hat G, \hat G'$ are isomorphic.

Notice first that any isomorphisms between $\hat G$ and $\hat G'$ must induce the identity as the transformation of the utilities. It follows that for each player, the isomorphism must map $\{ a_1,...,a_n\}$ into $\{a_1',...,a_n'\}$ and $\{\hat a_1,...,\hat a_n\}$ into $\{\hat a_1',...,\hat a_n' \}$. It follows immediately that if $\hat\phi_1,\hat\phi_2$ is an isomorphism between $\hat G$ and $\hat G'$, then restricting $\hat\phi_1,\hat\phi_2$ to $\{a_1,...,a_n\}$ yields an isomorphism between $G$ and $G'$.

It is left to show that if there is an isomorphism $(\phi_1,\phi_2)$ from $G$ to $G'$, then there is an isomorphism between $\hat G$ and $\hat G'$. Construct this isomorphism $(\hat \phi_1,\hat \phi_2)$ as follows. First, for the actions shared with $G$, $\hat\phi_1$ and $\hat \phi_2$ simply follow $\phi_1$ and $\phi_2$, i.e., for $i=1,2$ and $k=1,...,n$ set $\hat \phi_i(a_k) = \phi_i(a_k)$. For $\hat a_k$, we let $\hat\phi_1$ follow $\phi_2$. That is, if $\phi_2(a_k)=a_j$, then define $\hat \phi_1(\hat a_k) = \hat a_j$. Define $\hat \phi_2$ analogously. It is easy to show that $(\hat\phi_1,\hat\phi_2)$ thus defined is an isomorphism.
\end{proof}

\begin{proposition}
\label{prop:Pareto-improving-isomorphism-GI-complete}
The problems in \Cref{thm:GI-completeness} remain \GI-complete if instead of the existence of any isomorphism, we query the existence of a \textit{(strictly) Pareto-improving} isomorphism from $G$ to $G'$.
\end{proposition}

\begin{proof}
    \underline{\GI-hardness}: We can reduce the problem of determining the existence of \textit{any} game isomorphism to the problem of finding a (strictly) Pareto-improving isomorphism by adding a large constant to both players' payoff in $G'$.

    \underline{\GI-membership}: To prove this, we reduce from the problem of deciding the existence of a Pareto-improving isomorphism to the existence of any isomorphism. The main insight is that if we have two games $\game$ and $\game'$ and we posit that there is an isomorphism between them, then without knowing the isomorphism, we know how the isomorphism acts on the \textit{utilities}. It maps the lowest utility of one onto the lowest utility of the other, and so on. Thus, we can decide the existence of a (strictly) Pareto-improving isomorphism by first deciding whether this utility mapping of the prospective isomorphism is (strictly) Pareto-improving. If not, we can return ``No''. Otherwise, we return ``Yes'' if and only if the two games are isomorphic. 
\end{proof}

\begin{proposition}\label{prop:1-0-isomorphism-GI-complete}
The regular isomorphism problem in \Cref{thm:GI-completeness} (i.e., not the one constrained to Pareto-improving isomorphisms) remains \GI-complete if instead of the existence of any isomorphism, we query the existence of an isomorphism from $G$ to $G'$ with coefficients $1$ and $0$ for all players.
\end{proposition}

\begin{proof}
Since the problem is narrower than the corresponding problem in \Cref{thm:GI-completeness} (which we already know to be in \GI), all we need to prove is \GI-hardness. We prove this by reducing the general isomorphism problem to the one constrained to coefficients $1$ and $0$.

So take any games $\game$ and $\game'$ with utility functions $u$ and $u'$. Now obtain two new utility functions $\bar u$ and $\bar u'$ by normalizing the utilities to be between $0$ and $1$. If any player's utilities are constant, we simply leave that player's utilities untouched. Call the resulting games $\bar \game$ and $\bar \game'$. It is easy to see that $\game$ and $\game'$ are isomorphic if and only if $\bar \game$ and $\bar \game'$ are isomorphic. Further, it is easy to see that all isomorphisms between $\bar \game$ and $\bar \game'$ have coefficients $1$ and $0$.
\end{proof}

\begin{theorem}\label{thm:subgame-isomorphism-problem}
    The following problem is NP-complete. Given games $G$ and $G'=(A_1',A_2',\mathbf u)$, decide whether there exist $\tilde A_1\subseteq A_1'$ and $\tilde A_2 \subseteq A_2'$ such that $G$ is isomorphic to $(\tilde A_1, \tilde A_2, \mathbf u _{|\tilde A_1 \times \tilde A_2})$. The problem remains NP-complete if we restrict $G$ to have no dominated actions. It also remains NP-complete if we look only for Pareto-improving isomorphisms or isomorphisms with coefficients $1$ and $0$.
\end{theorem}

\begin{proof}
\NP-membership is easy. All hardness results are easy to show by reduction from the subgraph isomorphism problem which is \NP-hard by \Cref{lemma:subgraph-iso-hard}.
\end{proof}

\end{toappendix} 

\end{document}